\documentclass[pra,twocolumn,superscriptaddress,nofootinbib]{revtex4}

\usepackage{amsmath}
\usepackage{amsfonts,amssymb,amsthm,bbm,braket}
\usepackage{mathrsfs}
\usepackage{bm}
\usepackage{graphicx}   
\usepackage{subfigure}  
\usepackage{amsbsy} 
\usepackage[bold]{hhtensor}
\usepackage{tikz} 
\usepackage{mathtools} 
\usepackage{dsfont}
\usepackage{enumitem}
\usepackage{hhline}
\usepackage{array}
\usepackage{soul}
\usepackage{appendix}

\usepackage{thmtools, thm-restate}

\usepackage{algorithm}
\usepackage{algpseudocode}
\algnewcommand\algorithmicinput{\textbf{Input:}}
\algnewcommand\Input{\item[\algorithmicinput]}
\algnewcommand\algorithmicoutput{\textbf{Output:}}
\algnewcommand\Output{\item[\algorithmicoutput]}
\newcommand{\Desc}[2]{\Statex \makebox[2em][l]{#1}#2}
\usepackage{longtable}
\usepackage{multirow}
\usepackage{array}
\newcolumntype{L}{>{$}l<{$}}

\newcolumntype{x}[1]{>{\centering\arraybackslash\hspace{0pt}}p{#1}}

\newcommand{\ketbra}[2]{\ket{#1}\bra{#2}}
\newcommand\av[1]{\left\langle #1 \right\rangle}

\newcommand\id{\mathbbm{1}}

\newcommand\hilb{\mathcal{H}}

\newcommand{\R}{\mathbb R}
\newcommand{\C}{\mathbb C}
\newcommand{\Z}{\mathbb Z}
\renewcommand{\P}{\mathcal{P}}

\newcommand{\X}{\mathbf X}
\newcommand{\Zb}{\mathbf Z}

\newcommand{\rmi}{\mathrm{i}}
\newcommand{\rmx}{\mathrm{x}}
\newcommand{\rmz}{\mathrm{z}}

\newcommand{\rmC}{\mathrm{C}}
\newcommand{\rmP}{\mathrm{P}}

\newcommand{\LC}{\mathrm{LC}}
\newcommand{\sep}{\mathrm{sep}}

\newcommand{\bfv}{\mathbf{v}}
\newcommand{\bfw}{\mathbf{w}}

\newcommand{\bfy}{\mathbf{y}}

\newcommand{\scrM}{\mathscr{M}}

\newcommand{\scrU}{\mathscr{U}}

\newcommand{\T}{\mathcal{T}}

\newcommand{\tr}{\operatorname{tr}}
\newcommand{\diag}{\operatorname{diag}}
\newcommand{\rank}{\operatorname{rank}}
\newcommand{\Span}{\operatorname{Span}}
\newcommand{\supp}{\operatorname{supp}}
\newcommand{\rspan}{\operatorname{rspan}}

\def\<{\langle}  
\def\>{\rangle}  

\newcommand{\caA}{\mathcal{A}}

\newcommand{\caH}{\mathcal{H}}

\newcommand{\Y}{\mathcal{Y}}

\newcommand{\fkB}{\mathfrak{B}}

\newtheorem{corollary}{Corollary}
\newtheorem{theorem}{Theorem}
\newtheorem{lemma}{Lemma}
\newtheorem{proposition}{Proposition}
\newtheorem{conjecture}{Conjecture}

\usepackage[breaklinks=true]{hyperref}
\hypersetup{
  colorlinks   = true, 
  urlcolor     = blue, 
  linkcolor    = blue, 
  citecolor   = magenta 
}

\begin{document}
\renewcommand{\S}{\mathcal{S}}
\renewcommand{\L}{\mathcal{L}}
\setlength{\tabcolsep}{6pt}

\title{Optimal verification of stabilizer states}

\author{Ninnat Dangniam}
\email{ninnatdn@gmail.com}
\affiliation{
Department of Physics and Center for Field Theory and Particle Physics, Fudan University, Shanghai 200433, 
China}
\affiliation{State Key Laboratory of Surface Physics, Fudan University, Shanghai 200433, China}
\author{Yun-Guang Han}
\affiliation{
Department of Physics and Center for Field Theory and Particle Physics, Fudan University, Shanghai 200433, 
China}
\affiliation{State Key Laboratory of Surface Physics, Fudan University, Shanghai 200433, China}
\author{Huangjun Zhu}
\email{zhuhuangjun@fudan.edu.cn}
\affiliation{
Department of Physics and Center for Field Theory and Particle Physics, Fudan University, Shanghai 200433, 
China}
\affiliation{State Key Laboratory of Surface Physics, Fudan University, Shanghai 200433, China}
\affiliation{Institute for Nanoelectronic Devices and Quantum Computing, Fudan University, Shanghai 200433, China}
\affiliation{Collaborative Innovation Center of Advanced Microstructures, Nanjing 210093, China}

\date{\today}

\begin{abstract}
Statistical verification of a quantum state aims to certify whether a given unknown state is close to the target state with confidence. So far, sample-optimal verification protocols based on local measurements have been found only for disparate groups of states: bipartite pure states,  GHZ states, and antisymmetric basis states. 
In this work, we investigate systematically optimal verification of entangled stabilizer states using Pauli measurements.  First, we provide a lower bound on the sample complexity of any verification protocol based on separable measurements, which is independent of the number of qubits and the specific stabilizer state. Then we propose a simple algorithm for constructing optimal protocols based on Pauli measurements. Our calculations suggest that  optimal protocols based on Pauli measurements  can saturate the above bound for all entangled stabilizer states, and this claim is verified explicitly for states up to seven qubits. Similar results are derived when each party can choose only two measurement settings, say $X$ and $Z$. Furthermore, by virtue of the chromatic number, we provide an upper bound for the minimum number of settings required to verify any graph state, which is expected to be tight.  For experimentalists, optimal protocols and protocols with the minimum number of settings are explicitly provided for all equivalent classes of stabilizer states up to seven qubits. For theorists, general results on stabilizer states (including graph states in particular) and related structures derived here may be of independent interest beyond quantum state verification.
\end{abstract}

\maketitle

\tableofcontents

\section{Introduction}

Engineered quantum systems have the potential to efficiently perform tasks that are believed to be exponentially difficult for classical computers such as simulating quantum systems and solving certain computational problems. With the potential comes the challenge of verifying that the quantum devices give the correct results. The standard approach of quantum tomography accomplishes this task by fully characterizing the  unknown quantum system, but with the cost exponential in the system size. However, rarely do we need to completely characterize the quantum system  as we often have a good idea of how our devices work, and we may only need to know if the state produced or the operation performed is close to what we expect. The research effort to address these questions have grown into a mature subfield of quantum certification \cite{eisert2019}.

Statistical verification of a target quantum state $\rho =\ketbra{\Psi}{\Psi}$ \cite{pallister2018,hayashi2006,hayashi2009,ZhuH2019AdvS,ZhuH2019AdvL} is an approach for  certifying  that an unknown state $\sigma$ is ``close" to the target $\rho$ with some confidence. More precisely, the verification scheme accepts a density operator $\sigma$ that is close to the target state with the worst-case fidelity $1-\epsilon$ and confidence $1-\delta$. In other words, the  probability of accepting a ``wrong" state $\sigma$ with  $\av{\Psi|\sigma|\Psi} \le 1-\epsilon$ is at most  $\delta$. For the convenience of practical applications, usually the verification protocols are constructed using local operations and classical communication (LOCC). 
Such verification protocols have been gaining traction in the quantum certification community \cite{ZhuH2019E,ZhuH2019O,wang2019,li2019_bipartite,yu2019,liu2019,li2019_GHZ,li2020_phased_Dicke} because they are easy to implement and  potentially require only a small number of copies of the state. 
However, sample-optimal protocols under LOCC have been found only for bipartite maximally entangled states \cite{hayashi2006,hayashi2009,ZhuH2019O}, two-qubit pure states \cite{wang2019}, $n$-partite GHZ states \cite{li2019_GHZ}, and most recently antisymmetric basis states \cite{li2020_phased_Dicke}.  

Maximally entangled states and GHZ states are subsumed under the ubiquitous class of \emph{stabilizer states}, which can be highly entangled yet efficiently simulatable \cite{gottesmann1997,aaronson2004} and efficiently learnable \cite{rocchetto2018} under Pauli measurements.  Another notable  example of stabilizer states  are graph states \cite{hein2004,hein2006}, which have simple graphical representations that transform nicely under local Clifford unitary transformations. They find applications in secret sharing \cite{markham2008}, error correcting codes \cite{schlingemann2001a,schlingemann2001b}, and cluster states in particular are resource states for universal measurement-based quantum computing \cite{raussendorf2003}.
Stabilizer states and graph states 
can be defined for multiqudit systems with any local dimension; nevertheless, multiqubit stabilizer states are the most prominent because most quantum information processing tasks build on  multiqubit systems.
In this paper we
only consider qubit stabilizer states and graph states unless stated otherwise, but we believe that many results presented here can be generalized to the qudit setting as long as the local dimension is a prime. Efficient verification of stabilizer states have many applications, including but not limited to  blind quantum computing \cite{hayashi2015,takeuchi2019} and  quantum gate verification \cite{zhu2019gate,liu2019gate2019,zeng2019property-testing}.

While one might expect that the determination of an optimal verification strategy to be difficult in general, one could hope for the answer for stabilizer states in view of  their relatively simple structure. Given an $n$-qubit stabilizer state,   Pallister, Montanaro and Linden \cite{pallister2018} showed that the optimal strategy when restricted to the measurements of non-trivial stabilizers (to be introduced below) is to measure all $2^n-1$ of them with equal probabilities, which yields the optimal constant scaling of the number of samples,
\begin{align}
	N \approx \left\lceil\frac{2^n-1}{2^{n-1}} \frac{\ln\delta^{-1}}{\epsilon}\right\rceil \approx \left\lceil\frac{2\ln\delta^{-1}}{\epsilon}\right\rceil.
\end{align}
One could choose to measure only $n$ stabilizer generators at the expense of now a linear scaling \cite{pallister2018}:
\begin{align}
	N \approx \left\lceil\frac{n\ln\delta^{-1}}{\epsilon}\right\rceil.
\end{align}
This trade-off is not inevitable in general. Given a graph state associated with the graph $G$, by virtue of graph coloring, Ref.~\cite{ZhuH2019E} proposed an efficient protocol which requires  $\chi(G)$ measurement settings and $\lceil\chi(G)\epsilon^{-1}\ln\delta^{-1}\rceil$ samples. Here the chromatic number $\chi(G)$ of  $G$ is the smallest number of colors required  so that no two adjacent vertices share the same color. With this protocol, one can verify two-colorable graph states, such as  one- or two-dimensional cluster states, with  $\lceil 2\ln\delta^{-1}/\epsilon \rceil$ tests.

In this paper we study systematically optimal verification of stabilizer states using Pauli measurements. We prove that the spectral gap of any verification operator of an entangled stabilizer state based on separable measurements is upper bounded by $2/3$.
To verify the stabilizer state within infidelity $\epsilon$ and significance level $\delta$, therefore, the number of tests required is bounded from below by
\begin{align}\label{eq:2/3}
N = \biggl\lceil
\frac{1}{\ln[1-2\epsilon/3]}\ln\delta\biggr\rceil\approx \left\lceil\frac{3\ln\delta^{-1}}{2\epsilon}\right\rceil.
\end{align}  
Moreover,
we propose a simple algorithm for constructing optimal verification protocols of stabilizer states and graph states based on nonadaptive Pauli measurements.
An  optimal  protocol for each equivalent class of graph states with respect to local Clifford transformations (LC) and graph isomorphism is presented in Table~\ref{tab:optimal} in the Appendix, and our code is available on \href{https://github.com/ninnat/graph-state-verification}{Github}. These results suggest that for any entangled stabilizer state  the bound in \eqref{eq:2/3}
can be saturated by protocols built on  Pauli measurements.

In addition, we study the problem of optimal verification based on $X$ and $Z$ measurements and the problem of verification with the minimum number of measurement settings. This problem is of interest in many scenarios in which the accessible measurement settings are restricted.
Our study suggests that the maximum spectral gap achievable by $X$ and $Z$ measurements is $1/2$. For the  ring cluster state we prove this result rigorously by constructing an explicit optimal verification protocol.  We also prove that three settings based on Pauli measurements (or $X$ and $Z$ measurements) are both necessary and sufficient for verifying the odd ring cluster state with at least five qubits.

In the course of study, we introduce the concepts of admissible Pauli measurements and admissible test projectors for general stabilizer states and clarify their basic properties, which are of interest to quantum state verification in general. Meanwhile, we introduce several graph invariants that are tied to the verification of graph states and clarify their connections with the chromatic number. In addition to their significance to the current study, these results provide additional insights on stabilizer states and graph states themselves and are expected to find applications in various other related problems.

The rest of this paper is organized as follows. 
First, we present a brief introduction to quantum state verification in Sec.~\ref{sec:QSV} and preliminary results on  the stabilizer formalism in Sec.~\ref{sec:stabilizer}. In Sec.~\ref{sec:test}  we study canonical test projectors and admissible  test projectors for stabilizer states and graph states and clarify their properties.  In Sec.~\ref{sec:spectral-gap-optimization} we derive an upper bound for the spectral gap of verification operators  based on separable measurements. 
Moreover, we propose a simple algorithm for constructing optimal verification protocols based on  Pauli measurements and provide an explicit optimal protocol for each connected graph state up to seven qubits.
In Sec.~\ref{sec:XZ-verification}
we discuss optimal verification of graph states based on $X$ and $Z$ measurements. In Sec.~\ref{sec:minimal-verification} we consider the verification of graph states with the  minimum number of settings. Sec.~\ref{sec:summary} summarizes this paper. To streamline the presentation, the proofs of several technical results are relegated to the Appendices, which also contain Tables~\ref{tab:optimal} and \ref{tab:minimal-settings}.

\section{Statistical verification}\label{sec:QSV}
\subsection{The basic framework}
Let us formally introduce the framework of statistical verification of quantum states. Suppose we want to prepare the target state $\rho=|\Psi\>\<\Psi|$, but actually obtain the sequence of states $\sigma_1,\cdots,\sigma_N$ in $N$ runs. Our task is to determine whether these states are sufficiently close to the target state on average (with respect to the fidelity, say).
Following \cite{pallister2018,ZhuH2019AdvS,ZhuH2019AdvL}, we perform a local  measurement with binary outcomes $\{E_j,\id - E_j\}$, labeled as ``pass" and ``fail" respectively, on each state $\sigma_k$ for $k=1,\dots,N$ with some probability $p_j$. Each operator $E_j$ is called a \emph{test operator}. Here we  demand that the target state $\rho$ can  pass the test with certainty, which means $E_j\rho=\rho$. 
The sequence of states passes the verification procedure iff every outcome is ``pass".  The efficiency of the above verification procedure is  determined by  the \emph{verification operator}
\begin{align}
	\Omega = \sum_{j=1}^m p_j E_j,
\end{align}
where $m$ is the total number of measurement settings.

If the fidelity $\<\Psi|\sigma_k|\Psi\>$ is upper bounded by $ 1-\epsilon$, then 
the maximal average probability that $\sigma_k$ can pass each test  is \cite{pallister2018,ZhuH2019AdvL}
\begin{equation}\label{eq:PassingProb}
\max_{\<\Psi|\sigma|\Psi\>\leq 1-\epsilon }\tr(\Omega \sigma)=1- [1-\beta(\Omega)]\epsilon=1- \nu(\Omega)\epsilon.
\end{equation} 
Here  $\beta(\Omega)$ is the second largest eigenvalue of the verification operator $\Omega$, and  $\nu(\Omega):=1-\beta(\Omega)$ is the spectral gap from the maximum eigenvalue.
Suppose the states $\sigma_1, \sigma_2, \ldots, \sigma_N$ produced  are independent of each other. Then these states can pass $N$ tests with probability at most
\begin{align}\label{eq:PassProb}
\prod_{j=1}^N \tr(\Omega\sigma_j)\leq \prod_{j=1}^N [1-\nu(\Omega)\epsilon_j]\leq [1-\nu(\Omega)\bar{\epsilon}]^N, 
\end{align}
where  $\bar{\epsilon}=\sum_j \epsilon_j/N$  with $\epsilon_j=1-\<\Psi|\sigma_j |\Psi\>$ 
is the average infidelity \cite{ZhuH2019AdvS,ZhuH2019AdvL}.
If  $N$ tests are passed, then we can  ensure the condition
$\bar{\epsilon}<\epsilon$ with significance level $\delta= [1-\nu(\Omega)\epsilon]^N$. To verify these states within infidelity $\epsilon$ and significance level $\delta$, the number of tests required is \cite{pallister2018,ZhuH2019AdvS,ZhuH2019AdvL}
\begin{equation}\label{eq:NumTest}
N(\epsilon,\delta,\Omega)=\biggl\lceil
\frac{1}{\ln[1-\nu(\Omega)\epsilon]}\ln\delta\biggr\rceil
\leq 
\biggl\lceil
\frac{1}{\nu(\Omega)\epsilon}\ln\frac{1}{\delta}\biggr\rceil.
\end{equation}
If there is no restriction on the measurements, the optimal performance is achieved by performing the projective measurement onto the target state $\ketbra{\Psi}{\Psi}$ itself, which yields $\nu(\Omega)=1$ and  $N=\lceil\ln\delta/\ln(1-\epsilon)\rceil\leq \lceil\ln\delta^{-1}/\epsilon\rceil$ as the ultimate efficiency limit allowed by quantum theory.\footnote{A related certification framework by Kalev and Kyrillidis \cite{kalev2019} for stabilizer states is, in a sense, opposite to ours. In our framework, we are given the worst case fidelity and are asked to find an optimal measurement, whereas in their work we are given a (stabilizer) measurement and are asked to bound the worst case fidelity $1-\epsilon$ to the desired stabilizer state within some radius $r$ (their ``$\epsilon$").}

A set of test operators $\{E_j\}_{j=1}^m$ for $|\Psi\>$ is \emph{minimal} if any proper subset of $\{E_j\}_{j=1}^m$ cannot verify $|\Psi\>$ reliably because the common pass eigenspace of operators in the subset has dimension larger than one. A minimal set of test operators has the following properties.
\begin{proposition}\label{pro:SpectralGapMinSettings}
	Suppose $\Omega=\sum_j p_j E_j$ is a verification operator based on a minimal set of $m$ test operators. Then $\nu(\Omega)\leq 1/m$. If the inequality is saturated then $p_j=1/m$ for all $j$. 
\end{proposition}

\begin{proof}
	By assumption, for each  $k\in \{1,2,\ldots, m\}$, there exists a pure state $|\Psi_k\>$ that is orthogonal to the target state $|\Psi\>$ and belongs to the pass eigenspace of $E_j$ for all $j\neq k$, that is, $E_j|\Psi_k\>=|\Psi_k\>$. Therefore,
	\begin{align}
	\beta(\Omega)&\geq \<\Psi_k|\Omega|\Psi_k\>\geq \sum_{j\neq k} p_j \<\Psi_k|E_j|\Psi_k\>=\sum_{j\neq k} p_j\nonumber\\
	&=1-p_k\quad \forall k, 
	\end{align}
	which implies that 
	\begin{equation}
	\nu(\Omega)\leq \min_k p_k\leq 1/m. 
	\end{equation}
	Here the second inequality is saturated iff $p_k=1/m$ for all $k$. 
\end{proof}

\subsection{Verification of a tensor product}
Suppose the target state $|\Psi\>$ is a tensor product of the form $|\Psi\>=\bigotimes_{j=1}^J |\Psi_j\>$, where $J\geq 2$ and each tensor factor $|\Psi_j\>$ may be either separable or entangled. It is instructive to clarify the relation between the verification operators  of $|\Psi\>$ and that of each tensor factor.

Given a verification operator $\Omega$  for $|\Psi\>$, the \emph{reduced verification operator} of $\Omega$ for the tensor factor $|\Psi_j\>$ is defined as 
\begin{equation}\label{eq:ReducedVeriO}
\Omega_j:=\<\overline{\Psi}_j|\Omega|\overline{\Psi}_j\>,
\end{equation}
where $|\overline{\Psi}_j\>:=\bigotimes_{j'\neq j}|\Psi_{j'}\>$. Note that $\Omega_j|\Psi_j\>=|\Psi_j\>$, so
$\Omega_j$ is indeed a verification operator for $|\Psi_j\>$. If $\Omega$ is separable, then each $\Omega_j$ is also separable. Reduced test operators can be defined in a similar way.
\begin{proposition}\label{pro:ReducedVeriO}
	Suppose $\Omega$ is a verification operator for $|\Psi\>=\bigotimes_{j=1}^J |\Psi_j\>$, and $\Omega_j$ for $j=1,2,\ldots, J$ are  reduced verification operators of $\Omega$. Then 
	\begin{align}\label{eq:ReducedVeriGap}
	\beta(\Omega)\geq\max_{1\leq j\leq J} \beta(\Omega_j),\quad \nu(\Omega)\leq \min_{1\leq j\leq J} \nu(\Omega_j). 
	\end{align}
\end{proposition}
\begin{proof}
	\begin{align}
	\beta(\Omega_j)&=\max_{|\Phi_j\> :\<\Psi_j|\Phi_j\>=0} \<\Phi_j|\Omega_j|\Phi_j\>\nonumber\\
	&= \max_{|\Phi_j\> :\<\Psi_j|\Phi_j\>=0} (\<\Phi_j|\otimes \<\overline{\Psi}_j|)\Omega(|\Phi_j\>\otimes |\overline{\Psi}_j\>)\nonumber\\
	&\leq \max_{|\Phi\> :\<\Psi|\Phi\>=0}\<\Phi|\Omega|\Phi\>=\beta(\Omega),
	\end{align}
	which implies \eqref{eq:ReducedVeriGap}. 
\end{proof}

Conversely,  suppose $\Omega_j$ are verification operators for $|\Psi_j\>$ with spectral gap $\nu(\Omega_j)$ for $j=1,2,\ldots, J$. Let $\Omega=\bigotimes_{j=1}^J  \Omega_j$; then $\Omega$ is a verification operator for $|\Psi\>$, and $\Omega_j$ are reduced verification operators of $\Omega$ by the definition in \eqref{eq:ReducedVeriO}. Straightforward calculation shows that the spectral gap  of $\Omega$ reads
\begin{equation}
\nu(\Omega)=\min_{1\leq j\leq J}\nu(\Omega_j),
\end{equation}
which saturates the upper bound in \eqref{eq:ReducedVeriGap}. In addition, if each $\Omega_j$ can be realized by LOCC (Pauli measurements), then so can $\Omega$. On the other hand, the number of distinct test operators (measurement settings) required to realize $\Omega$ (naively as suggested by the definition) increases exponentially with the number $J$ of tensor factors. It is of practical interest to reduce this number.

Suppose $\Omega_j$ can be realized by the set of test operators $\{E_k^{(j)}\}_{k=1}^{m_j}$, that is, $\Omega_j=\sum_{k=1}^{m_j}p_k^{(j)} E_k^{(j)}$, where $(p_k^{(j)})_{k=1}^{m_j}$ is a probability vector. In addition, $|\Psi_j\>$ is the unique common eigenstate of $E_k^{(j)}$ with eigenvalue 1. Let $m=\max_j m_j$; then $|\Psi\>$ can be reliably  verified by the following test operators
\begin{equation}\label{eq:TestOperatorJoint}
E_k:=\bigotimes_{j=1}^J E_k^{(j)},\quad k=1,2,\ldots, m,
\end{equation}
where $E_k^{(j)}=\openone $ if $m_j<k\leq m$. To verify this claim, first note that $E_k|\Psi\>=|\Psi\>$ for $k=1,2,\ldots,m$, so each $E_k$ is a test operator for $|\Psi\>$.  Suppose $|\Phi\>$  is a common eigenstate of all $E_k$ with eigenvalue 1, that is, $\<\Phi| E_k|\Phi\>=1$. Let $\rho_j=\tr_{\bar{j}} (|\Phi\>\<\Phi|)$, where $\tr_{\bar{j}}$ denotes the partial trace over all tensor factors except for the $j$th factor. Then we have
\begin{align}
1\geq \tr(\rho_j E_k^{(j)})=\<\Phi| E_k^{(j)}\otimes \openone|\Phi\>\geq \<\Phi| E_k|\Phi\>=1
\end{align}
for all $j,k$. 
This equation implies that $\tr(\rho_j E_k^{(j)})=1$, so each $\rho_j$ is supported in the eigenspace of $E_k^{(j)}$ with eigenvalue 1 for all $k$. It follows that $\rho_j=|\Psi_j\>\<\Psi_j|$ and $|\Phi\>\<\Phi|=|\Psi\>\<\Psi|$, so  $|\Psi\>$ is the unique common eigenstate of all $E_k$ with eigenvalue~1 and it can be reliably  verified by the test operators in \eqref{eq:TestOperatorJoint}. 

Let $(q_k)_{k=1}^m$ be any probability vector with $q_k>0$ for all $k$ and $\Omega=\sum_{k=1}^m q_k E_k$; then $\Omega$ is a verification operator for $|\Psi\>$ with $\nu(\Omega)>0$ according to the above discussion. In addition, the reduced verification operator of $\Omega$ for tensor factor $|\Psi_j\>$ reads
\begin{equation}
\Omega_j=\<\overline{\Psi}_j|\Omega|\overline{\Psi}_j\>=\sum_{k=1}^m q_k E_k^{(j)}. 
\end{equation}
According to Proposition~\ref{pro:ReducedVeriO} we have
\begin{equation}
\nu(\Omega)\leq \nu(\Omega_j)=\nu\biggl(\sum_{k=1}^m q_k E_k^{(j)}\biggr)\leq \max_{(q_k')_k}\nu\biggl(\sum_{k=1}^m q_k' E_k^{(j)}\biggr),
\end{equation}
where the maximization is taken over all probability vectors with $m$ components. The right-hand side coincides with the maximum spectral gap achievable by any verification operator of $|\Psi_j\>$ that is based on the set of test operators $\{E_k^{(j)}\}_{k=1}^{m_j}$. Note that $q_k'=0$ for $m_j<k\leq m$ when the maximum spectral gap is attained given that $E_k^{(j)}=\openone $ for $m_j<k\leq m$.

\section{Stabilizer formalism}\label{sec:stabilizer}
\subsection{Pauli group}
Let $\hilb = (\C^2)^{\otimes n}$ be the Hilbert space of $n$ qubits. The Pauli group for one qubit is generated by the following three matrices: 
\begin{align}
X &=\begin{pmatrix}
0 & 1\\
1 & 0
\end{pmatrix},&
Y & =\begin{pmatrix}
0 & -\rmi\\
\rmi & 0
\end{pmatrix},&
Z & =\begin{pmatrix}
1 & 0\\
0 & -1
\end{pmatrix}. 
\end{align} 
The $n$-fold tensor products of Pauli matrices and the identity $\{\id,X,Y,Z\}^{\otimes n}$ form an orthogonal basis for the space $\fkB(\caH)$ of linear operators on $\hilb$. Together with the phase factors $\{\pm1,\pm \rmi\}$, these operators  generate  the \emph{Pauli group} $\P_n$, which has  order $4^{n+1}$. Two elements  of the Pauli group either commute  or anticommute.

Up to phase factors, $n$-qubit Pauli operators
can be labeled by vectors in the binary symplectic space $\Z_2^{2n}$ endowed with the symplectic form 
\begin{align}\label{eq:SymplecticForm}
[\mu,\nu] &:= \mu^TJ\nu,\quad J=\begin{pmatrix}
0 & -\id_n \\
\id_n & 0
\end{pmatrix}.
\end{align}
Let $\mu^\rmx$ ($\mu^\rmz$) be the vector composed of the first (last) $n$ elements of $\mu$; then $\mu = (\mu^\rmx;\mu^\rmz)$, where the semicolon denotes the vertical concatenation. In addition, we have $[\mu,\nu] =  \mu^\rmz \cdot \nu^\rmx + \mu^\rmx \cdot \nu^\rmz$. (Addition and subtraction are the same in arithmetic modulo 2.) The \emph{Weyl representation} \cite{gross2006} of each vector  $\mu\in \Z_2^{2n}$ yields  a Pauli operator
\begin{align}\label{eq:WeylRep}
g(\mu) = \rmi^{\mu^\rmx \cdot \mu^\rmz} \X^{\mu^\rmx} \Zb^{\mu^\rmz},
\end{align}
where $\X^{\mu^\rmx} = X_1^{\mu_1^\rmx} \otimes \cdots \otimes X_n^{\mu_n^\rmx}$ and similarly for $\Zb^{\mu^\rmz}$. Each Pauli operator in $\P_n$ is equal to $\rmi^k g(\mu)$ for $k=0,1,2,3$ and $\mu\in \Z_2^{2n}$. For example, we have $X = g((1;0))$, $Z = g((0;1))$, and $Y = g((1;1))$.
By the following identity 
\begin{align}
g(\mu)g(\nu) = \rmi^{[\mu,\nu]} g(\mu+\nu) = \rmi^{2[\mu,\nu]}g(\nu)g(\mu),
\end{align}
 $g(\mu)$ and $g(\nu)$ commute iff $[\mu,\nu] = 0$.
 
The symplectic complement  of a subspace $W$ in $\Z_2^{2n}$ is defined as 
\begin{align}
W^{\perp} = \{\mu \in \Z_2^{2n} \,|\, [\mu,\nu] = 0, \forall \nu \in W \}.
\end{align}
The  subspace $W$ is \emph{isotropic} if $W \subset W^{\perp}$, in which case $[\mu,\nu]=0$  for all $\mu,\nu\in W$. Hence all Pauli operators associated with vectors in $W$ commute with each other.
The maximal dimension of any isotropic subspace is $n$, and such a maximal isotropic subspace satisfies the equality $W=W^\perp$ and is called \emph{Lagrangian}. 
Each isotropic subspace $W$ of dimension $k$ is determined by a $2n\times k$ basis matrix over $\Z_2^{2n}$ whose columns form a basis of $W$. Conversely, a $2n\times k$ matrix $M$ is a basis matrix for  an isotropic subspace  iff the following condition holds
\begin{equation}
M^TJM=0_{k\times k}.
\end{equation}
Two Lagrangian subspaces $W$ and $W'$ of $\Z_2^{2n}$ are \emph{complementary} if their intersection is trivial (consists of the zero vector only), in which case  $\Span(W\cup W')=\Z_2^{2n}$. 

The Clifford group is the normalizer of the Pauli group $\P_n$. Up to phase factors, it is generated by phase gates, Hadamard gates for individual qubits and controlled-not gates for all pairs of qubits. Its quotient over the Pauli group is isomorphic to the symplectic group with respect to the symplectic form in \eqref{eq:SymplecticForm}.

\subsection{Stabilizer codes} 
A subgroup $\S$ of $\P_n$ is a \emph{stabilizer group} if $\S$ is commutative and does not contain $-\id$. Since $\S$ cannot contain a Pauli operator with phases $\pm \rmi$ (otherwise $-\id \in \S$), every element except the identity has order 2. Thus, $\S$  is isomorphic to an elementary abelian group $\Z_2^k$ of order $2^k$, where $k\leq n$ is the number of minimal generators.  Suppose that the stabilizer group $\S$ is generated by the $k$ generators $S_1, S_2, \ldots, S_k$; then the elements of $\S$ can be labeled by vectors in $\Z_2^k$ as follows,
\begin{equation}\label{eq:Sy}
	S^{\bfy}=\prod_{j=1}^k   S_j^{y_j}, \quad \bfy\in \Z_2^k. 
\end{equation}
The \emph{stabilizer code} $\caH_\S$ of $\S$ is the common eigenspace of eigenvalue 1 of all Pauli operators in $\S$, which has dimension $2^{n-k}$. Alternatively, it is also defined as the common eigenspace of eigenvalue 1 of the $k$ generators $S_1, S_2, \ldots, S_k$. The projector onto the code space reads
\begin{equation}\label{eq:StabCode}
\Pi_\S=\frac{1}{|\S|}\sum_{S\in \S}S=\prod_{j=1}^k \frac{1+S_j}{2}.
\end{equation}
Conversely, $\S$ happens to be the group of all Pauli operators in $\P_n$ that stabilize the stabilizer code $\caH_\S$. So there is a one-to-one correspondence between stabilizer groups and stabilizer codes. 
To later establish the relation between stabilizer groups and isotropic subspaces, we introduce the \emph{signed stabilizer group} of the  stabilizer code $\caH_\S$ to be the union 
\begin{equation}\label{def:signed-stab}
\bar{\S}:=\S\cup (-\S),
\end{equation}
where  $-\S:=\{-S|S\in \S\}$. 

Given the stabilizer group  $\S$  with generators  $S_1, S_2, \ldots, S_k$,
for each $\bfw\in \Z_2^k$, define $\S_\bfw$ as the  group generated by $(-1)^{w_j}S_j$ for $j=1,2,\ldots, k$; then $\S_\bfw$ is also a  stabilizer group. The stabilizer code $\caH_{\S_\bfw}$ of $\S_\bfw$
is the common eigenstate of  $S_1, S_2, \ldots, S_k$ with eigenvalue $(-1)^{w_j}$ for $j=1,2,\ldots, k$ and is also denoted by $\caH_{\S,\bfw}$. The projector onto the stabilizer code reads
\begin{equation}\label{eq:StabCode2}
\Pi_{\S,\bfw}=\prod_{j=1}^n \frac{1+(-1)^{w_j}S_j}{2}=\sum_{\bfy\in \Z_2^n} \chi_\bfw(S^\bfy) S^\bfy, 
\end{equation}
where
\begin{equation}
\chi_\bfw(S^\bfy)=\chi_\bfw(\bfy)=(-1)^{\bfw\cdot \bfy}
\end{equation}
can be understood as a character on $\S$ or $\Z_2^k$ \cite{serre1977}. Note that all stabilizer codes $\caH_{\S,\bfw}$ for $\bfw\in \Z_2^k$ share the same signed stabilizer group, that is, $\bar{\S}_\bfw=\bar{\S}$.

According to the Weyl representation  in \eqref{eq:WeylRep}, each element in $\S$ is equal to $g(\mu)$ or $-g(\mu)$ for  $\mu\in \Z_2^{2n}$. In this way, $\S$ is associated with an isotropic subspace $W\subset \Z_2^{2n}$  of dimension $k$, and there is a one-to-one correspondence between elements in $\S$ and vectors in $W$. Suppose 
the $k$ generators $S_1, S_2, \ldots, S_k$ of $\S$ correspond to the $k$ symplectic vectors
$\mu_1,\mu_2,\ldots, \mu_k$,  which form a basis in $W$.  Then $S^\bfy$   corresponds to the vector $\sum_j y_j \mu_j=M\bfy$ for each  $\bfy\in\Z_2^k$, where $M:=(\mu_1,\mu_2,\dots,\mu_k)$ is  a basis matrix for $W$. Note that all the stabilizer groups $\S_\bfw$ for $\bfw\in \Z_2^k$
are associated with the same isotropic subspace 
according to the above correspondence, and this correspondence extends to the signed stabilizer group $\bar{\S}$. Conversely, given an isotropic subspace $W$ of dimension $k$, $2^k$ stabilizer groups can be constructed as follows. Let $\{\mu_j\}_{j=1}^k$ be any basis for $W$; for each vector  $a$ in $\Z_2^k$,  a stabilizer group can be constructed from the  $k$ generators $(-1)^{a_j}g(\mu_j)$ for $j=1,2,\ldots, k$. All these stabilizer groups extend to a common signed stabilizer group. 
In this way, there is a one-to-one correspondence between signed stabilizer groups and isotropic subspaces.

Suppose $\S$ and $\S'$ are two $n$-qubit stabilizer groups of orders $2^k$ and $2^{k'}$, respectively; let  $\Pi_\S$ and $\Pi_{\S'}$ be the projectors onto the corresponding stabilizer codes. 
Then the overlap between $\Pi_\S$ and $\Pi_{\S'}$ reads
\begin{align}\label{eq:OverLap}
\tr(\Pi_{\S}\Pi_{\S'})&=\frac{1}{|\S|\cdot|\S'|}\sum_{S\in \S, S'\in \S'}
\tr(S S')\nonumber\\
&=\frac{2^n}{|\S|\cdot|\S'|}(|\S\cap \S'|-|(-\S)\cap \S'|)\nonumber\\
&=\begin{cases}
\frac{2^n|\bar{\S}\cap\S'|}{|\S|\cdot|\S'|} & \S\cap\S'=\bar{\S}\cap\S',\\
0 & \mathrm{otherwise}.
\end{cases}
\end{align}
Note that $\S\cap \S'$ is a subgroup of $\bar{\S}\cap \S'$ of index 2 if  $\S\cap \S'\neq \bar{\S}\cap \S'$. Equation \eqref{eq:OverLap} implies the following equation
\begin{align}\label{eq:OverLap2}
\tr(\Pi_{\S,\bfw}\Pi_{\S',\bfw'})
&=\begin{cases}
\frac{2^n|\bar{\S}\cap\S'|}{|\S|\cdot|\S'|} & \S_\bfw\cap\S'_{\bfw'}=\bar{\S}\cap\S',\\
0 & \mathrm{otherwise}.
\end{cases}
\end{align}
for all $\bfw\in \Z_2^k$ given that $\bar{\S}_\bfw=\bar{\S}$ and $\bar{\S}\cap\S'_{\bfw'}=\bar{\S}\cap\S'$. In addition, we have 
\begin{align}\label{eq:Overlap3}
\sum_\bfw\tr(\Pi_{\S,\bfw}\Pi_{\S'})=\tr(\Pi_{\S'})=\frac{2^n}{|\S'|}
\end{align}
thanks to the equality $\sum_\bfw \Pi_{\S,\bfw}=\id$. 
So the number of vectors $\bfw\in \Z_2^n$ at which $\tr(\Pi_{\S,\bfw}\Pi_{\S'})\neq 0$
 is equal to   $|\S|/|\bar{\S}\cap\S'|$.

\begin{lemma}\label{lem:StabIntersection}
	Suppose $\S_j, \T_j$ are stabilizer groups on $\caH_j$ for $j=1,2,\ldots, J$; let  $\S=\S_1\times \S_2\times \cdots \times \S_J$ and $\T=\T_1\times \T_2\times \cdots \times \T_J$ be stabilizer groups on $\caH_1\otimes \caH_2\otimes \cdots \otimes\caH_J$ and let $\bar{\T}$ be the signed stabilizer group associated with $\T$. Then 
	\begin{align}
	\S\cap \bar{\T}&=\L_1 \times \L_2\times\cdots \times \L_J, \label{eq:StabIntersection}
	\end{align}
	where  $\L_j=\S_j\cap \bar{\T}_j$ with $\bar{\T}_j$ being the signed stabilizer groups associated with $\bar{T}_j$. 
\end{lemma}
\begin{proof}
	To simplify the notation, here we prove \eqref{eq:StabIntersection} in the case $J=2$; the general case can be proved in a similar way. Any  $S\in \S$ has the form $S=S_1\otimes S_2$ with $S_1\in \S_1$ and $S_2\in \S_2$. If in addition $S\in \bar{\T}$, then $S_1\in \bar{\T}_1$ and $S_2\in \bar{\T}_2$. Therefore, $S_1\in \L_1$ and $S_2\in \L_2$, so that $S\in \L_1\times \L_2$, which implies that $\S\cap \bar{\T}\subseteq\L_1 \times \L_2$. 
	
	Conversely, any  $S\in \L_1\times \L_2$ has the form  $S=S_1\otimes S_2$ with $S_1\in \S_1\cap \bar{\T}_1$ and $S_2\in \S_2\cap\bar{\T}_2$, which implies that $S\in \S$ and $S\in \bar{T}$. Therefore,  $\L_1 \times \L_2\subseteq \S\cap \bar{\T}$, which confirms \eqref{eq:StabIntersection} in view of the opposite   inclusion relation  derived above. 
\end{proof}

\subsection{\label{sec:StabState}Stabilizer states} 
When the stabilizer group $\S$ is maximal, that is, $|\S|=2^n$, the stabilizer code $\caH_\S$  has dimension 1 and is represented by a normalized state called a \emph{stabilizer state} and denoted by $|\S\>$.
Note that  $|\S\>$ is uniquely determined by $\S$ up to an overall phase factor. According to \eqref{eq:StabCode}, the projector onto $|\S\>$ reads
\begin{equation}\label{eq:StabState}
|\S\>\<\S|=\Pi_\S=\frac{1}{2^n}\sum_{S\in \S} S=\prod_{j=1}^n \frac{1+S_j}{2},
\end{equation}
where  $S_1, S_2, \ldots, S_n$ are a set of generators of $\S$. For each $\bfw\in \Z_2^n$, define $\S_\bfw$ as the  group generated by $(-1)^{w_j}S_j$ for $j=1,2,\ldots, n$; then $\S_\bfw$ is also a maximal stabilizer group. In addition, the associated stabilizer state  $|\S_\bfw\>$
is the common eigenstate of  $S_1, S_2, \ldots, S_n$ with eigenvalue $(-1)^{w_j}$ for $j=1,2,\ldots, n$. Let  $S^{\bfy}=\prod_{j=1}^n S_j^{y_j}$ for  $\bfy\in \Z_2^n$ as in \eqref{eq:Sy} with $k=n$; then the projector onto $|\S_\bfw\>$ reads
\begin{equation}\label{eq:StabState2}
|\S_\bfw\>\<\S_\bfw|=\Pi_{\S,\bfw}=\!\prod_{j=1}^n \frac{1+(-1)^{w_j}S_j}{2}=\!\sum_{\bfy\in \Z_2^n} (-1)^{\bfw\cdot \bfy} S^\bfy, 
\end{equation}
which reduces to \eqref{eq:StabState} when $\bfw=00\cdots 0$. The set of stabilizer states $|\S_\bfw\>$ for $\bfw\in \Z_2^n$ forms an orthonormal basis in $\hilb$, known as a \emph{stabilizer basis}. Stabilizer bases are in one-to-one correspondence with Lagrangian subspaces in $\Z_2^{2n}$. Based on this observation, one can determine  the total number of $n$-qubit stabilizer states, with the result  \cite{gross2006}
\begin{equation}
2^n \prod_{j=1}^n (2^j + 1) \ge 2^{n(n+3)/2},
\end{equation}
which is exponential in the number $n$ of qubits.
 
Note that all  stabilizer states in a stabilizer basis share the same signed stabilizer group as defined in  \eqref{def:signed-stab}, that is, 
$\bar{\S}_\bfw=\bar{\S}$ for all $\bfw\in \Z_2^n$. In this way, there is a one-to-one correspondence between signed stabilizer groups and stabilizer bases (and Lagrangian subspaces).  

Suppose  $\S$ and $\S'$ are two $n$-qubit maximal stabilizer groups. 
Then the fidelity between $|\S\>$ and $|\S'\>$  reads
\begin{align}\label{eq:StabFidelity}
|\<\S|\S'\>|^2
&=\begin{cases}
2^{-n}|\bar{\S}\cap\S'|\ & \S\cap\S'=\bar{\S}\cap\S',\\
0 & \mathrm{otherwise},
\end{cases}
\end{align}
according to 
\eqref{eq:OverLap}, 
where $\bar{\S}=\S\cup (-\S)$ is the signed stabilizer group of $|\S\>$. In addition,
\begin{align}\label{eq:StabFidelity2}
|\<\S_\bfw|\S'_{\bfw'}\>|^2=\begin{cases}
2^{-n}|\bar{\S}\cap\S'| & \S_\bfw\cap\S'_{\bfw'}=\bar{\S}\cap\S',\\
0 & \mathrm{otherwise}
\end{cases}
\end{align}
for all $\bfw,\bfw'\in \Z_2^n$ according to 
\eqref{eq:OverLap2}.
The number of vectors $\bfw\in \Z_2^n$ that satisfy 
the equality $|\<\S_\bfw|\S'\>|^2=
2^{-n}|\bar{\S}\cap\S'|$ is equal to   $2^{n}/|\bar{\S}\cap\S'|$. Two maximal stabilizer groups  $\S$ and $\S'$ are \emph{complementary} if $|\bar{\S}\cap\S'|=1$ or, equivalently, $|\S\cap\bar{\S'}|=1$, which is the case iff the Lagrangian subspaces associated with $\S$ and $\S'$, respectively, are complementary. In addition, $\S$ and $\S'$ are complementary iff the stabilizer bases associated with $\S$ and $\S'$ are mutually unbiased \cite{durt2010}, that is, 
\begin{equation}
|\<\S_\bfw|\S'_{\bfw'}\>|^2=2^{-n}\quad \forall \bfw,\bfw'\in \Z_2^n. 
\end{equation}

The following lemma is useful to computing the fidelities between stabilizer states; see Appendix \ref{app:lemma2} for a proof. 
\begin{lemma}\label{lem:StabFidelity}
Suppose $\S$ and $\S'$ are two maximal stabilizer groups with basis matrix $M$ and $M'$, respectively. Then 
\begin{align}
&|\bar{\S}\cap\S'|=|\ker(M^T JM')|=2^{n-\rank(M^T JM')},\nonumber\\
&=2^{n-\rank(M_\rmz^TM_\rmx'+M_\rmx^TM_\rmz')}, \label{eq:StabIntersect}\\
&\max_{\bfw\in \Z_2^n}|\<\S_\bfw|\S'_{\bfw'}\>|^2=
2^{-n}|\ker(M^T JM')|\nonumber\\
&=2^{-\rank(M^T JM')}=2^{-\rank(M_\rmz^TM_\rmx'+M_\rmx^TM_\rmz')}, \label{eq:StabFidelity3}
\end{align}
where $M_\rmx$ ($M_\rmz$) denotes  the submatrix of $M$  composed of the first (last) $n$ rows, and  $M_\rmx'$ ($M_\rmz'$) is defined in a similar way. 
\end{lemma}

\subsection{\label{sec:GraphState}Graph states}
Before introducing graph states, it is helpful to briefly review basic concepts related to graphs. A graph $G = (V,E)$ is defined by a vertex set  $V$ and an  edge set $E$ in which each element of $E$ is a two-element subset of $V$.  Without loss of generality, the vertex set can be chosen to be $V=\{1,2,\ldots,n\}$. Two vertices $i,j\in V$ are adjacent if $\{i,j\}\in E$. The neighbor of a vertex $j$ is composed of all vertices that are adjacent to $j$.
The adjacency relation is characterized by  the \emph{adjacency matrix} $A$, which is an $n\times n$ symmetric matrix with $A_{i,j}=1$ if $i$ and $j$ are adjacent and $A_{i,j}=0$ otherwise. A subset $B$ of $V$ is an \emph{independent set} if every two vertices in $B$ are not adjacent. The \emph{independence number} $\alpha(G)$ of $G$ is the maximum cardinality of independent sets of $G$. A coloring of $G$ is an assignment of colors to the vertices such that every two adjacent vertices receive different colors. The \emph{chromatic number} $\chi(G)$ of $G$ is the minimum number of colors required to color $G$. The graph $G$ is two colorable if $G$ can be colored using two distinct colors, that is, $\chi(G)\leq 2$.  

Denote by $|+\>=(|0\>+|1\>)/\sqrt{2}$  the eigenstate of $X$ with eigenvalue 1.
To each graph $G=(V,E)$ with $n$ vertices, a graph state $|G\>$ of $n$-qubits (corresponding to the $n$ vertices)
can be constructed from the product state $\ket{+}^{\otimes n}$ by applying the controlled-$Z$ gate 
\begin{equation}
\mathrm{CZ} = \ketbra{0}{0}\otimes\id + \ketbra{1}{1}\otimes Z
\end{equation}
to every pair of qubits that are adjacent. In other words, 
\begin{equation}
|G\>=\prod_{\{i,j\}\in E} \mathrm{CZ}_{i,j} \ket{+}^{\otimes n},
\end{equation}
where $\mathrm{CZ}_{i,j}$ denotes the CZ gate acting on the adjacent qubits $i$ and $j$. 
Note that all these CZ gates commute with each other, so their order in the product does not matter. To give some examples, a path or linear graph yields a linear cluster state; a ring or cycle  yields a ring cluster state; a square lattice yields a two-dimensional cluster state.  A star graph or  complete graph  yields a GHZ state up to a local Clifford transformation. 

\begin{figure*}
    \centering
    \begin{subfigure}
        \centering
        \includegraphics[height=1.2in]{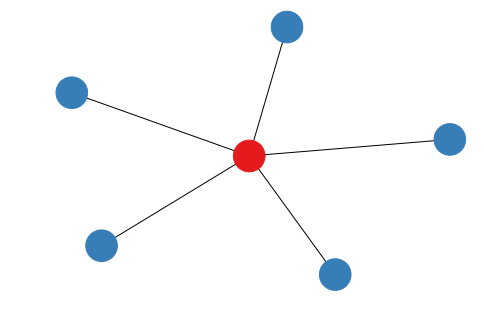}
    \end{subfigure}
    ~
    \begin{subfigure}
        \centering
        \includegraphics[height=1.2in]{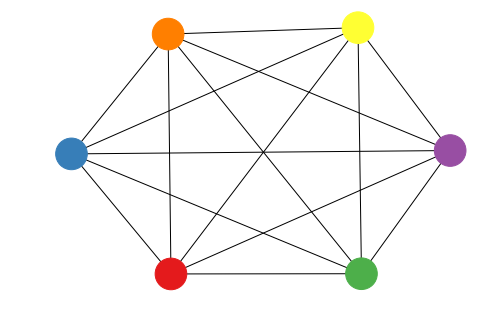}
    \end{subfigure}
    \caption{A star graph and a complete graph with the same number of vertices (depicted with minimum colorings) are equivalent under LC. Their associated graph states are both LC-equivalent to the GHZ state.}
    \label{fig:star-complete}
\end{figure*}

Alternatively, the graph state $|G\>$ is uniquely defined (up to a phase factor)
as the common eigenstate with eigenvalue 1 of the  $n$ stabilizer generators
\begin{align}\label{eq:graph-stabilizer-generator}
S_j = X_j \otimes \bigotimes_{(j,k) \in E} Z_k,  
\end{align}
which generate the stabilizer group $\S_G$ of $|G\>$. The elements of $\S_G$ can be labeled by vectors in $\Z_2^n$ as in \eqref{eq:Sy} with $k=n$, that is, 
\begin{equation}\label{eq:GraphStab}
S^{\bfy}=\prod_{j=1}^n S_j^{y_j}, \quad \bfy\in \Z_2^n. 
\end{equation}
For each $\bfw\in \Z_2^{2n}$, let $|G_\bfw\>$ be the common eigenstate of $S_j$ with eigenvalue $(-1)^{w_j}$ for $j=1,2,\ldots,n$. Then \eqref{eq:StabState2} reduces to 
\begin{equation}\label{eq:GraphState2}
|G_\bfw\>\<G_\bfw|=\prod_{j=1}^n \frac{1+(-1)^{w_j}S_j}{2}=\sum_{\bfy\in \Z_2^n} (-1)^{\bfw\cdot \bfy} S^\bfy.
\end{equation}
The stabilizer states $|G_\bfw\>$ for $\bfw\in \Z_2^n$ form a  stabilizer basis, known as a \emph{graph basis}.

Let $\mathbf{1}$ be the $n\times n$ identity matrix over $\Z_2$. The \emph{canonical basis matrix} of the graph state $|G\>$ is defined
 as the vertical concatenation of $\mathbf{1}$ and the adjacency matrix $A$ and is denoted by $\tilde{A}=(\mathbf{1};A)$ [in contrast with the horizontal concatenation denoted by $(\mathbf{1},A)$]. 
The symplectic vector associated with the stabilizer operator $S^\bfy$ in \eqref{eq:GraphStab}  can be expressed as  $\tilde{A}\bfy$, so the isotropic subspace $V_G$ associated with the graph state $|G\>$ is given by
\begin{equation}\label{eq:GraphIsoSpace}
V_G=\{\tilde{A}\bfy |\bfy\in \Z_2^n\}. 
\end{equation}

Thanks to the special form of the canonical basis matrices of graph states, Lemma~\ref{lem:StabFidelity} can be simplified as follows.
\begin{lemma}\label{lem:StabGraphFidelity}
	Suppose $G$ is an $n$-vertex graph  with adjacency matrix $A$  and $\S$ is a maximal stabilizer group of $n$ qubits with basis matrix $M$. Then 
	\begin{align}
	&|\S_G\cap\bar{\S}|=|\ker(M_\rmz^T +M_\rmx^TA )|=2^{n-\rank(M_\rmz^T +M_\rmx^TA )}, \label{eq:StabGraphIntersect}\\
	&\max_{\bfw'\in \Z_2^n}|\<G_\bfw|\S_{\bfw'}\>|^2=
	2^{-n}|\ker(M_\rmz^T +M_\rmx^TA )|\nonumber\\
	&=2^{-\rank(M_\rmz^T +M_\rmx^TA )}, \label{eq:StabGraphFidelity}
	\end{align}
	where $M_\rmx$ ($M_\rmz$) denotes  the submatrix of $M$  composed of the first (last) $n$ rows. 
\end{lemma}

\begin{lemma}\label{lem:GraphFidelity}
Suppose $G$ and $G'$ are two $n$-vertex graphs with adjacency matrices $A$ and $A'$, respectively. 	
 Then 
	\begin{align}
	&|\S_G\cap\bar{\S}_{G'}|=|\ker(A+A')|=2^{n-\rank(A+A')}, \label{eq:GraphIntersect}\\
	&\max_{\bfw'\in \Z_2^n}|\<G_\bfw|G_{\bfw'}\>|^2=
	2^{-n}|\ker(A+A')|=2^{-\rank(A+A')}. \label{eq:GraphFidelity}
	\end{align}
\end{lemma}

It is known that every stabilizer state is equivalent to a graph state under some local Clifford  transformation (LC)
\cite{schlingemann2001b,grassl2002,nest2004}. In addition, Calderbank-Shor-Steane (CSS) states are equivalent to graph states of two-colorable graphs, and vice versa \cite{chen2007}. Recall that a stabilizer state is a CSS state if its stabilizer group can be expressed as the product of two groups one of which is generated by $X$ operators for individual qubits, while the other is generated by $Z$ operators. 
Two graph states $|G\>$ and $|G'\>$ are equivalent under LC iff the corresponding graphs $G$ and $G'$ are equivalent under local complementation, hence both transformations are abbreviated as LC. An LC with respect to a vertex $j$ turns adjacent (nonadjacent) vertices in the neighborhood of $j$ into nonadjacent (adjacent) vertices. For example, a complete graph is turned into a star graph under LC with respect to any given vertex. This fact explains why the corresponding graph states are equivalent under LC. Both graph states are equivalent to GHZ states, as illustrated in  Fig.~\ref{fig:star-complete}.

\section{Test projectors for stabilizer states}\label{sec:test}

\subsection{\label{sec:StabTest}Canonical test projectors for stabilizer states}
Pauli measurements are the simplest measurements that can be applied to verifying stabilizer states.  Each Pauli measurement for a single qubit can be specified by a symplectic vector in $\Z_2^2$; to be specific, Pauli $X$, $Y$, and $Z$ measurements correspond to the vectors $(1;0)$, $(0;1)$, and $(1;1)$, respectively, while the trivial measurement corresponds to the vector $(0;0)$. 
Each Pauli measurement on $n$ qubits can be specified by a symplectic vector $\mu=(\mu^\rmx;\mu^\rmz)$ in $\Z_2^{2n}$, which means  the Pauli operator measured for qubit $j$ is  $\rmi^{\mu^\rmx_j \mu^\rmz_j} X_j^{\mu^\rmx_j}Z_j^{\mu^\rmz_j}$ for $j=1,2,\ldots,n$.
The weight of the Pauli measurement is the number of qubit $j$ such that $(\mu^\rmx_j,\mu^\rmz_j)\neq (0,0)$. The Pauli measurement is \emph{complete} if  the weight is equal to $n$, so that the Pauli measurement for every  qubit is nontrivial. In that case, the corresponding symplectic vector $\mu$ is also called complete. The set of complete symplectic vectors in $\Z_2^{2n}$ is denoted by $(\Z_2^{2n})_\rmC$.

Let $\T_\mu$ be the stabilizer group generated by  the $n$ local Pauli operators $\rmi^{\mu^\rmx_j \mu^\rmz_j} X_j^{\mu^\rmx_j}Z_j^{\mu^\rmz_j}$ associated with the Pauli measurement; then $|\T_\mu|=2^k$, where  $k$ is the weight of the Pauli measurement. Let $\bar{\T}_\mu=\T_\mu \cup (-\T_\mu)$ be the signed stabilizer group. 
Each outcome of the Pauli measurement corresponds to a common eigenspace of the stabilizer group $\T_\mu$. It can be specified by a vector $\bfv\in \Z_2^n$, which corresponds to the common eigenspace  of $\rmi^{\mu^\rmx_j \mu^\rmz_j} X^{\mu^\rmx_j}Z^{\mu^\rmz_j}$ with eigenvalue $(-1)^{v_j}$ for $j=1,2,\ldots,n$. The  projector onto the eigenspace  reads 
\begin{align}\label{eq:OutcomeProj}
\Pi_{\mu,\bfv}= \bigotimes_{j=1}^n \frac{\id + (-1)^{v_j}\rmi^{\mu^\rmx_j \mu^\rmz_j} X_j^{\mu^\rmx_j}Z_j^{\mu^\rmz_j}}{2}.
\end{align}
Note that $v_j$ can only take on the value 0 if the Pauli measurement for qubit $j$ is trivial; otherwise, $\Pi_{\mu,\bfv}=0$. So those vectors $\bfv$ of interest to us belong to a subspace of $\Z_2^n$ of dimension $k$. Nevertheless, vectors outside this subspace do not affect the following analysis.

Suppose we want to verify the $n$-qubit stabilizer  state $|\S\>$ associated with the stabilizer group $\S$. Then any test operator $P$ based on the  Pauli measurement $\mu$ is a linear combination of $\Pi_{\mu,\bfv}$ for $\bfv\in \Z_2^n$.  To guarantee that the target state $|\S\>$ can always pass the test, $P$ must satisfy $P\geq \Pi_{\mu,\bfv}$ whenever $\<\S|\Pi_{\mu,\bfv}|\S\> >0$. The \emph{canonical test projector} based on the  Pauli measurement $\mu$ is defined as 
\begin{equation}
P_\mu =\sum_{\<\S|\Pi_{\mu,\bfv}|\S\> >0} \Pi_{\mu,\bfv}.
\end{equation}
This concept was introduced for GHZ states in \cite{li2019_GHZ}. 
By construction we have $P_\mu\leq E$ for any other test operator $E$ based on the same Pauli measurement, so it is natural to choose canonical test projectors if we want to construct an optimal verification protocol.

The stabilizer group $\L_\mu=\S\cap \bar{\T}_\mu$ is called a \emph{local subgroup} of $\S$ associated with the Pauli measurement $\mu$. Given any two stabilizer operators $S, S'$ in $\L_\mu$, the tensor factors of  $S, S'$  for any given qubit commute with each other. Therefore, $\L_\mu$ is locally commutative, and hence the name. The stabilizer operators in $\L_\mu$ can be measured simultaneously by the Pauli measurement $\mu$.

\begin{lemma}\label{lem:TestProjCanonical}
For any   $\mu\in \Z_2^{2n}$, the canonical test projector $P_\mu$ is identical to the stabilizer code projector associated with $\L_\mu$, that is,    
\begin{equation}\label{eq:TestProjCanonical}
P_\mu =\frac{1}{|\L_\mu |}\sum_{S\in \L_\mu }S.   
\end{equation}
\end{lemma}
\begin{proof}[Proof of Lemma~\ref{lem:TestProjCanonical}]
	 According to \eqref{eq:OverLap2},  $\<\S|\Pi_{\mu,\bfv}|\S\>$ is equal to either 0 or $|\bar{\T}_\mu\cap \S|/|\T_\mu|=|\L_\mu|/|\T_\mu|$. Moreover, by \eqref{eq:Overlap3}, it is nonzero for exactly  $|\T_\mu|/|\L_\mu|$ vectors $\bfv$ in $\Z_2^n$.
When $\<\S|\Pi_{\mu,\bfv}|\S\>$ is nonzero, the support of $\Pi_{\mu,\bfv}$ is contained in the stabilizer code $\caH_\mu$  associated with $\L_\mu$. So the support of $P_\mu$ lies in  $\caH_\mu$. In addition, we have
\begin{equation}
\tr(P_\mu)=\frac{|\T_\mu|}{|\L_\mu|}\frac{2^n}{|\T_\mu|}=\frac{2^n}{|\L_\mu|},
\end{equation}
so the support of $P_\mu$ has the same dimension as $\caH_\mu$.  It follows that $P_\mu$ must be the projector onto $\caH_\mu$, which implies \eqref{eq:TestProjCanonical}. 
\end{proof}

Lemma~\ref{lem:TestProjCanonical} implies that every canonical test projector of the stabilizer state $|\S\>$  is diagonal in the stabilizer basis associated with the stabilizer group $\S$, the diagonal elements are equal to either 0  or 1, and the rank of the test projector is equal to a power of 2. As a consequence, all canonical test projectors commute with each other. If the local subgroup $\L_\mu$ is trivial, that is, $|\L_\mu|=1$, then the test projector is equal to the identity and thus trivial; the corresponding Pauli measurement is thus useless to verify the stabilizer state $|\S\>$. 

To determine the diagonal elements of $P_\mu$ in the stabilizer basis, we need to specify a concrete basis. To this end, we can choose  any minimal set of generators for $\S$, say $S_1, S_2, \ldots, S_n$.   For each $\bfw\in \Z_2^n$, let $|\S_\bfw\>$ be the common eigenstate of $S_j$ with eigenvalues $(-1)^{w_j}$ for $j=1,2,\ldots, n$. Then $\{|\S_\bfw\>\}_{\bfw\in \Z_2^n}$ forms a stabilizer basis (cf.~Sec.~\ref{sec:StabState}).
Define $\bm{a}_\mu$ as the 
 $(2^n-1)\times 1$ column vector  composed of the entries 
\begin{equation}
\bm{a}_{\mu,\bfw}:=\<\S_\bfw|P_\mu|\S_\bfw\>,\quad \bfw\in \Z_2^n,\quad \bfw\neq  0,
\end{equation}
then the  test projector $P_\mu$ is determined by  $\bm{a}_\mu$ given that $\<\S_\bfw|P_\mu|\S_\bfw\>=1$ when $\bfw=0$. In view of this fact, the vector $\bm{a}_\mu$ is referred to as the \emph{test vector} associated with 
 the test projector $P_\mu$ or the Pauli measurement $\mu$ (with respect to a given stabilizer basis). Here the index $\bfw$ can also be replaced by the natural number $\sum_j 2^{j-1} w_j$ if necessary.

To efficiently compute the diagonal elements $\<\S_\bfw|P_\mu|\S_\bfw\>$ of the test projector $P_\mu$ and determine the test vector $\bm{a}_\mu$, we need to introduce additional tools.
Let $\scrU_1$ ($\scrU_2$) be the index set of qubits for which the Pauli measurement $\mu$ is nontrivial (trivial), that is, 
\begin{align}
\scrU_1&=\{j|\mu_j^\rmx= 1 \mbox{ or } \mu_j^\rmz=1 \},\\
\scrU_2&=\{j|\mu_j^\rmx=0 \mbox{ and } \mu_j^\rmz=0 \}.
\end{align}
Let  $M_\S$ be the basis matrix of $\S$ associated with the  generators  $S_1, S_2, \ldots, S_n$; denote by $M_{\S}^\rmx$ and $M_{\S}^\rmz$ the first $n$ rows and last $n$ rows, respectively.
Define
\begin{align}
M_{\S,\mu} &:=\diag(\mu^\rmz)M_{\S}^\rmx+\diag(\mu^\rmx)M_{\S}^\rmz,\quad \mu\in \Z_2^{2n},\\
\tilde{M}_{\S,\mu}&:=(M_{\S,\mu} (\scrU_1);M_{\S}^\rmx(\scrU_2); M_{\S}^\rmz(\scrU_2)),
\end{align}
where   $M_{\S,\mu} (\scrU_1)$ is the matrix composed of the rows of $M_{\S,\mu}$ indexed by the set $ \scrU_1$ and similarly for $M_{\S}^\rmx(\scrU_2)$ and $M_{\S}^\rmz(\scrU_2)$. Note that $\tilde{M}_{\S,\mu}=M_{\S,\mu}$ if the Pauli measurement $\mu$ is complete, in which case $\scrU_2$ is empty.
For an $m\times n$ matrix $M$ defined over $\Z_2$, denote by $\rank(M)$ the rank of $M$ and $\ker(M)$ the kernel of $M$:
\begin{equation}
\ker(M)=\{\bfy\in \Z_2^n|M \bfy=0\}. 
\end{equation}
Denote by $\rspan(M)$ the row span of $M$:
\begin{equation}
\rspan(M)=\{\bfv M  |\bfv\in \Z_2^m\}. 
\end{equation}

\begin{theorem}\label{thm:TestProjStab} 	
\begin{align}
\L_\mu&=\{S^\bfy| \bfy\in \ker(\tilde{M}_{\S,\mu})\},\label{eq:LocalSubgroupStab}\\
P_\mu &=\frac{1}{|\ker(\tilde{M}_{\S,\mu}) |}\sum_{\bfy\in \ker(\tilde{M}_{\S,\mu})}S^\bfy,\label{eq:TestProjStab}\\	
	\bm{a}_{\mu,\bfw}&=\<\S_\bfw|P_\mu|\S_\bfw\>
	=\begin{cases}
	1 & \bfw\in \rspan(\tilde{M}_{\S,\mu}),\\
	0 &\mathrm{otherwise}. 
	\end{cases}	\label{eq:StabDiagElement}
	\end{align}
\end{theorem}
Theorem~\ref{thm:TestProjStab} is proved in Appendix \ref{app:thm1_2}. Note that $\tilde{M}_{\S,\mu}$ reduces to $M_{\S,\mu}$ when the Pauli measurement is complete.  As an implication of Theorem~\ref{thm:TestProjStab},
 the order of the local subgroup $\L_\mu$ and the rank of the test projector $P_\mu$ read
\begin{align}
|\L_\mu|&=|\ker(\tilde{M}_{\S,\mu})|=2^{n-\rank(\tilde{M}_{\S,\mu})},\label{eq:LocalSubgroupOrder}\\
\tr(P_\mu)&=|\rspan(\tilde{M}_{\S,\mu})|=2^{\rank(\tilde{M}_{\S,\mu})}. \label{eq:StabProjRank}
\end{align}
The test projector $P_\mu$ is trivial iff $\tilde{M}_{\S,\mu}$ has  rank $n$ (full rank).

Let $r=\rank(\tilde{M}_{\S,\mu})$ and let $\bfw_1, \bfw_2,\ldots,\bfw_r$ be a basis in $\rspan(\tilde{M}_{\S,\mu})$. Then $\rspan(\tilde{M}_{\S,\mu})$ coincides with the span of these basis vectors, that is,
\begin{equation}
\rspan(\tilde{M}_{\S,\mu})=\Biggl\{\sum_{j=1}^r a_j \bfw_j |a_j\in \Z_2^n\; \forall j\Biggr \};
\end{equation}
this observation is helpful to computing the diagonal elements of $P_\mu$ in the stabilizer basis. Equation \eqref{eq:StabDiagElement} is equivalent to the following equation,
\begin{equation}\label{eq:StabDiagElement2}
\<G_\bfw|P_\mu|G_\bfw\>=\begin{cases}
1 & \bfw\cdot\bfy=0 \quad \forall \bfy \in
\ker(\tilde{M}_{\S,\mu}),\\
0 &\mathrm{otherwise},
\end{cases}
\end{equation}
since $\bfw\in \rspan(\tilde{M}_{\S,\mu})$ iff $\bfw\cdot\bfy=0$ for all $\bfy \in \ker(\tilde{M}_{\S,\mu})$. 
 Let $\bfy_1,\bfy_2, \ldots, \bfy_{n-r}$ be a basis in $\ker(\tilde{M}_{\S,\mu})$, which has  dimension $n-r$;  then  $\bfw\in \rspan(\tilde{M}_{\S,\mu})$ iff $\bfw\cdot\bfy_j=0$ for $j=1,2,\ldots, n-r$.

To determine the minimum rank of canonical test projectors, it suffices to consider complete Pauli measurements. According to \eqref{eq:StabProjRank}, we have
\begin{equation}
\min_{\mu\in \Z_2^{2n}} P_\mu=\min_{\mu\in (\Z_2^{2n})_\rmC}P_\mu=2^{\kappa(\S)}, \label{eq:MinRankStab}
\end{equation}
where 
\begin{equation}\label{eq:kappaS}
\kappa(\S):=\min_{\mu\in (\Z_2^{2n})_\rmC} \rank(M_{\S,\mu}), 
\end{equation} 
given that $\tilde{M}_{\S,\mu}=M_{\S,\mu}$ for $\mu\in (\Z_2^{2n})_\rmC$. Note that $2^{\kappa(\S)}$ is also the minimum rank of all test operators of $|\S\>$ based on Pauli measurements since the canonical test projector attains the minimum rank for a given Pauli measurement.
The following lemma relates the geometric measure of entanglement $\Lambda(\S)$ of any stabilizer state $\ket{\S}$ \cite{wei2003} to $\kappa(\S)$.
Denote by $\mathrm{Prod}$ the set of pure product states and by $\mathrm{Prod}_\rmP$ the set of pure product states whose single-qubit reduced states are eigenstates of Pauli operators.
Define
\begin{align}
	\Lambda(\S)&=\Lambda(|\S\>):= \max_{|\varphi\>\in \mathrm{Prod}} |\<\varphi|\S\>|^2, \label{eq:Lambda} \\
	\Lambda_\rmP(\S)&=\Lambda(|\S\>):= \max_{|\varphi\>\in \mathrm{Prod}_\rmP} |\<\varphi|\S\>|^2. \label{eq:LambdaP}
\end{align}

\begin{lemma}\label{lem:LambdaKappa}
	\begin{align}\label{eq:LambdaKappaStab}
	\Lambda(\S)\geq 	\Lambda_\rmP(\S)=2^{-\kappa(\S)}
	\end{align}
\end{lemma}

\begin{proof}
The inequality in \eqref{eq:LambdaKappaStab} follows from the definitions of $\Lambda(\S)$ and $\Lambda_\rmP(\S)$ in \eqref{eq:Lambda} and \eqref{eq:LambdaP}. To prove the equality $\Lambda_\rmP(\S)=2^{-\kappa(G)}$,
	note that each state in $\mathrm{Prod}_\rmP$ is a stabilizer state, and the state projector has the form in \eqref{eq:OutcomeProj} with  $\mu\in (\Z_2^{2n})_\rmC$. Therefore,
	\begin{align}
	\Lambda_\rmP(\S)&=\max_{\mu\in (\Z_2^{2n})_\rmC,\bfv} \<G|\Pi_{\mu,\bfv\in \Z_2^n}|G\>=\max_{\mu\in (\Z_2^{2n})_\rmC} 2^{-n}|\L_\mu|\nonumber\\
	&=\max_{\mu\in (\Z_2^{2n})_\rmC}2^{-\rank(M_{\S,\mu})}=2^{-\kappa(\S)},
	\end{align}	
where the second equality follows from \eqref{eq:StabFidelity2}, and the third equality follows from \eqref{eq:LocalSubgroupOrder} given that $\tilde{M}_{\S,\mu}=M_{\S,\mu}$ for $\mu\in (\Z_2^{2n})_\rmC$. 
\end{proof}

\subsection{\label{sec:TestProjGraph}Canonical test projectors for graph states}

For graph states, the discussions in the previous section can be simplified. Here we only consider complete Pauli measurements.
Suppose that $|G\>$ and $\{|G_\bfw\>\}_{\bfw\in \Z_2^n}$ are the  graph state and the corresponding graph basis associated with the graph $G=(V,E)$ as defined in Sec.~\ref{sec:GraphState}. Let $A$  be the adjacency matrix of $G$; then $\tilde{A}=(\mathbf{1};A)$ is the canonical basis matrix for $|G\>$. Define 
\begin{equation}\label{eq:Amu}
A_\mu =\diag(\mu^\rmz)+\diag(\mu^\rmx)A,\quad \mu\in \Z_2^{2n}. 
\end{equation}
The rank, kernel, and row span of $A_\mu$ are denoted by $\rank(A_\mu)$, $\ker(A_\mu)$, and $\rspan(A_\mu)$, respectively.

\begin{theorem}\label{thm:TestProjGraph}Suppose $\mu\in (\Z_2^{2n})_\rmC$; then
\begin{align}
\L_\mu&=\{S^\bfy| \bfy\in \ker(A_{\mu})\},\label{eq:LocalSubgroupGraph}\\
P_\mu &=\frac{1}{|\ker(A_{\mu}) |}\sum_{\bfy\in \ker(A_{\mu})}S^\bfy,\label{eq:TestProjGraph}\\		
\bm{a}_{\mu,\bfw}&=\<G_\bfw|P_\mu|G_\bfw\>
	=\begin{cases}
	1 & \bfw\in \rspan(A_\mu),\\
	0 & \mathrm{otherwise}. \label{eq:GraphDiagElement}
	\end{cases}
	\end{align}
\end{theorem}
\noindent Theorem~\ref{thm:TestProjGraph} is a special case of Theorem~\ref{thm:TestProjStab} tailored to the verification of a graph state based on a complete Pauli measurement; a simplified proof is presented in Appendix \ref{app:thm1_2}. As a corollary, we have
\begin{align}
|\L_\mu|&=|\ker(A_{\mu})|=2^{n-\rank(A_\mu)},\label{eq:GraphLocalSubgroupOrder}   \\
\tr(P_\mu)&=|\rspan(A_\mu)|=2^{\rank(A_\mu)}. \label{eq:GraphPmuRank}
\end{align}
The test projector $P_\mu$ is trivial iff $A_\mu$ has full rank.

Let $r=\rank(A_\mu)$ and let $\bfw_1 \bfw_2,\ldots,\bfw_r$ be a basis in $\rspan(A_\mu)$. Then $\rspan(A_\mu)=\{\sum_{j=1}^r a_j \bfw_j |a_j\in \Z_2^n\; \forall j \}$. Equation \eqref{eq:GraphDiagElement} is equivalent to the following equation,
\begin{equation}\label{eq:GraphDiagElement2}
\<G_\bfw|P_\mu|G_\bfw\>=\begin{cases}
1 & \bfw\cdot\bfy=0 \quad \forall \bfy \in
 \ker(A_\mu),\\
 0 &\mathrm{otherwise},
\end{cases}
\end{equation}
given that $\bfw\in \rspan(A_\mu)$ iff $\bfw\cdot\bfy=0$ for all $\bfy \in \ker(A_\mu)$. 
Let $\bfy_1,\bfy_2, \ldots, \bfy_{n-r}$ be a basis in $\ker(A_\mu)$, which has dimension $n-r$; then  $\bfw\in \rspan(A_\mu)$ iff $\bfw\cdot\bfy_j=0$ for all $j=1,2,\ldots, n-r$. These observations are helpful to computing $\<G_\bfw|P_\mu|G_\bfw\>$ efficiently.

Similar to \eqref{eq:MinRankStab}, the minimum rank of test projectors $P_\mu$ for $|G\>$ is given by 
\begin{equation}
\min_{\mu\in (\Z_2^{2n})_\rmC}P_\mu=2^{\kappa(G)},
\end{equation}
where 
\begin{equation}
\kappa(G):=\min_{\mu\in (\Z_2^{2n})_\rmC} \rank(A_\mu)=\kappa(\S_G).
\end{equation} 
Define
\begin{align}
\Lambda(G)&=\Lambda(|G\>):= \max_{|\varphi\>\in \mathrm{Prod}} |\<\varphi|G\>|^2=\Lambda(\S_G), \label{eq:LambdaGraph}\\
\Lambda_\rmP(G)&=\Lambda(|G\>):= \max_{|\varphi\>\in \mathrm{Prod}_\rmP} |\<\varphi|G\>|^2=\Lambda_\rmP(\S_G). \label{eq:LambdaGraphP}
\end{align} 
Then Lemma~\ref{lem:LambdaKappa} reduces to 
\begin{lemma}\label{lem:LambdaKappaGraph}
\begin{align}\label{eq:LambdaKappaGraph}
\Lambda(G)\geq 	\Lambda_\rmP(G)=2^{-\kappa(G)}.
\end{align}
\end{lemma}

The following lemma sets an upper bound for $\kappa(G)$, which in turn yields an upper bound for the minimum rank of test projectors $P_\mu$. 
\begin{lemma}\label{lem:kappaUB}
	$\kappa(G)\leq n-\alpha(G)$, where $\alpha(G)$ is the independence number of $G$. 
\end{lemma}
\noindent Lemma~\ref{lem:kappaUB} and \eqref{eq:LambdaKappaGraph} imply the following inequalities
\begin{equation}
\Lambda(G)\geq 	\Lambda_\rmP(G)\geq 2^{\alpha(G)-n},
\end{equation}
which was originally derived in \cite{markham2007}.

\begin{proof}[Proof of Lemma~\ref{lem:kappaUB}]
Let $B$ be an independent set of the graph $G = (V,E)$ with cardinality $\alpha(G)$. Consider the Pauli measurement in which $X$ measurements are performed on all qubits in $B$ and $Z$ measurements are performed on all qubits in the complement $\overline{B} = V\setminus B$. Let $\mu$ be the corresponding symplectic vector,
then $\mu^\rmx_j=1$ iff $j\in B$, while $\mu^\rmz_j=1$ iff $j\in \overline{B}$. 
Since $B$ is an independent set, $\rspan(\diag(\mu^\rmx)A) \subseteq \rspan(\diag(\mu^\rmz))$. Consequently, $\rank A_{\mu} = \rank (\diag(\mu^\rmz)) = |\overline{B}|$, and we have
\begin{equation}
\kappa(G)\leq \rank(A_\mu)=|\overline{B}|=n-\alpha(G).
\end{equation}
\end{proof}

\subsection{\label{sec:admissible}Admissible test projectors}
Suppose we want to verify the stabilizer state $|\S\>$ based on Pauli measurements. Let $E$ be a test operator  based on the Pauli measurement specified by the symplectic vector $\mu\in \Z_2^{2n}$.  The test operator $E$ is \emph{not admissible} if there exists another test operator $E'$  based on a Pauli measurement such that $E'\leq E$ and $\tr(E')<\tr(E)$. Let $P_\mu$ be the canonical test projector associated with the Pauli measurement $\mu$, then $E\geq P_\mu$, so an admissible test operator is necessarily a canonical test projector. The test vector $\bm{a}_\mu$ is (not) admissible if the  test projector $P_\mu$  is (not) admissible.
A Pauli measurement is (not) admissible if the corresponding canonical test projector is (not) admissible. Previously, the concepts of admissible measurements and admissible test projectors are considered only for GHZ states \cite{li2019_GHZ}.
\begin{proposition}\label{pro:incomplete-setting}
Any admissible Pauli measurement is complete. 
\end{proposition}
\begin{proof}
Following the proof for GHZ states in \cite{li2019_GHZ}, we prove the contrapositive, that an incomplete Pauli measurement is inadmissible. Without loss of generality, consider a Pauli measurement of weight $n-1$ on the $n$-qubit target stabilizer state $\ket{\S}$. After $n-1$ single-qubit Pauli measurements, the reduced states of the remaining party for all possible outcomes are eigenstates of a Pauli operator. So we can  always find an extra Pauli measurement on the remaining qubit that makes the canonical test projector smaller.
\end{proof}

As an example, consider the three-qubit linear cluster state defined by the three stabilizer generators $S_1 = XZ\id, S_2 =ZXZ$, and $S_3=\id ZX$. The state can be written as
\begin{align}
	\ket{G} = \frac{\ket{+}\ket{0}\ket{+} + \ket{-}\ket{1}\ket{-}}{\sqrt{2}},
\end{align}
where $\ket{\pm}$ is an eigenstate of $X$ with the eigenvalue $\pm1$.
If we perform $X$ and $Z$ measurements on the first and second qubit, respectively, then 
the reduced state of the third qubit is left in one of the $X$ eigenstates. Thus, measuring $X$ on the last qubit gives an admissible test projector $P_{XZX}$, whereas measuring $Y$ or $Z$ results in the same inadmissible test projector that has higher rank than $P_{XZX}$. Further calculation shows that there are five  admissible Pauli measurements in total, namely, $XZX$, $ZXZ$, $ZYY$, $YYZ$, $YXY$ (cf. Table~\ref{tab:optimal}).

Different Pauli measurements may give rise to the same canonical test projector. When more than one measurement settings share the same canonical test projector, the projector can be realized by an incomplete Pauli measurement that the two settings have in common. Therefore, this canonical test projector cannot be admissible by
by Proposition~\ref{pro:incomplete-setting}. This observation yields the following lemma.
\begin{lemma}\label{lem:admissibleOneToOne}
	No two admissible Pauli measurements  lead to the same canonical test projector.
\end{lemma}
Lemma~\ref{lem:admissibleOneToOne} shows that there is a one-to-one correspondence between admissible Pauli measurements and admissible canonical test projectors. 

\begin{corollary}\label{cor:admissible}
A canonical test projector based on a Pauli measurement is admissible iff it cannot be realized by an incomplete Pauli measurement. 
\end{corollary}
\begin{proof}
If the test projector can be realized by an incomplete Pauli measurement, then it is not admissible by Proposition~\ref{pro:incomplete-setting}. If the test projector  is realized by a complete Pauli measurement $\mu$, but is not admissible, then we can find a complete Pauli measurement $\mu'$, such that $P_{\mu'}\leq P_\mu$ and $\tr(P_{\mu'})<\tr(P_\mu)$. Note that $P_\mu$ can also be realized by the incomplete Pauli measurement that the two Pauli measurements $\mu$ and $\mu'$ have in common. This observation completes the proof of Corollary~\ref{cor:admissible}.
\end{proof}

The following lemma is a useful tool for determining admissible test projectors of stabilizer states; it is a direct consequence of Theorem~\ref{thm:TestProjStab}. 
\begin{lemma}\label{lem:admissibleStab}
	Suppose $P_\mu$ and $P_{\mu'}$ are the canonical test projectors for the stabilizer state $|\S\>$ based on  Pauli measurements $\mu$ and $\mu'$, respectively.  Then the following statements are equivalent: 
	\begin{enumerate}
		\item $P_{\mu'}\leq P_{\mu}$;
		\item $\|\bm{a}_{\mu'}\circ \bm{a}_\mu\|_1=\|\bm{a}_{\mu'}\|_1$; 
		\item $\rspan(\tilde{M}_{\S,\mu'})\leq \rspan(M_{\S,\mu})$; 
		\item $\ker(\tilde{M}_{\S,\mu'})\geq \ker(\tilde{M}_{\S,\mu})$; 
		\item $(\tilde{M}_{\S,\mu}; \tilde{M}_{\S,\mu'})$ and $\tilde{M}_{\S,\mu}$ have the same rank. 
	\end{enumerate}
\end{lemma}
\noindent Here $\bm{a}_{\mu}$ and $\bm{a}_{\mu'}$ are the test vectors associated with $P_\mu$ and $P_{\mu'}$, respectively; $\bm{a}_{\mu'}\circ \bm{a}_\mu$ denotes the element-wise product of $\bm{a}_{\mu'}$ and $\bm{a}_\mu$; and $\|\bm{a}_{\mu'}\|_1$ denotes the 1-norm of $\bm{a}_{\mu'}$, that is, $\|\bm{a}_{\mu'}\|_1=\sum_{\bfw\in \Z_2^n,\bfw\neq 0} \bm{a}_{\mu',\bfw}$. 

In the case of graph states, Lemma~\ref{lem:admissibleStab} can be simplified as follows, assuming that  the Pauli measurements are complete.
\begin{lemma}\label{lem:admissibleGraph}
Suppose $P_\mu$ and $P_{\mu'}$ are the canonical test projectors for the graph state $|G\>$ based on the complete Pauli measurements $\mu$ and $\mu'$, respectively. Let $A$ be the adjacency matrix of $G$. Then the following statements are equivalent:
\begin{enumerate}
\item $P_{\mu'}\leq P_{\mu}$;
\item $\|\bm{a}_{\mu'}\circ \bm{a}_\mu\|_1=\|\bm{a}_{\mu'}\|_1$; 
\item $\rspan(A_{\mu'})\leq \rspan(A_{\mu})$; 
\item $\ker(A_{\mu'})\geq \ker(A_{\mu})$; 
\item $(A_\mu; A_{\mu'})$ and $A_\mu$ have the same rank. 
\end{enumerate}
\end{lemma}

 The total number of admissible test projectors for $|G\>$ is denoted by $\eta(G)$. Note that $\eta(G)$ is also the total number of admissible Pauli measurements by Lemma~\ref{lem:admissibleOneToOne}. 
 The value of $\eta(G)$ for every equivalent class of connected  graph states up to seven qubits is shown in Table~\ref{tab:optimal}.

\section{Optimal verification of graph states}\label{sec:spectral-gap-optimization}

In this section we study optimal verification of entangled (possibly nonconnected) graph states based on Pauli measurements. (Generalization to stabilizer states is straightforward.)
First, we show that the spectral gap of any verification protocol based on
separable measurements cannot surpass $2/3$. We also derive a necessary condition for Pauli measurements to attain the maximum spectral gap. Then we present a simple algorithm for constructing  optimal verification protocols based on Pauli measurements. Using this algorithm, 
we  construct an optimal verification protocol, attaining the maximum spectral gap of $2/3$, for every equivalent class of connected graph states up to seven qubits.  We believe that  the maximum spectral gap of $2/3$ can be attained for all graph states associated with nonempty graphs. Here  a nonempty graph is a graph that has at least one edge.

\subsection{Efficiency limit of separable and Pauli measurements}

Given a graph state $|G\>$ associated with the graph $G=(V,E)$, denote by $\nu(G)$ the maximum spectral gap of verification operators that are based on nonadaptive Pauli measurements, and by $\nu_\sep(G)$ the maximum spectral gap attainable by separable measurements. A verification operator $\Omega$ based on Pauli measurements is optimal if $\nu(\Omega)=\nu(G)$.

\begin{lemma}\label{lemma:optimal}
	Suppose $G=(V,E)$ is a (possibly nonconnected) nonempty graph. Then  $\nu(G)\leq \nu_\sep(G)\leq 2/3$.
\end{lemma}
\begin{proof}
	The inequality $\nu(G)\leq \nu_\sep(G)$ follows from the fact that 	Pauli measurements are separable measurements.	To prove the inequality $\nu_\sep(G)\leq 2/3$, note that the  graph state $\ket{G}$ is maximally entangled with respect to at least one bipartition into a single qubit and the other $n-1$ qubits since $G$ is nonempty \cite{hein2004,hein2006}.
	A separable measurement for $\ket{G}$  is necessarily separable with respect to the bipartition. Now note that the spectral gap for a Bell state or two-qubit maximally entangled state cannot be increased by increasing the local dimension of one of the subsystems. Therefore, the spectral gap of any verification operator $\Omega$ for $\ket{G}$ that is based on separable measurements cannot be larger than the maximum spectral gap $2/3$ achievable for a Bell state \cite{hayashi2009,pallister2018,ZhuH2019O,wang2019}, that is,  $\nu(\Omega)\leq 2/3$. This result implies the inequality $\nu_\sep(G)\leq 2/3$ and confirms Lemma~\ref{lemma:optimal}. 
\end{proof}

While Lemma~\ref{lemma:optimal} establishes the bound for a vast class of measurements, it provides no information as to what kind of measurements can saturate the bound.
Specializing to Pauli measurements, we prove a simple necessary requirement for constructing an optimal verification protocol:  $X$, $Y$, and $Z$ should be measured with an equal probability for each qubit.

\begin{theorem}\label{thm:optimal}
	Suppose $G=(V,E)$ is a (possibly nonconnected)  graph with no isolated vertex and $\Omega$ is a verification operator of $|G\>$ based on Pauli measurements. Let $p_{j}^X$,  $p_{j}^Y$,  $p_{j}^Z$ be the probability that $X$, $Y$, and $Z$ measurements are performed on the $j$th qubit, respectively. 
	Then 
	\begin{align}
	\beta(\Omega)&\geq \max_{j\in V} \max \{p_{j}^X,p_{j}^Y,p_{j}^Z\}\geq \frac{1}{3},  \label{eq:betaLB}  \\
	\nu(\Omega)&\leq \frac{2}{3}.\label{eq:nuUB}
	\end{align}
	If $\nu(\Omega)=2/3$ or $\beta(\Omega)=1/3$, then $p_{j}^X=p_{j}^Y=p_{j}^Z=1/3$ for all $j\in V$. 
\end{theorem}
\begin{proof}
	Equation \eqref{eq:nuUB} follows from \eqref{eq:betaLB}. To prove \eqref{eq:betaLB}, we can assume that the verification protocol employs only canonical test projectors. Let $P_j^X$  be the product of all canonical test projectors based on Pauli measurements with $X$ measurement on qubit $j$, then $P_j^X$ is also a test projector for $|G\>$ (not necessarily associated with a Pauli measurement) since all canonical test projectors commute with each other. In addition, $P_j^X$ has rank at least 2 since, otherwise, the reduced state of $|G\>$ on qubit $j$ would be a pure state, in contradiction with the assumption that $G$ has no isolated vertex, which implies that  every single-qubit reduced state is completely mixed \cite{hein2004,hein2006}.  Define $P_j^Y$ and $P_j^Z$ in a similar way. 
	Then we have
	\begin{gather}
	\Omega\geq p_j^X P_j^X+p_j^Y P_j^Y+p_j^Z P_j^Z \quad \forall j,\\
	\beta(\Omega)\geq  \max \{p_{j}^X,p_{j}^Y,p_{j}^Z\}\geq \frac{1}{3}\quad  \forall j,
	\end{gather}
	which confirms \eqref{eq:betaLB} and implies \eqref{eq:nuUB}. The last statement in Theorem~\ref{thm:optimal} is an immediate consequence of \eqref{eq:betaLB}. 
\end{proof}

If $G$ is not connected and has $J$ connected components $G_1, G_2, \ldots, G_J$. Then $|G\>=\bigotimes_{j=1}^J|G_j\>$ and we have 
\begin{equation}
\nu(G) = \min_j \nu(G_j)
\end{equation}
according to Proposition~\ref{prop:nonconnected-cover} in Appendix \ref{app:nonconnected}. In addition, if $\Omega_j$ are optimal verification operators for $|G_j\>$, then $\bigotimes_{j=1}^J \Omega_j$ is an optimal verification operator for $|G\>$.
As a corollary, $\nu(G)=2/3$ if $\nu(G_j)=2/3$ for $j=1,2,\ldots, J$. To construct optimal verification protocols, therefore, we can focus on graph states of connected graphs.  

\subsection{\label{sec:Algorithm}Algorithm for constructing an optimal verification protocol}
To construct an optimal verification protocol based on Pauli measurements, it suffices to consider canonical test projectors associated with admissible Pauli measurements. Suppose we perform the Pauli measurement $\mu$ with probability $p_\mu$, then the resulting verification operator reads 
\begin{equation}\label{eq:GraphVO}
\Omega=\sum_{\mu} p_\mu P_\mu, 
\end{equation}
where $P_\mu$ is the canonical test projector associated with the Pauli measurement $\mu$. Since all canonical test projectors are diagonal in the graph basis, the operator $\Omega$ is also diagonal in the graph basis. Let 
\begin{equation}\label{eq:lambda}
\bm{\lambda}=\sum_{\mu} p_{\mu} \bm{a}_{\mu}. 
\end{equation}
where $\bm{a}_{\mu}$ is the test vector associated with the projector $P_\mu$ and the Pauli measurement $\mu$.
Then  all eigenvalues of $\Omega$ except the one associated with the target state are contained in $\bm{\lambda}$. In particular, we have
\begin{align}\label{eq:betaTestVec}
\beta(\Omega)&=\|\bm{\lambda}\|_\infty=\max_{\bfw\in \Z_2^n, \bfw\neq 0} p_\mu a_{\mu,\bfw} \nonumber \\
&=\max_{\bfw\in \Z_2^n, \bfw\neq 0}\sum_{\mu,\rspan(A_\mu)\ni \bfw} p_\mu,
\end{align}
where $\|\bm{\lambda}\|_\infty$ denotes the maximum of the elements in the vector $\bm{\lambda}$, and the last equality follows from Theorem~\ref{thm:TestProjGraph}.  
Maximizing the spectral gap $\nu(\Omega)=1-\beta(\Omega)$ is equivalent to minimizing $\beta(\Omega)$. To this end we need to only consider admissible test vectors (cf.~Sec.~\ref{sec:admissible}).

Define the \emph{test matrix} $\caA$ as the matrix composed of all admissible test vectors $\bm{a}_\mu$ as column vectors, and  let $\bm{p}$ be the column  vector composed of the probabilities $p_\mu$. Then $\bm{\lambda}=\caA\bm{p}$ and $\beta(\Omega)=\|\caA\bm{p}\|_\infty$. The minimization of $\beta(\Omega)$ can be cast as a constrained optimization problem
\begin{equation}
\begin{aligned}
&\underset{p \in \R^m}{\text{minimize}} &\quad \left\Vert \mathcal{A}\bm{p} \right\Vert_{\infty}& \\
&\text{subject to} &\quad p_\mu \ge 0& \\
&&	\sum_\mu p_\mu = 1&. \\
\end{aligned}
\end{equation}
After choosing a proper order of the test vectors
the minimization can be expressed as a linear programming:
\begin{equation}\label{algo:max-spectral-gap}
\begin{aligned}
&\text{minimize} &y &  \\
&\text{subject to} &(\mathcal{A}\bm{p})_j \le y,\quad 	p_j \ge 0\quad \forall j,& & \\
&&	\sum_j p_j = 1.& \\
\end{aligned}
\end{equation}
\noindent Putting all these together, we have a recipe for finding an optimal verification protocol as shown in Algorithm~\ref{algo:graph-verify}. Our Python module that implements the algorithm using the convex optimization package \texttt{cvxpy} and open source solvers are available on \href{https://github.com/ninnat/graph-state-verification}{Github}.

\begin{algorithm}[h]
	\caption{\protect\raggedright Finding an optimal verification protocol for graph states}
	\begin{algorithmic}[1]
		\Input 
		\Desc{$A$}{Adjacency matrix for an $n$-qubit graph state $\ket{G}$}
		\Output 
		\Desc{$\scrM$}{\hspace{2.5em} The set of Pauli measurements (represented by vectors in $(\Z_2^{2n})_\rmC$) employed in the optimal verification protocols}
		\Desc{$\{p_\mu\}_{\mu\in \scrM}$}{\hspace{3em}Probabilities for individual measurement settings as in \eqref{eq:GraphVO}}
		\Desc{$\nu(\Omega)$}{\hspace{2.5em}Optimal spectral gap}
		\State Determine all nontrivial test vectors. 
		\begin{enumerate}
			\item For each $\mu\in (\Z_2^{2n})_\rmC$, compute the matrix $A_\mu$ in \eqref{eq:Amu}. 
			\item Discard $A_{\mu}$ that has full rank by checking the determinant.
			\item  Compute the test vector $\bm{a}_\mu$ by virtue of Theorem~\ref{thm:TestProjGraph}. 
		\end{enumerate}
		\State Determine all admissible test vectors and construct the test matrix $\caA$. 
		\State Solve the linear program \eqref{algo:max-spectral-gap}. 
	\end{algorithmic}\label{algo:graph-verify}
\end{algorithm}

\subsection{Special examples}\label{sec:examples}

\begin{figure*}
  \begin{center}
    \includegraphics[width=0.6\linewidth]{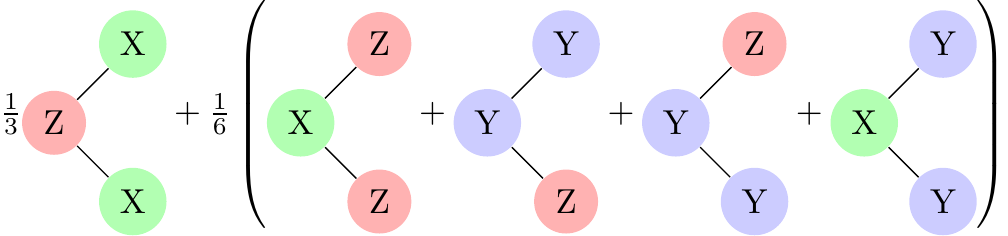}
  \end{center}
	\caption{Graphical representation of the optimal verification scheme for the star graph. Each term in the sum represents a Pauli measurement setting. 
	Measurement settings in the parenthesis are chosen with the probability 1/6 each; they are given in \eqref{eq:GHZ-Y-settings} in terms of all possible sets $\Y$ of parties that perform $Y$ measurements (displayed in blue) with $|\Y|=2t$ where $t=0,1,\dots,\lfloor n/2 \rfloor$.}
	\label{fig:GHZ}
\end{figure*}

To illustrate the algorithm presented in  Sec.~\ref{sec:Algorithm}, we first consider graph states associated with 
 star graphs  (No. 3, 5, 9, and 20 in \cite{hein2004}; omitted in Table~\ref{tab:optimal}). These states can be turned into  GHZ states by applying the Hadamard gate on each non-central qubit.

All information about the canonical test projector associated with a Pauli measurement $\mu = (\mu^\rmx;\mu^\rmz) \in \Z_2^{2n}$ is encoded in the matrix $A_{\mu}$ defined in \eqref{eq:Amu}, which can be thought of as the matrix version of the symplectic inner product between $\mu$ and the adjacency matrix $A$. For the star graph with $n$ vertices we have
\begin{align}
A = 
\begin{pmatrix}
	0 & 1 & \cdots & 1 \\
	1 & & & \\
	\vdots & & \mbox{\huge0} & \\
	1 & & & &  
\end{pmatrix}, 
&&
A_{\mu} = 
\begin{pmatrix}
	\mu^\rmz_1 & \mu^\rmx_1 & \mu^\rmx_1 & \cdots & \mu^\rmx_1 \\
	\mu^\rmx_2 & \mu^\rmz_2 & & \raisebox{-1ex}{\LARGE 0} &  \\
	\mu^\rmx_3 & & \mu^\rmz_3 & & \\
	\vdots & \mbox{\LARGE 0} & & \ddots & \\
	\mu^\rmx_n & & & & \mu^\rmz_n
\end{pmatrix}.
\end{align}
According to \eqref{eq:GraphPmuRank}, the rank of the canonical test projector associated with $\mu$ is the cardinality of the subspace spanned by the rows of $A_{\mu}$. Analysis shows that  the smallest subspace is obtained when $\mu^\rmz_1 = 1$, $\mu^\rmx_j = 1$ for $j \neq 1$, and every other component is zero; that is, when the Pauli measurement is $ZX\cdots X$. This Pauli measurement effective measures the $n-1$ stabilizer generators associated with the $n-1$ non-central qubits.
The resulting rank-2 test projector reads
\begin{align}\label{eq:rank-2-GHZ}
P_0 = \ketbra{0}{0}\otimes(\ketbra{+}{+})^{\otimes n-1} + \ketbra{1}{1}\otimes(\ketbra{-}{-})^{\otimes n-1}.
\end{align}
It is the only admissible Pauli measurement that measures $Z$ on the central qubit and, thus, must be included in any optimal verification protocol with probability 1/3 according to Theorem~\ref{thm:optimal}.

\begin{figure*}
  \begin{center}
    \includegraphics[width=1\linewidth]{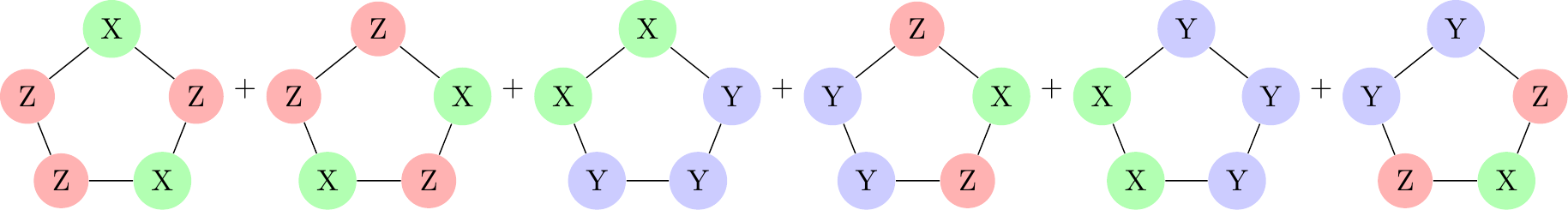}
  \end{center}
	\caption{Graphical representation of an optimal verification scheme for the five-qubit ring cluster state. Each term in the sum represents a Pauli measurement setting, and each setting is chosen with the probability 1/6.}
	\label{fig:ring}
\end{figure*}

By virtue of the connection with the GHZ state and the result derived in \cite{li2019_GHZ},
the other admissible Pauli measurements can be described compactly as follows.
We perform either $X$ or $Y$ measurement on the central qubit,  while either $Y$ or $Z$ measurement on each of the other qubits.
Let $\Y$ be the set of parties that perform $Y$ measurements assuming that $|\Y|$ is even, and $\overline{\Y}$ be the set of parties that perform a $Z$ measurement on a non-central qubit or an $X$ measurement on the central qubit. Let $|\Y|=2t$ with $t = 0,1,\dots,\lfloor n/2 \rfloor$ and
\begin{align}\label{eq:GHZ-Y-settings}
	S_{\Y} &= (-1)^t\bigotimes_{j\in\Y} Y_j \bigotimes_{j'\in\overline{\Y}} \sigma_{j'},
\end{align}
where $\sigma_1 = X_1$ and $\sigma_{j'} = Z_{j'}$ for all $j' \neq 1$. The Pauli measurement effectively  measures the product of the stabilizer $S_1$ of the central qubit together with all the stabilizers of all non-central qubits in the set $\Y$ (cf.~Fig.~\ref{fig:GHZ}). Thus, the corresponding canonical test projector 
\begin{align}
	P_{\Y} = \frac{\id + S_{\Y}}{2}
\end{align}
has rank $2^{n-1}$, and there are $2^{n-1}$ such canonical test projectors.
The unique optimal verification strategy is constructed by performing the test $P_0$ with probability $1/3$ and all other admissible tests $P_\Y$ with probability $1/(3\times 2^{n-2})$ each. The resulting  verification operator reads
\begin{align}
	\Omega = \frac{1}{3}\left(P_0 + \sum_{\Y} \frac{1}{2^{n-2}}P_{\Y} \right)=\frac{1}{3}\left(2|G\>\<G|+\openone \right).
\end{align} 
Compared with other $n$-qubit  graph states with $n\leq 7$, it turns out the graph state associated with the star graph requires the most number of measurement settings to construct an optimal verification protocol.

In general, optimal verification strategies based on Pauli measurements are not unique.  For example, let us consider the five-qubit ring cluster state  (No.~8 in Table~\ref{tab:optimal}). Calculation shows that this graph state has 21 admissible Pauli measurements. Two distinct optimal verification strategies can be constructed from six admissible Pauli measurements each as  presented below (left and right columns); each setting is measured with probability 1/6, which is consistent with Theorem~\ref{thm:optimal}.
\begin{center}
	\begin{tabular}{c|c}
		XZZXZ & ZXZZX \\
		ZZXZX & ZXZXZ \\
		XXYYY & XYYYX \\
		ZYYZX & XZYYZ \\
		YXXYY & YYXXY \\
		YYZXZ & YZXZY
	\end{tabular}
\end{center}
Note that the two strategies have no measurement setting in common. All 12 
 canonical test projectors associated with these Pauli measurements have rank 8 (some of other nine admissible test projectors have higher ranks).

\subsection{Generic case}

By means of the algorithm presented in Sec.~\ref{sec:Algorithm}, we have found an optimal verification protocol based on Pauli measurements for each equivalent class of connected graph states up to seven qubits \cite{hein2004}. It turns out the maximum spectral gap can always attain the upper bound $2/3$ for separable measurements presented in Lemma~\ref{lemma:optimal}. Therefore, we believe that the upper bound $2/3$ for the spectral gap can be saturated by Pauli measurements for all graph states. More details can be found in Table~\ref{tab:optimal}. 
Note that  optimal verification strategies are in general not unique, and it is possible that the numbers of measurement settings in the optimal  strategies presented in Table~\ref{tab:optimal} can be reduced.
According to this table,  linear cluster states require fewer measurement settings to construct optimal protocols compared with most other graph states:  only 5, 6, 6, 8, and 12 settings are required for 3, 4, 5, 6, and 7 qubits, respectively. 
Curiously, graph state No. 42 requires only 6 measurement settings to construct an optimal verification protocol, which is fewer than  the number 7 of qubits. This is the only connected graph state that has this property as far as we know.

\section{Optimal verification of graph states with $X$ and $Z$ measurements}\label{sec:XZ-verification}
According to \cite{ZhuH2019E}, all graph states can be verified using only $X$ and $Z$ measurements. Here we present an upper bound for the spectral gap of any verification protocol based on $X$ and $Z$ measurements and show that this bound can be saturated for all connected graph states up to seven qubits (which again implying that the bound can be saturated for nonconnected graph states built from these graphs according to Proposition~\ref{prop:nonconnected-cover}). In addition, we construct an optimal verification protocol for every ring cluster state. 

\subsection{\label{sec:XZprotocol}Verification protocols based on $X$ and $Z$ measurements}

\begin{corollary}\label{cor:optimalXZ}
Suppose $G=(V,E)$ is a (possible nonconnected) nonempty graph, and $\Omega$ is a verification operator of the graph state $|G\>$
 based on $X$ and $Z$ measurements. Then the spectral gap satisfies $\nu(\Omega)\leq 1/2$. To saturate this upper bound, $X$ and $Z$ should be measured with an equal probability for each qubit. 
\end{corollary}
This result is a simple corollary of  Theorem~\ref{thm:optimal}. Moreover, the bound $\nu(\Omega)\leq 1/2$  applies whenever 
each party can perform only two types of Pauli measurements. 

To construct an optimal protocol for the graph state $|G\>$ based on $X$ and $Z$ measurements, it suffices to consider canonical test projectors based on complete $X$ and $Z$ measurements. Each subset $B$ of $V$ determines a complete Pauli measurement by performing $X$ measurements on qubits in $B$ and $Z$ measurements on qubits in the complement $\overline{B}=V\setminus B$; the corresponding local subgroup and canonical test projector are denoted by $\L_B$ and  $P_B$, respectively. 
Conversely, each complete Pauli measurement based on $X$ and $Z$ measurements is determined by such a subset in  $V$. In view of this fact, complete Pauli measurements based on $X$ and $Z$ measurements and the corresponding canonical test projectors can be labeled by subsets in $V$. 
Let $\mu$ be the symplectic vector associated with the Pauli measurement determined by $B$,
then $\mu^\rmx_j=1$ iff $j\in B$, while $\mu^\rmz_j=1$ iff $j\in \overline{B}$.
Conversely, $B=\{j|\mu^\rmx_j=1\}$. Based on this observation, the local subgroup $\L_B$ and the test projector $P_B$ can be determined by virtue of Theorem~\ref{thm:TestProjGraph}, where the subscript $\mu$ characterizing the Pauli measurement can be replaced by $B$ to manifest the role of the set $B$; in particular, $A_\mu$ can be expressed as $A_B$. According to  \eqref{eq:Amu} and the above discussion, we have 
\begin{equation}
\rank(A_B)\geq |\overline{B}|=n-|B|,
\end{equation}
and the inequality  is saturated iff $B$ is an independent set (cf. the proof of Lemma~\ref{lem:kappaUB}).
In conjunction with \eqref{eq:GraphLocalSubgroupOrder} and  \eqref{eq:GraphPmuRank}, this equation
implies the following corollary.
\begin{corollary}\label{cor:TestProjXZ}
Let $B$ be a subset of the vertex set $V$ of the graph $G$. Then 	
\begin{align}
|\L_B|&\leq 2^{|B|},    \label{eq:LocalSubgroupXZ}\\
\tr(P_B)&\geq 2^{n-|B|}.  \label{eq:TestProjXZ}
\end{align}
The inequality \eqref{eq:LocalSubgroupXZ} is saturated iff $B$ is an independent set; the same holds for the inequality \eqref{eq:TestProjXZ}. 
\end{corollary}

When $B$ is an independent set, $\L_B$ is generated by the stabilizer operators $S_j$ for all $j\in B$, that is,
\begin{equation}\label{eq:LocalSubgroupInd}
\L_B=\langle\{S_j| j\in B\}\rangle;
\end{equation} 
accordingly, the canonical test projector reads
\begin{equation}\label{eq:TestProjInd}
P_B=\prod_{j\in B}\frac{1+S_j}{2}. 
\end{equation}
In addition, the test vector in \eqref{eq:GraphDiagElement} reduces to
\begin{align}
\bm{a}_{B,\bfw}&=\<G_\bfw|P_B|G_\bfw\>
=\begin{cases}
1 & \supp(\bfw)\in \overline{B},\\
0 &\mathrm{otherwise}, \label{eq:GraphDiagElementInd}
\end{cases}
\end{align}
where 
\begin{equation}
\supp(\bfw)=\{j| w_j=1\}.
\end{equation}
Test projectors based on  independent sets play a key role in constructing the cover and coloring protocols in  \cite{ZhuH2019E}. Compared with other test projectors, these test projectors  are easy to visualize; in addition, it is easier to compute the spectral gap of verification operators based on such test projectors. Nevertheless,  more general test projectors are helpful to enhance the spectral gap. 

The following lemma follows from \eqref{eq:graph-stabilizer-generator}, \eqref{eq:LocalSubgroupInd}, and \eqref{eq:TestProjInd}.
\begin{lemma}\label{lem:TwoTestBB}
Suppose $B$ and $B'$ are two subsets of the vertex set of the graph $G=(V,E)$, and $B$ is an independent set. If $P_{B'}\leq P_B$ or equivalently, $\L_{B'} \supseteq \L_B$, then $B\subseteq B'$. If  $B$ and $B'$ are independent sets, then  $P_{B'}\leq P_B$ iff $B\subseteq B'$. 
\end{lemma}

Optimal verification  protocols based on $X$ and $Z$ measurements can be found using a similar algorithm as presented in Sec.~\ref{sec:Algorithm} with minor modification. It turns out the bound $\nu(\Omega)\leq 1/2$ in Corollary~\ref{cor:optimalXZ} can be saturated for all entangled graph states up to seven qubits shown in Table~\ref{tab:minimal-settings}. (In fact, this result also holds for graph states not listed in the table; see Appendix \ref{app:general-graphs}.) 
\begin{proposition}
For any entangled CSS state or entangled graph state associated with a two-colorable graph, the maximum spectral gap achievable by $X$ and $Z$ measurements is $1/2$.
\end{proposition}
\begin{proof}
First, consider an entangled graph state. According to Corollary~\ref{cor:optimalXZ}, the spectral gap of any verification operator based on $X$ and $Z$ measurements is upper bounded by $1/2$. If the graph is two-colorable, then the bound $1/2$ can be saturated by a coloring protocol proposed in \cite{ZhuH2019E}, so the maximum spectral gap achievable by $X$ and $Z$ measurements is $1/2$.

Next, consider  an entangled CSS state. According to Theorem~\ref{thm:optimal},   the spectral gap of any verification operator based on $X$ and $Z$ measurements is also upper bounded by $1/2$. Meanwhile, the bound can be saturated by the protocol composed of the two measurement settings $XX\cdots X$ and $ZZ\cdots Z$ chosen with an equal probability. This result is consistent with the fact that any CSS state is equivalent to a graph state associated with a two-colorable graph, and vice versa \cite{chen2007}.
\end{proof}

\subsection{\label{sec:AdmissibleXZ}Admissible test projectors based on $X$ and $Z$ measurements}

In contrast to the definitions in Sec.~\ref{sec:admissible}, there are two sensible definitions of admissible measurements and test projectors based
on $X$ and $Z$ measurements. Such a test projector (and corresponding measurement) is  admissible if there is no smaller test projectors based on  Pauli measurements. The test projector (and corresponding measurement) is weakly admissible if there is no smaller test projectors  based on $X$ and $Z$ measurements. Let $\eta_{XZ}$ be the number of admissible test projectors based
on $X$ and $Z$ measurements and $\eta_{XZ}'$  the number of weakly admissible test projectors. 
Obviously, an admissible test projector based on $X$ and $Z$ measurements is weakly admissible,  so  we have $\eta_{XZ} \le \eta_{XZ}'$. This inequality is usually strict since a  weakly admissible measurement is not necessarily admissible.
Actually,  $\eta_{XZ}$ and $\eta'_{XZ}$ are not invariant under LC and the equality $\eta_{XZ} = \eta_{XZ}'$ is too strong to expect.

As an example, let $U_n$ be an $n$-qubit Clifford unitary operator with $n\geq 3$ that interchanges $Y$ and $Z$ for  every qubit, and consider the state $\ket{\Psi} = U_n\ket{G}$, where $\ket{G}$ is the star-graph state as discussed in Sec.~\ref{sec:examples}.
Here, the measurement $X\cdots X$ yields a rank-4 test projector $P_{\mu}$ that is weakly admissible. However, this projector cannot be admissible according to the discussion in  Sec.~\ref{sec:examples} (cf. \cite{li2019_GHZ}). To be concrete, let 
$P_{\mu'}$  be the unique rank-2 admissible test projector based on the measurement $YX\cdots X$ (which would have been $ZX\cdots X$ before the unitary operator $U_n$ is applied), then we have $P_{\mu'} \le P_{\mu}$ and  $\tr(P_{\mu'}) < \tr(P_{\mu})$, so $P_\mu$ is not admissible.

For graph states, it turns out the two notions of admissibility coincide. The following proposition is proved in Appendix~\ref{app:coloring}.
\begin{proposition}\label{pro:eta_XZ}
Given  any graph state $|G\>$, a  test projector based on $X$ and $Z$ measurements  is admissible iff it is weakly admissible; $\eta_{XZ}(G) = \eta_{XZ}'(G)$.
\end{proposition}
\noindent The values of $\eta_{XZ}(G)$ for graph states up to seven qubits are shown in Table~\ref{tab:minimal-settings}.
These results suggest that  the number of admissible settings based on $X$ and $Z$ measurements grows only linearly with the number of qubits,
although  the total number of admissible  measurement settings grows exponentially  (cf. Table~\ref{tab:optimal}).

The following lemma clarifies all admissible test projectors that are based on independent sets; see Appendix~\ref{app:coloring} for a proof. This lemma provides further insight on the cover and coloring protocols proposed in \cite{ZhuH2019E}.
\begin{lemma}\label{lemma:max-ind-iff-admissible}
Suppose $B$ is an independent set of the graph $G=(V,E)$. Then
the test projector $P_B$ given  in \eqref{eq:TestProjInd} is admissible iff
$B$ is a maximal independent set.
\end{lemma}

\subsection{Optimal verification of ring cluster states}
A ring cluster state is the graph state associated with a ring graph, that is, a cycle. Here we show that the upper bound $1/2$ for the spectral gap in Corollary~\ref{cor:optimalXZ} can be saturated for all ring cluster states by constructing an optimal verification protocol using $X$ and $Z$ measurements only.

Let $|G\>$ be the $n$-qubit ring cluster state associated with the ring graph $G=(V,E)$, where $V=\{1,2,\ldots, n\}$. When $n$ is even, according to \cite{ZhuH2019E}, an optimal  protocol for verifying $|G\>$ can be constructed using two measurement settings associated with the two independent sets $B_1=\{1,3,\ldots,n-3, n-1\}$ and $B_2=\{2,4,\ldots,n-2, n\}$, respectively.

To construct an optimal protocol when $n$ is odd, we first introduce   $n+1$ subsets of $V$ defined as follows,
\begin{align}
&B_j=\{j,j+2,j,\ldots, j+n-3\},\quad j=1,2,\ldots,n,\\
&B_{n+1}=V=\{1,2,\dots, n\};
\end{align}
note that the first $n$ subsets are independent sets of $G$. Each set $B_j$ for $j=1,2,\ldots, n+1$ defines a canonical test by performing $X$ measurements on qubits in $B_j$ and $Z$ measurements on qubits in the complement $V\setminus B_j$ as described in Sec.~\ref{sec:XZprotocol}. Now an optimal protocol can be constructed by performing the  $n+1$ tests $P_{B_j}$   with probability $1/(n+1)$ each; the resulting verification operator reads
\begin{equation}
\Omega=\frac{1}{n+1}\sum_{j=1}^{n+1}P_{B_j}. 
\end{equation}

To corroborate our claim, note that the test projectors $P_{B_j}$ for $j=1,2,\ldots, n$ are determined by \eqref{eq:TestProjInd} and all have rank $2^{(n+1)/2}$; the corresponding test vectors are determined by \eqref{eq:GraphDiagElementInd}. To determine the test vector $\bm{a}_{B_{n+1}}$ associated with the  test projector $P_{B_{n+1}}$, note that $A_{B_{n+1}}=A$, so that $a_{B_{n+1},\bfw}=1$ iff $\bfw\in \rspan(A)$ according to Theorem~\ref{thm:TestProjGraph}. In addition,  $A$ has rank $n-1$, and  $\ker(A)$ is spanned by $\{1,1,\ldots, 1\}$, which implies that
\begin{align}\label{eq:TestVecBx}
a_{B_{n+1},\bfw}&=\<G_\bfw|P_{B_{n+1}}|G_\bfw\>=\begin{cases}
1 & \mbox{$\bfw$ has even weight},\\
0& \mbox{otherwise}.
\end{cases}
\end{align}
Therefore,
\begin{align}
&\beta(\Omega)=\frac{1}{n+1}\max_{\bfw\in \Z_2^n,\bfw\neq 0}\sum_{j=1}^{n+1} a_{B_j,\bfw}\nonumber\\
&=
\frac{1}{n+1}\max_{\bfw\in \Z_2^n,\bfw\neq 0}\bigl(
|\{j|\supp(\bfw)\in \overline{B}_j\}|+ \bm{a}_{B_{n+1},\bfw}\bigr)\nonumber\\
&=\frac{1}{2}.
\end{align}
Here the last inequality follows from \eqref{eq:TestVecBx} and the fact that $|\{j|\supp(\bfw)\in \overline{B}_j\}|\leq (n+1)/2$, and the inequality is saturated iff $\bfw$ has weight 1. 

\section{Verification with minimum number of settings}\label{sec:minimal-verification}

Besides the spectral gap, the number of distinct measurement settings is an important figure of merit in practice.  What is the  measurement complexity required to verify a stabilizer state? Here we summarize our understanding about this problem. Since every stabilizer state is equivalent to a graph state under local Clifford transformations, we can  focus on graph states. All statements proven in this section that are not specific to a particular graph apply to connected graphs as well as nonconnected ones. The relation between verification of a nonconnected graph state and verification of its connected components is clarified in Appendix~\ref{app:nonconnected}. 
 
\subsection{Local cover number}

Suppose  $G=(V,E)$ is  a graph with adjacency matrix $A$. Let $|G\>$ be  the graph state associated with the graph $G$ and the stabilizer group $\S_G$. To verify $|G\>$ based on Pauli measurements, we need to construct a verification operator $\Omega$ with positive spectral gap, that is, $\nu(\Omega)>0$ or, equivalently, $\beta(\Omega)<1$. The following lemma clarifies the necessary requirements. 
\begin{lemma}\label{lem:ValidProtocol}
	Suppose $\Omega=\sum_{\mu\in \scrM} p_\mu P_\mu^G $ with $\scrM\in (\Z_2^{2n})_\rmC$ with $p_\mu>0$ for all $\mu\in \scrM$. Then the following statements are equivalent: 
	\begin{enumerate}
		\item $\nu(\Omega)>0$;
		\item $\< \cup_{\mu\in \scrM}\L_\mu\>=\S_G$;
		\item $\Span(\cup_{\mu\in \scrM} \ker(A_\mu))=\Z_2^n$;
		\item $\cap_{\mu\in \scrM} \rspan(A_\mu)=0$;
		\item $\prod^\circ_{\mu\in \scrM}\bm{a}_{\mu}=0$. 
	\end{enumerate}
	Here $\prod^\circ_{\mu\in \scrM}\bm{a}_{\mu} = 0$ means that the element-wise product
	$\prod^\circ_{\mu\in \scrM}a_{\mu,\bfw} = \prod^\circ_{\mu\in \scrM}\av{G_{\bfw}|P_{\mu}|G_{\bfw}} =0$ for all $\bfw\in \Z_2^n$ with $\bfw\neq 0$. 
\end{lemma}
\begin{proof}
	The equivalence of statements 1 and 2 follows from the fact that $P_\mu$ is the stabilizer code projector associated with the local subgroup $\L_\mu$ according to Lemma~\ref{lem:TestProjCanonical}. The equivalence of statements 2 and 3 follows from \eqref{eq:LocalSubgroupGraph} in Theorem~\ref{thm:TestProjGraph}. The equivalence of statements 3 and 4 is a simple fact in linear algebra. The equivalence of statements 1 and 5 follows from \eqref{eq:betaTestVec}. 
\end{proof}

The \emph{local cover number}  $\tilde{\chi}(G)$ is defined as  the minimum number of Pauli measurement settings required to 
 verify $|G\>$. It  is equal to  the minimum cardinality of $\scrM$ with $\scrM\subset (\Z_2^{2n})_\rmC$ that satisfies one of the statements in Lemma~\ref{lem:ValidProtocol} and can be expressed as follows,
\begin{align}
\tilde{\chi}(G)&=\min\{|\scrM| \,|\, \<\cup_{\mu\in \scrM}\L_\mu\>=\S_G  \} \label{eq:LocalCoverNum}\\
&=\min\{|\scrM| \,|\, \Span(\cup_{\mu\in \scrM} \ker(A_\mu))=\Z_2^n  \}\\
&=\mathrm{min} \{|\scrM| \,|\, \cap_{\mu\in \scrM} \rspan(A_\mu)=0 \}\\
&=\mathrm{min} \{|\scrM| \,|\, \textstyle{\prod^\circ_{\mu\in \scrM}\bm{a}_{\mu}=0} \}.
\end{align}
Here the terminology is inspired by 
 \eqref{eq:LocalCoverNum} according to which $\tilde{\chi}(G)$ is equal to the minimum number of local stabilizer groups required to generate the stabilizer group of $|G\>$. The above equations can be applied to computing $\tilde{\chi}(G)$, although such algorithms are not efficient. In general,  we cannot expect to find an efficient algorithm in view of the connection (via Proposition~\ref{pro:LocalCoverNum} below) between $\tilde{\chi}(G)$ and the chromatic number $\chi(G)$, which is NP-hard to compute \cite{godsil2001}. 
The local cover number  $\tilde{\chi}_{XZ}(G)$ can be defined and computed in a similar way, except that only $X$ and $Z$ measurements are considered.

If one of the statements in Lemma~\ref{lem:ValidProtocol} holds, then  we have $\Span(\cup_{\mu\in \scrM} \ker(A_\mu))=\Z_2^n$, which implies that $2^n\leq \prod_{\mu\in \scrM} |\ker(A_\mu)|$, so that
\begin{align}
&n\leq \sum_{\mu\in \scrM}\dim(\ker(A_\mu))\leq |\scrM|\max_{\mu\in \scrM}\dim(\ker(A_\mu))\nonumber\\
&\leq |\scrM|\max_{\mu\in (\Z_2^n)_\rmC}[n-\rank(A_\mu)]=|\scrM|[n-\kappa(G)]. 
\end{align}
As a corollary, we have 
\begin{equation}\label{eq:LocalCoverNumLB}
\tilde{\chi}(G)\geq\frac{n}{n-\kappa(G)}.
\end{equation}

If  one of the five statements in  Lemma~\ref{lem:ValidProtocol} holds, and $p_\mu=1/m$ for all $\mu\in \scrM$, where $m=|\scrM|$,  then 
\begin{equation}
\beta(\Omega)=\frac{1}{m}\max_{\bfw\in \Z_2^n, \bfw\neq 0} a_{\mu,\bfw}\leq \frac{m-1}{m}
\end{equation}
according to \eqref{eq:betaTestVec}. Here the inequality follows from the fact that $a_{\mu,\bfw}=0$ for at least one $\mu\in \scrM$  for each $\bfw\in \Z_2^n$ with $\bfw\neq 0$. Therefore, $\nu(\Omega)\geq 1/m$. If $m=|\scrM|=\tilde{\chi}(G)$, then $\nu(\Omega)\leq 1/m$ according to Proposition~\ref{pro:SpectralGapMinSettings} given that the number of measurement settings cannot be reduced. These observations imply the following theorem, which clarifies the efficiency limit of verification protocols based on the minimum number of Pauli measurement settings.
\begin{theorem}\label{thm:SpectralMinMax}
	The maximum spectral gap of  verification operators of $|G\>$  based on $\tilde{\chi}(G)$ distinct Pauli measurements is $1/\tilde{\chi}(G)$.
\end{theorem}
\noindent Theorem~\ref{thm:SpectralMinMax} follows from  Proposition~\ref{pro:SpectralGapMinSettings} and the commutativity of canonical test projectors, so it applies to all graph states, irrespective whether the graph is connected or not. According to the coloring protocol proposed in \cite{ZhuH2019E}, by virtue of $\chi(G)$ settings based on $X$ and $Z$ measurements, we can achieve spectral gap $1/\chi(G)$. Theorem~\ref{thm:SpectralMinMax} may be seen as a generalization of this result.

\subsection{Connection to the chromatic number}
Here we discuss the connection between the local cover numbers $\tilde{\chi}(G), \tilde{\chi}_{XZ}(G)$ and the chromatic number $\chi(G)$. Note that $\tilde{\chi}(G)$ is invariant under LC  of the graph $G$ (corresponding to LC of the graph state $|G\>$), but this is not the case for $\tilde{\chi}_{XZ}(G)$ and $\chi(G)$. To remedy this defect,  define $\tilde{\chi}_2(G)$ as the minimum number of  settings required to verify $|G\>$ when each party can perform only two different Pauli measurements. Note that each party needs to perform at least two different Pauli measurements to verify any graph state of a connected graph with two or more vertices \cite{ZhuH2019E,ZhuH2019AdvL}. 
Define $\chi_\LC(G)$ as the minimum chromatic number of any graph that is equivalent to $G$ under LC, that is, 
\begin{equation}
\chi_\LC(G)=\min_{G'\overset{\LC}{\simeq} G}\chi(G'),
\end{equation}
where the symbol $\overset{\LC}{\simeq}$ means equivalence under LC.

\begin{proposition}\label{pro:LocalCoverNum}
$\tilde{\chi}(G)\leq \tilde{\chi}_2(G)\leq \chi_\LC(G)\leq \chi(G)$ for any graph $G$. 
\end{proposition}
\begin{proof}
Here the first and third inequalities follow from the definitions. To prove the second inequality, let $G'$ be a graph that is equivalent to $G$ under LC and satisfies $\chi(G')=\chi_{\LC}(G)$. According to \cite{ZhuH2019E}, $|G'\>$ can be verified by a coloring protocol composed of  $\chi(G')$ distinct settings based on $X$ and $Z$ measurements. Therefore,
\begin{equation}
\tilde{\chi}_2(G)=\tilde{\chi}_2(G')\leq \tilde{\chi}_{XZ}(G')\leq \chi(G')=\chi_{\LC}(G),
\end{equation} 
which confirms the second inequality in Proposition~\ref{pro:LocalCoverNum}. 
\end{proof}

\begin{restatable}{conjecture}{conjChi}\label{con:chi}
	$\tilde{\chi}(G)= \tilde{\chi}_2(G)= \chi_\LC(G)$ for any graph $G$. 
\end{restatable}
\noindent We have verified Conjecture~\ref{con:chi} for all connected graphs  up to seven vertices (graph states up to seven qubits). Actually we have 
\begin{equation}\label{eq:chi-equality}
\tilde{\chi}(G)=\tilde{\chi}_2(G)= \tilde{\chi}_{XZ}(G)= \chi(G)=\chi_{\LC}(G)
\end{equation} 
for all the graphs shown in Table~\ref{tab:minimal-settings} and all nonconnected graphs built from these graphs thanks to Proposition~\ref{prop:nonconnected-cover} in Appendix~\ref{app:nonconnected}.
(This result does not mean that \eqref{eq:chi-equality} holds  for all connected graphs up to seven vertices since $\tilde{\chi}_{XZ}(G)$ and $\chi(G)$ are not invariant under LC; see Appendix~\ref{app:general-graphs} for more detail.) 
Therefore, for all such graphs, 
the maximum spectral gap of  verification operators  based on the minimum number of settings is $1/\chi(G)$ according to Theorem~\ref{thm:SpectralMinMax}.
Incidentally, all protocols with the  minimum number of settings in Table~\ref{tab:minimal-settings} are chosen to be coloring protocols.
Table~\ref{tab:minimal-settings} in addition contains the fractional chromatic number $\chi^*(G)$ for all the graphs listed. The inverse fractional chromatic number $1/\chi^*(G)$ is the maximum spectral gap achievable by the cover protocol proposed in \cite{ZhuH2019E}.

\begin{proposition}\label{pro:chi1}
	$\tilde{\chi}(G)=1$ iff $G$ is an empty  graph (with no edges). 
\end{proposition}

\begin{proof}
	If $G$ is empty, then $|G\>$ is a  product state of the form $|+\>^{\otimes n}$, which can be verified by  performing $X$ measurements on all qubits, so 	$\tilde{\chi}(G)=1$. Conversely, if $|G\>$ can be verified by a single setting based on a Pauli measurement, then $|G\>$ must be a tensor product of eigenstates of local Pauli operators, so $G$ must be an empty graph. Alternatively, Proposition~\ref{pro:chi1} follows from Lemma~3 in \cite{ZhuH2019O}  and Proposition~3 in \cite{ZhuH2019AdvL}.
\end{proof}

As an implication of Propositions~\ref{pro:LocalCoverNum} and \ref{pro:chi1}, we have 	$\tilde{\chi}(G)=2$ for any nonempty two-colorable graph. The following theorem provides a partial converse and confirms Conjecture~\ref{con:chi} in a special case of practical interest. See Appendix~\ref{app:thm:chi2} for a proof.
\begin{theorem}\label{thm:chi2}
	$\tilde{\chi}(G)=2$ iff $\chi_{\LC}(G)=2$. A stabilizer state can be verified by two settings based on Pauli measurements iff it is equivalent to a CSS state or, equivalently, a graph state of a two-colorable graph. 
\end{theorem}

The following proposition determines the local cover numbers of odd ring graphs, which 
 are typical examples of graphs that are not two-colorable. 
\begin{proposition}\label{pro:chiOddRing}
Suppose $G$ is an odd ring graph with at least five vertices; then
	\begin{equation}
	\tilde{\chi}(G)=\tilde{\chi}_2(G)=\tilde{\chi}_{XZ}(G)= \chi_\LC(G)= \chi(G)=3. 
	\end{equation} 	
\end{proposition}
\noindent To prove Proposition~\ref{pro:chiOddRing}, note that $\tilde{\chi}(G)\leq \chi(G)=3$ for the odd ring graph thanks to Proposition~\ref{pro:LocalCoverNum}. Conversely, $\tilde{\chi}(G)\geq 3$ according to \eqref{eq:LocalCoverNumLB} and the following lemma, which is proved in Appendix~\ref{app:lem:LambdaRing}.

\begin{lemma}\label{lem:LambdaRing}
	Suppose $G$ is a ring graph with $n\geq 4$ vertices. Then 
	\begin{align}
	\Lambda_\rmP(G)=2^{-\lfloor(n+1)/2\rfloor}, \label{eq:LambdaRing}\\
	\kappa(G)=\lfloor(n+1)/2\rfloor. \label{eq:kappaRing}
	\end{align}
\end{lemma}

\section{Summary and open problems}\label{sec:summary}

We have investigated  systematically optimal verification of stabilizer states (including graph states in particular) using Pauli measurements. We proved that the spectral gap of any verification operator of any entangled stabilizer state based on separable measurements (including Pauli measurements) is bounded from above by $2/3$.  Moreover, we introduced the concepts of canonical test projectors, admissible Pauli measurements, and admissible test projectors and clarified their properties. By virtue of these concepts, we proposed a simple algorithm for constructing optimal verification protocols based on (nonadaptive) Pauli measurements. Although this algorithm is not efficient for large systems, it enables us to construct an optimal protocol for any stabilizer state up to ten qubits without difficulty. In particular, our calculation shows that the bound $2/3$ for the spectral gap can be saturated for all entangled stabilizer states up to seven qubits. It is quite surprising that the maximum spectral gap seems to be independent of the specific stabilizer state, although different stabilizer states may have very different entanglement structures. 
In the case of graph states,  we also prove that the upper bound for the spectral gap is reduced to $1/2$ if only $X$ and $Z$ measurements are accessible. Again, this bound can be saturated for all entangled graph states up to seven qubits. These results naturally lead to the following conjectures.
\begin{conjecture}\label{con:gapPauli}
For any entangled stabilizer state, the maximum spectral gap of  verification operators based on Pauli measurements is $2/3$. 
\end{conjecture}
\begin{conjecture}\label{con:gapXZ}
	For any graph state of a nonempty graph, the maximum spectral gap of  verification operators based on $X$ and $Z$ measurements is $1/2$. 
\end{conjecture}
Conjecture~\ref{con:gapPauli}  holds for GHZ states according to \cite{li2019_GHZ}. For a graph state associated with the graph $G$, this conjecture amounts to the following equality
\begin{equation}
\nu(G)=\frac{2}{3},
\end{equation}
where  $\nu(G)$ denotes the maximum spectral gap of verification operators for $|G\>$ that are based on  Pauli measurements.
When the local dimension is an odd prime $p$ instead of 2, we believe that the maximum spectral gap of  verification operators based on Pauli measurements is $p/(p+1)$, which  holds for GHZ states according to \cite{li2019_GHZ}.
 Conjecture~\ref{con:gapXZ} holds for graph states associated with two-colorable graphs and  ring graphs according to Sec.~\ref{sec:XZ-verification}.

In addition, we studied the problem of verifying graph states with the minimum number of settings. For any given graph state $|G\>$, it turns out that the minimum  number of settings  required $\tilde{\chi}(G)$ is upper bounded by the chromatic number $\chi_\LC(G)$ minimized over LC equivalent graphs. In addition, the upper bound still applies even if each party can perform only two types of Pauli measurements, so we have   $\tilde{\chi}(G)\leq \tilde{\chi}_2(G)\leq \chi_\LC(G)$ (cf. Proposition~\ref{pro:LocalCoverNum}). Actually, the two inequalities are saturated for all two-colorable graphs, all graphs up to seven vertices, and ring (or cycle) graphs (cf. Sec.~\ref{sec:minimal-verification} and Table~\ref{tab:minimal-settings}). These facts lead to the following conjecture originally stated in Sec.~\ref{sec:minimal-verification}.
\conjChi*

In the future, it would be desirable to prove or disprove the above conjectures. In either case, we may gain further insight on quantum state verification and stabilizer states themselves. The number of admissible Pauli measurements (or $X$ and $Z$ measurements) and its scaling behavior with the number of qubits are also of interest from the theoretical perspective. In practice, it is desirable to find more efficient approaches for constructing optimal verification protocols and protocols with the minimum number of measurement settings. Furthermore, our study leads to the following question, which is of interest beyond the immediate focus of this work: What are the generic and maximum values of $\tilde{\chi}(G)$, $ \tilde{\chi}_2(G)$, and $\chi_\LC(G)$, respectively, for graphs of $n$ vertices.

\section*{Acknowledgments}

We thank Zihao Li for stimulating discussion.
This work is supported by  the National Natural Science Foundation of China (Grant No.~11875110) and  Shanghai Municipal Science and Technology Major Project (Grant No.~2019SHZDZX01).

\appendix

\begin{appendices}

\section{Proof of Lemma~\ref{lem:StabFidelity}}\label{app:lemma2}
\begin{proof}[Proof of Lemma~\ref{lem:StabFidelity}]
	Let $W$ and $W'$ be the isotropic subspaces associated with $S$ and $S'$, respectively. Then $W'=\{M'\bfy\,|\,\bfy\in \Z_2^n\}$ and 
	\begin{align}
W\cap W'= W^\perp\cap W'
	=\{M'\bfy\,|\,\bfy\in \Z_2^n, M^TJ M'\bfy =0\}. 
	\end{align}
	Therefore,
	\begin{align}
	&|\bar{\S}\cap\S'|=|W\cap W'|\nonumber\\
	&=|\{M'\bfy\,|\,\bfy\in \Z_2^n, M^TJ M'\bfy=0 \}|=|\ker(M^T JM')|\nonumber\\
	&=2^{n-\rank(M^T JM')},
	\end{align}
	which confirms the first two equalities in  \eqref{eq:StabIntersect}. The last equality in \eqref{eq:StabIntersect} follows from the following equality
	\begin{align}
	M^T JM'&=M_\rmz^TM_\rmx'+M_\rmx^TM_\rmz'. 
	\end{align}
	Equation \eqref{eq:StabFidelity3} follows from \eqref{eq:StabFidelity2} and \eqref{eq:StabIntersect}.
\end{proof}

\section{Proofs of Theorem~\ref{thm:TestProjStab} and Theorem~\ref{thm:TestProjGraph}}\label{app:thm1_2}
\begin{proof}[Proof of Theorem~\ref{thm:TestProjStab}]
	Let $V_\mu$ be the isotropic subspace associated with $\T_\mu$. Then the column span of $M_\mu = (\diag(\mu^\rmx);\diag(\mu^\rmz))$, where the semicolon denotes the vertical concatenation, coincides with $V_\mu$. ($M_\mu$ is a basis matrix for $V_\mu$ when the Pauli measurement is complete). Let $N_\mu$ be the $n\times n$ diagonal matrix over $\Z_2$ such that $(N_\mu)_{jj}=1$ iff $j\in \scrU_2$. Then by construction, the column span of the block matrix
	\begin{equation}
	M_\mu^{\perp}=\begin{pmatrix}
	\diag(\mu^\rmx) & N_\mu & 0\\
	\diag(\mu^\rmz) & 0& N_\mu
	\end{pmatrix}
	\end{equation}
	coincides with $V_\mu^{\perp}$, the symplectic complement of $V_\mu$.
	(Note that the first $n$ columns of	$M_\mu^{\perp}$ coincide with $M_\mu$.) 
	Therefore, 
	\begin{equation}
	V_\mu=(V_\mu^{\perp})^{\perp}=\ker((M_\mu^{\perp})^T J),
	\end{equation}
	where $J$ is the symplectic form \eqref{eq:SymplecticForm}.

	Let $V_\S$ be the Lagrangian subspace associated with the stabilizer group $\S$. Then $V_\S=\{M_\S \bfy |\bfy\in \Z_2^n\}$, where $M_S$ is the basis matrix of $\S$.  
	$U_\mu=V_\mu\cap V_\S$ is the isotropic subspace associated with the local subgroup $\L_\mu$. In addition, we have
	\begin{align}
	U_\mu&=V_\mu\cap V_\S=\ker((M_\mu^{\perp})^TJ)\cap V_\S\nonumber\\
	&=\{M_\S\bfy |\bfy\in \Z_2^n, (M_\mu^{\perp})^T JM_\S\bfy=0 \}\nonumber\\
	&=\{M_\S\bfy |\bfy\in \Z_2^n, \tilde{M}_{\S,\mu}\bfy=0\} \nonumber\\
	&=\{M_\S\bfy | \bfy\in \ker(\tilde{M}_{\S,\mu})\}=M_\S \ker(\tilde{M}_{\S,\mu}).
	 \label{eq:LocalIsotropicStab}
	\end{align}
	To derive the fourth equality, note that 
	\begin{align}
	(M_\mu^{\perp})^T JM_\S&=\begin{pmatrix}
	\diag(\mu^\rmz)M_{\S}^\rmx+\diag(\mu^\rmx)M_{\S}^\rmz\\
	N_\mu M_{\S}^\rmz\\
	N_\mu M_{\S}^\rmx
	\end{pmatrix}\nonumber\\
	& =(M_{\S,\mu} (\scrU_1);N_\mu M_{\S}^\rmz; N_\mu M_{\S}^\rmx),
	\end{align}	
	which reduces to $\tilde{M}_{\S,\mu}$ after interchanging $N_\mu M_{\S}^\rmz$ and $N_\mu M_{\S}^\rmx$ and deleting rows of zeros. As an implication of \eqref{eq:LocalIsotropicStab}, we have $|\ker(\tilde{M}_{\S,\mu})|=|U_\mu|=|\L_\mu|$ and 
	\begin{equation}\label{eq:LocalSubgroupProof}
\L_\mu=\{S^\bfy|\bfy\in \ker(\tilde{M}_{\S,\mu})\}, 
	\end{equation}
which confirms \eqref{eq:LocalSubgroupStab}. Equation \eqref{eq:TestProjStab} follows from \eqref{eq:LocalSubgroupStab} and Lemma~\ref{lem:TestProjCanonical}. 
	
	Furthermore, Lemma~\ref{lem:TestProjCanonical} and \eqref{eq:StabState2}  imply that
	\begin{align}
	&\<\S_\bfw|P_\mu|\S_\bfw\>=\frac{1}{2^n|\L_\mu |}\sum_{S'\in \L_\mu }\sum_{\bfy\in \Z_2^n}(-1)^{\bfw\cdot\bfy}\tr(S'S^{\bfy})\nonumber\\
	&=\frac{1}{|\L_\mu |}\sum_{\bfy\in\Z_2^n, S^\bfy\in \L_\mu }(-1)^{\bfw\cdot\bfy}=\frac{1}{|\L_\mu |}\sum_{\bfy\in \Z_2^n, M_\S\bfy\in U_\mu}(-1)^{\bfw\cdot\bfy}\nonumber\\
	&=\frac{1}{|\L_\mu |}\sum_{\bfy\in \Z_2^n, \tilde{M}_{\S,\mu} \bfy=0}(-1)^{\bfw\cdot\bfy}=\frac{1}{|\L_\mu |}\sum_{\bfy\in\ker(\tilde{M}_{\S,\mu})}(-1)^{\bfw\cdot\bfy},\label{eq:DiagonalEleStab}
	\end{align}
	where the fourth equality follows from \eqref{eq:LocalIsotropicStab}.  The summation in \eqref{eq:DiagonalEleStab} is nonzero iff  $\bfw\cdot\bfy=0$ for all $\bfy \in \ker(\tilde{M}_{\S,\mu})$. This is the case iff $\bfw\in \rspan(\tilde{M}_{\S,\mu})$, in which  case we have
	\begin{equation}
	\sum_{\bfy\in\ker(\tilde{M}_{\S,\mu})}(-1)^{\bfw\cdot\bfy}=|\ker(\tilde{M}_{\S,\mu})|=|U_\mu|=|\L_\mu|,
	\end{equation}
	which implies \eqref{eq:StabDiagElement}. 
\end{proof}

\begin{proof}[Proof of Theorem~\ref{thm:TestProjGraph}]
	Let $V_\mu$ be the Lagrangian subspace associated with $\T_\mu$ and  $M_\mu:=(\diag(\mu^\rmx);\diag(\mu^\rmz))$; then $M_\mu$ is a basis matrix for $V_\mu$. Let $V_G$ be the Lagrangian subspace associated with the graph state $|G\>$; then we have $V_G=\{\tilde{A}\bfy |\bfy\in \Z_2^n\}$, where  $\tilde{A}=(\mathbf{1};A)$ is the canonical basis matrix for $V_G$. Let  $U_\mu=V_\mu\cap V_G$; then  $U_\mu$ is the isotropic subspace associated with the local subgroup $\L_\mu$. In addition,
	\begin{align}
	U_\mu&=V_\mu\cap V_G=V_\mu^{\perp}\cap V_G\nonumber\\
	&=\{\tilde{A}\bfy |\bfy\in \Z_2^n, M_\mu^T J\tilde{A}\bfy=0\}\nonumber\\
	&=\{\tilde{A}\bfy |\bfy\in \Z_2^n, A_\mu\bfy=0\}   \nonumber\\
	&=\{\tilde{A}\bfy | \bfy\in \ker(A_\mu)\}=\tilde{A}\ker(A_\mu),	
	\label{eq:LocalIsotropic}
	\end{align}
where $J$ is defined in \eqref{eq:SymplecticForm}, and the fourth equality follows from the fact that $M_\mu^T J\tilde{A}=A_\mu$. 	As an implication of \eqref{eq:LocalIsotropic}, we have $|\ker(A_\mu)|=|U_\mu|=|\L_\mu|$ and 
\begin{equation}
\L_\mu=\{S^\bfy|\bfy \in \ker(A_\mu)\},
\end{equation} 
which confirm \eqref{eq:LocalSubgroupGraph}. Equation \eqref{eq:TestProjGraph} follows from \eqref{eq:LocalSubgroupGraph} and Lemma~\ref{lem:TestProjCanonical}. 

Furthermore, Lemma~\ref{lem:TestProjCanonical} and \eqref{eq:GraphState2} imply that
	\begin{align}
	&\<G_\bfw|P_\mu|G_\bfw\>=\frac{1}{2^n|\L_\mu |}\sum_{S'\in \L_\mu }\sum_{\bfy\in \Z_2^n}(-1)^{\bfw\cdot\bfy}\tr(S'S^{\bfy})\nonumber\\
	&=\frac{1}{|\L_\mu |}\sum_{\bfy\in\Z_2^n, S^\bfy\in \L_\mu }(-1)^{\bfw\cdot\bfy}=\frac{1}{|\L_\mu |}\sum_{\bfy\in \Z_2^n, \tilde{A}\bfy\in U_\mu}(-1)^{\bfw\cdot\bfy}\nonumber\\
	&=\frac{1}{|\L_\mu |}\sum_{\bfy\in \Z_2^n, A_\mu \bfy=0}(-1)^{\bfw\cdot\bfy}=\frac{1}{|\L_\mu |}\sum_{\bfy\in\ker(A_\mu)}(-1)^{\bfw\cdot\bfy},\label{eq:DiagonalEle}
	\end{align}
	where the fourth equality follows from \eqref{eq:LocalIsotropic}.  The summation in \eqref{eq:DiagonalEle} is nonzero iff  $\bfw\cdot\bfy=0$ for all $\bfy \in \ker(A_\mu)$. This is the case iff $\bfw$ belongs to the row span of $A_\mu$, in which case
	\begin{equation}
	\sum_{\bfy\in\ker(A_\mu)}(-1)^{\bfw\cdot\bfy}=|\ker(A_\mu)|=|U_\mu|=|\L_\mu|,
	\end{equation}
which implies \eqref{eq:GraphDiagElement}. 
\end{proof}

\section{Proofs of Proposition~\ref{pro:eta_XZ} and Lemma~\ref{lemma:max-ind-iff-admissible}}\label{app:coloring}
\begin{proof}[Proof of Proposition~\ref{pro:eta_XZ}]Note that an admissible test projector based on $X$ and $Z$ measurements is automatically weakly admissible. To prove Proposition~\ref{pro:eta_XZ}
we need to prove that if a test projector $P_{\mu_{XZ}}$ based on $X$ and $Z$ measurements is inadmissible, then it is not weakly admissible either; in other words, there always exists a smaller test projector $P_{\nu_{XZ}} \leq  P_{\mu_{XZ}}$ with $\tr(P_{\nu_{XZ}})<  \tr(P_{\mu_{XZ}})$ that is also based on $X$ and $Z$ measurements. According to Corollary~\ref{cor:admissible}, the inadmissible measurement $\mu_{XZ}$ can be replaced by an incomplete measurement $\mu'_{XZ}$ on $k<n$ qubits. 
After the measurement $\mu'_{XZ}$, the reduced state on the remaining $n-k$ qubits is a stabilizer state of the form $U_{\mathrm{LC}}\ket{G'}$, where $G'$ is a graph of $n-k$ vertices, and $U_{\LC}$ is an outcome-dependent local Clifford unitary \cite{hein2004}. Crucially, when $\mu'_{XZ}$ consists of $X$ and $Z$ measurements, the unitary operator $U_{\LC}$  can only map $X$ to $Z$ and vice versa up to an overall  phase factor. Therefore, we can obtain a smaller test projector by performing suitable $X$ and $Z$ measurements on the remaining $n-k$ qubits, which implies that $P_{\mu_{XZ}}$  is not weakly admissible and confirms the proposition.
\end{proof}

\begin{proof}[Proof of Lemma~\ref{lemma:max-ind-iff-admissible}]	In one direction, suppose that $B$ is not maximal and let $B'$ be a larger independent set containing $B$. Then we have $P_{B'}\leq P_B$ and $\tr(P_{B'})<\tr(P_B)$ according to \eqref{eq:TestProjInd}, so $P_B$ is not admissible.

In the other direction, suppose that $B$ is maximal. If $P_B$ is not admissible, then it is not weakly admissible either by Proposition~\ref{pro:eta_XZ}. So there exists a subset $B'$ of $V$ such that $P_{B'}\leq P_B$ and $\tr(P_{B'})<\tr(P_B)$, that is,
\begin{equation}\label{eq:LBB}
\L_{B'} \supset \L_B=\langle\{S_j| j\in B\}\rangle,
\end{equation}
where $\L_B$ and $\L_{B'}$ are the local subgroups associated with the Pauli measurements determined by $B$ and $B'$, respectively, and  the equality follows from \eqref{eq:LocalSubgroupInd}.
Equation~\eqref{eq:LBB} implies that $B\subset B'$ by Lemma~\ref{lem:TwoTestBB}. 
Since $B$ is a maximal independent set by assumption, there must exist a vertex $j\in B$ that is adjacent to some vertex $k\in B'\setminus B$, which implies that $S_j\notin \L_{B'}$, in contradiction with \eqref{eq:LBB}. This contradiction completes the proof of Lemma~\ref{lemma:max-ind-iff-admissible}.
\end{proof}

\section{\label{app:nonconnected}Verification of graph states of nonconnected graphs}

Let $G$ be a disjoint union of (possibly empty) connected graphs $\{G_j\}_{j=1}^J$. 
Then $\ket{G} = \bigotimes_{j=1}^J \ket{G_j}$ is a graph state which is not genuinely multipartite entangled \cite{hein2004,hein2006}. Here we clarify the relations between optimal verification of $|G\>$ based on Pauli measurements and that of $|G_j\>$. It is worth pointing out that optimal protocols (with respect to the spectral gap or the number of measurement settings) can be constructed from canonical test projectors; cf. Sec.~\ref{sec:TestProjGraph}. 

\begin{lemma}\label{lem:TestProjCanProd}
	Every canonical test projector $P$ for  $\ket{G} = \bigotimes_{j=1}^J \ket{G_j}$ has a tensor-product form $P=\bigotimes_{j=1}^J P_j$, where $P_j$ is a canonical test projector for $|G_j\>$, and vice versa. 
\end{lemma}
\begin{proof}
	The stabilizer group $\S$ of $|G\>$ is a direct product of  the form $\S=\S_1\times \S_2\times \cdots\times \S_J$, where  $\S_j$ for $j=1,2,\ldots, J$ are the  stabilizer groups of $|G_j\>$, respectively. Suppose the test projector $P$ is associated with the Pauli measurement specified by the symplectic vector $\mu$; let $\T_\mu$ and $\bar{T}_\mu$ be the stabilizer group and signed stabilizer group associated with the Pauli measurement $\mu$. Then $\T_\mu$ has the form $\T_\mu =\T_1\times \T_2\times \cdots \times \T_J$, where $\T_j$ for $j=1,2,\ldots, J$ are stabilizer groups associated with certain Pauli measurements on $|G_j\>$, respectively. Let $\L_\mu=\S\cap \bar{\T}_\mu$ be the local subgroup associated with the Pauli measurement $\mu$. Thanks to Lemma~\ref{lem:StabIntersection}, $\L_\mu$ has the form $\L_\mu =\L_1\times \L_2\times \cdots\times \L_J$, where $\L_j=\S_j\cap \bar{\T}_j$ are local subgroups of $|G_j\>$. According to  \eqref{eq:TestProjCanonical} we have
	\begin{align}
	P=\frac{1}{|\L_\mu |}\sum_{S\in \L_\mu }S=\bigotimes_{j=1}^J\Biggl(\frac{1}{|\L_j |}\sum_{S_j\in \L_j }S_j\Biggr) =\bigotimes_{j=1}^J P_j,
	\end{align}
	where
	\begin{equation}
	P_j=\frac{1}{|\L_j |}\sum_{S_j\in \L_j }S_j
	\end{equation}
	are canonical test projectors for  $|G_j\>$.

	Conversely, suppose $P_j$ are  canonical test projectors for  $|G_j\>$ that are associated with the local subgroups $\L_j$ for $j=1,2,\ldots, J$. Then $P=\bigotimes_{j=1}^J P_j$ is a canonical test projector for $|G\>$ that is associated with the local subgroup $\L_1\times \L_2\times \cdots\times \L_J$. This observation completes the proof of Lemma~\ref{lem:TestProjCanProd}. 
\end{proof}

\begin{figure*}
    \centering
    \begin{subfigure}
        \centering
        \includegraphics[height=1.35in]{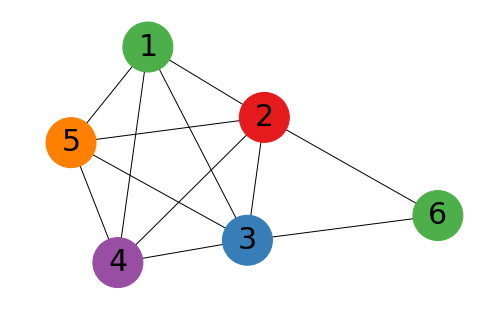}
    \end{subfigure}%
    ~
    \begin{subfigure}
        \centering
        \includegraphics[height=1.35in]{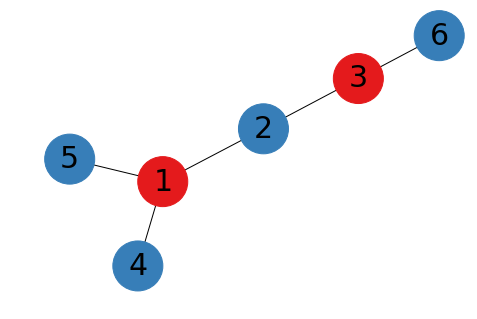}
    \end{subfigure}
    \caption{\label{fig:strict_eq}On the left is an example of a graph $G$ for which the inequalities $\tilde{\chi}_2(G) \le \tilde{\chi}_{XZ}(G) \le \chi(G)$ are strict. This six-vertex graph is No. 94 in the graph database \cite{mckay-database}, with $\tilde{\chi}_2(G) = \chi_{\mathrm{LC}}(G) = 2, \tilde{\chi}_{XZ}(G) = 4$, and $\chi(G) = 5$. The corresponding graph state can be verified by two measurement settings, say $YYYYYZ$ and $YZXYYY$ [$\tilde{\chi}_2(G) = 2$], or four settings based on $X$ and $Z$ measurements, say $ZZZZXX$, $ZZZXZX$, $XXXZZX$, and $ZZXZZZ$ [$\tilde{\chi}_{XZ}(G) = 4$]. The graph $G$ is equivalent to the 2-colorable graph on the right (No.12 in Table~\ref{tab:optimal})
under LC with respect to the vertices No. 1, 2, 6, and 3 in succession.}
\end{figure*}

\begin{proposition}\label{prop:nonconnected-cover}
	Suppose the graph $G$ is a disjoint union of (possibly empty) connected subgraphs $G_j$ for $j=1,2,\dots, J$. Then 
	\begin{align}
	\nu(G) &= \min_j \nu(G_j),  \label{eq:uGprod} \\
	\tilde{\chi}(G) &= \max_j \tilde{\chi}(G_j), \label{eq:local_cover_subgraph} \\
	\tilde{\chi}_2(G) &= \max_j \tilde{\chi}_2(G_j), \\
	\tilde{\chi}_{XZ}(G) &= \max_j \tilde{\chi}_{XZ}(G_j). \label{eq:local_cover_subgraphXZ} 
	\end{align}
\end{proposition}
Here $\nu(G)$ denotes the maximum spectral gap of verification operators of $|G\>$ that are based on Pauli measurements. Equation~\eqref{eq:uGprod} still applies if we consider the maximum spectral gap achievable by separable measurements, which follows from almost the same reasoning as the one presented below. $\tilde{\chi}(G)$ denotes the minimum number of Pauli measurement settings required to verify $|G\>$, while $\tilde{\chi}_{XZ}(G)$ and $\tilde{\chi}_2(G)$ denote the minimum numbers of settings based on $XZ$ measurements and two measurement settings for each party, respectively. Incidentally, the minimum number of measurement settings based on the coloring protocol \cite{ZhuH2019E} is equal to the chromatic number $\chi(G)$, which satisfies
\begin{align}
\chi(G) = \max_j \chi(G_j). 
\end{align}
So Proposition~\ref{prop:nonconnected-cover} may be seen as a generalization of this equation.

\begin{proof}
	Suppose $\Omega$ is an optimal verification operator for $|G\>$ with  $\nu(\Omega)=\nu(G)$ that can be realized by canonical test projectors. Let $\Omega_j$ be the reduced verification operator of $\Omega$ for $|G_j\>$. Then $\Omega_j$ can also be realized by canonical test projectors according to  Lemma~\ref{lem:TestProjCanProd}. Therefore,
	\begin{equation}\label{eq:nuGprodProof1}
	\nu(G)=\nu(\Omega)\leq \min_j \nu(\Omega_j)\leq \min_j\nu(G_j), 
	\end{equation}
	where the second inequality follows from Proposition~\ref{pro:ReducedVeriO}.

	Conversely, suppose $\Omega_j$ for $j=1,2\ldots, J$ are optimal verification operators of $G_j$ that are based on Pauli measurements, so that $\nu(\Omega_j)=\nu(G_j)$. Let $\Omega=\bigotimes_{j=1}^J \Omega_j$; then $\Omega$ is a verification operator of $|G\>$ that is based on Pauli measurements. Therefore,
	\begin{equation}\label{eq:nuGprodProof2}
	\nu(G)\geq \nu(\Omega)=\min_{1\leq j\leq J}\nu(\Omega_j)= \min_j \nu(G_j).
	\end{equation}
	Equations~\eqref{eq:nuGprodProof1} and \eqref{eq:nuGprodProof2} together imply \eqref{eq:uGprod}.

	To prove \eqref{eq:local_cover_subgraph}-\eqref{eq:local_cover_subgraphXZ}, suppose $|G\>$ can be verified by $m$ canonical test projectors $P_1, P_2, \ldots, P_m$; let $\Omega=\sum_k P_k/m$. Let $|\overline{G}_j\>= \bigotimes_{j'\neq j} |G_{j'}\>$,  $P_k^{(j)}=\<\overline{G}_j|P_k|\overline{G}_j\>$, and $\Omega_j=\sum_k P_k^{(j)}/m$; then $P_k^{(j)}$ for $k=1,2,\ldots, m$ are canonical test projectors for $|G_j\>$ according to Lemma~\ref{lem:TestProjCanProd}. Moreover, $|G_j\>$ can be verified by these canonical test projectors since $\nu(\Omega_j)\geq \nu(\Omega)\geq1/m$. 
	If each $P_k$ is based on $XZ$ measurements or two measurement settings for each party, then each $P_k^{(j)}$ has the same property. These observations imply that
	\begin{align}
	\tilde{\chi}(G) &\geq  \max_j \tilde{\chi}(G_j), \\
	\tilde{\chi}_2(G) &\geq  \max_j \tilde{\chi}_2(G_j), \\
	\tilde{\chi}_{XZ}(G) &\geq  \max_j \tilde{\chi}_{XZ}(G_j).
	\end{align}
	
	Conversely,  suppose $|G_j\>$ can be verified by the set of canonical test projectors $\{P_k^{(j)}\}_{k=1}^{m_j}$. Let $m=\max_j m_j$; then $|G\>$ can be  verified by the following canonical test projectors
	\begin{equation}
	P_k:=\bigotimes_{j=1}^J P_k^{(j)},\quad k=1,2,\ldots, m,
	\end{equation}
	where $P_k^{(j)}=\openone $ if $m_j<k\leq m$; cf.~\eqref{eq:TestOperatorJoint}. Let $\Omega=\sum_k P_k/m$, then $\nu(\Omega)\geq 1/m$ since these canonical test projectors commute with each other.
	If each $P_k^{(j)}$ is based on $XZ$ measurements or two measurement settings for each party, then each $P_k$ has the same property. These observations imply that
	\begin{align}
	\tilde{\chi}(G) &\leq  \max_j \tilde{\chi}(G_j), \\
	\tilde{\chi}_2(G) &\leq  \max_j \tilde{\chi}_2(G_j), \\
	\tilde{\chi}_{XZ}(G) &\leq  \max_j \tilde{\chi}_{XZ}(G_j),
	\end{align}
	which confirms \eqref{eq:local_cover_subgraph}-\eqref{eq:local_cover_subgraphXZ} in view of the opposite inequalities derived above.
\end{proof}

\section{General connected graphs up to seven vertices}\label{app:general-graphs}

When restricted to $X$ and $Z$ measurements, many results on the verification of graph states are not invariant under LC. Therefore, it is of interest to consider those connected graphs up to seven vertices not necessarily listed in Table~\ref{tab:optimal}. 
Here we briefly discuss optimal verification protocols (with respect to the spectral gap and the number of measurement settings) of graph states associated with these graphs. There are 996 such (non-isomorphic) graphs \cite{mckay-database}. Our calculation shows that  the maximum spectral gap achievable by  $X$ and $Z$ measurements is $1/2$ for all these graph states.
Since every such graph $G$ is equivalent under LC to some graph in Table~\ref{tab:optimal}, $\chi_{\mathrm{LC}}(G)$ is either 2 or 3. By Proposition~\ref{pro:LocalCoverNum} and the results presented in Table~\ref{tab:minimal-settings}, we have $\tilde{\chi}(G) = \tilde{\chi}_2(G) = \tilde{\chi}_{\mathrm{LC}}(G)$ for all these graphs. In contrast, $\tilde{\chi}_{XZ}(G)$ can take any value from $\tilde{\chi}_2(G)$ up to $\chi(G)$, so \eqref{eq:chi-equality} does not hold in general. A graph $G$ for which the inequalities $\tilde{\chi}_2(G) \le \tilde{\chi}_{XZ}(G) \le \chi(G)$ are strict is shown in Fig.~\ref{fig:strict_eq}.

\begin{proposition}
	For the complete graph of $n$ vertices, $\tilde{\chi}_{XZ}(G) = \chi(G) = n$.
\end{proposition}
\begin{proof}
	The equality $\chi(G) = n$	is an immediate corollary of the assumption that $G$ is a complete graph of $n$ vertices.
	According to the coloring protocol proposed in \cite{ZhuH2019E}, any graph state $|G\>$ can be verified by $\chi(G)$ settings based on $X$ and $Z$ measurements, which implies that $\tilde{\chi}_{XZ}(G) \leq \chi(G) = n$. To complete the proof, it remains to prove that $\tilde{\chi}_{XZ}(G)\geq n$. 
	
	It is straightforward to verify that the canonical test projector based on $Y^{\otimes n}$ has rank 2. Moreover, all canonical test projectors based on $X$ and $Z$ measurements have ranks either $2^{n-1}$ or $2^n$ [cf. \eqref{eq:GraphPmuRank}], so the corresponding local subgroups are either trivial or have  order 2, given that all canonical test projectors of the standard GHZ state based on $X$ and $Y$ measurements have ranks either $2^{n-1}$ or $2^n$ according to \cite{li2019_GHZ} (cf. Sec.~\ref{sec:examples}). So at least $n$ settings based on $X$ and $Z$ measurements are required to verify $|G\>$ in view of  Lemma~\ref{lem:ValidProtocol}, that is, $\tilde{\chi}_{XZ}(G)\geq n$, which completes the proof.
\end{proof}

\section{\label{app:thm:chi2}Proof of Theorem~\ref{thm:chi2}}
\begin{proof}[Proof of Theorem~\ref{thm:chi2}]
	If $\chi_{\LC}(G)=2$, then $G$ is nonempty, so $\tilde{\chi}(G)=2$ according to Propositions \ref{pro:LocalCoverNum} and \ref{pro:chi1}. 
	
	Conversely, if 	$\tilde{\chi}(G)=2$, then  $G$ is nonempty according to Proposition~\ref{pro:chi1}. In addition,  $|G\>$ can be verified by two settings based on Pauli measurement. First, suppose $G$ is connected, then $|G\>$ is genuinely multipartite entangled, so the Pauli operators measured for each qubit associated with the two settings must be different according to  Proposition~3 in \cite{ZhuH2019AdvL}. By a suitable local Clifford transformation $U$, the state $U|G\>$ can be verified by two  measurement settings in which one setting is based on $X$ measurements only, while the other setting is based on $Z$ measurements only. Up to a sign factor, each generator of the local subgroup associated with the first (second) setting is a product of some $X$ ($Z$) operators for individual qubits.
	Therefore, the stabilizer group of $U|G\>$ can be generated by a set of generators each of which is a product of local $X$ operators only or  a product of local $Z$ operators only. It follows that  $U|G\>$ is a CSS state, so $|G\>$ is equivalent to a graph state of a two-colorable graph according to \cite{chen2007}.
	In other words, $G$ is equivalent to a two-colorable graph under LC, which implies that $\chi_{\LC}(G)=2$ given that $G$ is nonempty. 
	
	If $G$ is not connected, then each connected component of $G$  is  equivalent to a two-colorable graph under LC, so the same holds for $G$.
	Therefore, we still have $\chi_{\LC}(G)=2$.  
	
	The second statement in Theorem~\ref{thm:chi2}  follows from the first statement and the fact that every stabilizer state is equivalent to a graph state under a local Clifford transformation \cite{schlingemann2001b,grassl2002,nest2004}.
\end{proof}

\section{\label{app:lem:LambdaRing}Proof of Lemma~\ref{lem:LambdaRing}}

\begin{proof}[Proof of Lemma~\ref{lem:LambdaRing}]
	Equation \eqref{eq:kappaRing} follows from \eqref{eq:LambdaRing} and Lemma~\ref{lem:LambdaKappaGraph}, so it suffices to prove 	\eqref{eq:LambdaRing}. 
	When $n$ is even, \eqref{eq:LambdaRing} is proved in \cite{markham2007}. When $n$ is odd, 
	Lemmas \ref{lem:LambdaKappaGraph} and \ref{lem:kappaUB} imply that 
	\begin{align}\label{eq:LambdaRingLB}
	\Lambda_\rmP(G)\geq 2^{n-\alpha(G)}=2^{-(n+1)/2}=2^{-\lfloor(n+1)/2\rfloor}, 
	\end{align}
	given that $\alpha(G)=(n-1)/2$. So  it remains to prove the opposite inequality $\Lambda_\rmP(G)\leq 2^{-(n+1)/2}$. 
	
	Suppose that $\ket{\varphi} = \ket{\varphi_1}\otimes\ket{\varphi_2}\otimes\cdots\otimes\ket{\varphi_n}$ is a tensor product of eigenstates of Pauli $X,Y$, or $Z$ such that $\Lambda_\rmP(G) = |\av{\varphi|G}|^2$. Then 
	\begin{equation}
	\Lambda_\rmP(G)=|\<\varphi|G\>|^2=\frac{1}{2}|\<\varphi'|\Psi'\>|^2\leq \frac{1}{2}\Lambda(|\Psi'\>), 
	\end{equation}
	where the kets $\ket{\varphi'} = \ket{\varphi_1}\otimes\ket{\varphi_2}\otimes\cdots\otimes\ket{\varphi_{n-1}}$ and $|\Psi'\>=\sqrt{2} \<\varphi_n|G\>$ denote $(n-1)$-qubit stabilizer states.
	If $|\varphi_n\>$ is an eigenstate of $Z$, then $|\Psi'\>$ is an $(n-1)$-qubit linear cluster state by the general set of rules in \cite{hein2004}. If $|\varphi_n\>$ is an eigenstate of $Y$, by  contrast, then $|\Psi'\>$ is equivalent to a ring cluster state. In both cases, we  have $\Lambda(|\Psi'\>)=2^{-(n-1)/2}$ \cite{markham2007}, which implies that $|\<\varphi|G\>|^2\leq 2^{-(n+1)/2}$. The same inequality holds if at least one of the tensor factors $|\varphi_j\>$ is an eigenstate of $Z$ or $Y$. It remains to consider the case in which every $|\varphi_j\>$ is an eigenstate of $X$, so that $|\varphi\>$ belongs to the graph basis associated with the empty graph. According to Lemma~\ref{lem:GraphFidelity}, then we have 
	\begin{equation}
	|\<\varphi|G\>|^2\leq 2^{-\rank(A)}=2^{-(n-1)}\leq 2^{-(n+1)/2},
	\end{equation}
	where $A$ is  the adjacency matrix of $G$. It follows that $\Lambda_\rmP(G)\leq 2^{-(n+1)/2}$, which implies \eqref{eq:LambdaRing} in view of 
	\eqref{eq:LambdaRingLB}. 
\end{proof}

\clearpage
\pagebreak

\section{Table of optimal verification protocols}

\begin{longtable*}{ccccccc}
	\caption{\label{tab:optimal}Optimal verification protocols
 for  connected graph states up to seven qubits. There are 45 equivalent classes with respect to LC and graph isomorphism and here the labeling follows from \cite{hein2004}. Graph states associated with an edge and star graphs (No. 3, 5, 9, and 20) are omitted since optimal protocols for these states have a simple description as discussed  in Sec.~\ref{sec:examples} (cf. Fig.~\ref{fig:GHZ}). 
 For each graph,  the optimal protocol is specified by  Pauli measurement settings shown in the fifth column together with the corresponding probabilities shown in the sixth column. The spectral gap $\nu(\Omega)$ of the verification operator $\Omega$ is $2/3$, which attains the upper bound presented  in Theorem~\ref{thm:optimal}. 
 For completeness, the table also shows a minimum coloring of each graph, the number $\#(\Omega)$ of measurement settings in the optimal protocol and the ranks of canonical test projectors. In addition, $\eta(G)$ denotes the total number of admissible test projectors (that is, the number of admissible Pauli measurements) for the graph state $|G\>$.
} \\ 
	\hhline{=======}
	No. & Graph $G$ & $\eta(G)$ & $\#(\Omega)$ & Setting & Probability & Rank \\ 
	\hline
	
	\multirow{5}{*}{2} & \multirow{5}{*}{\includegraphics[width=45mm, trim=0 0 0 0.5in]{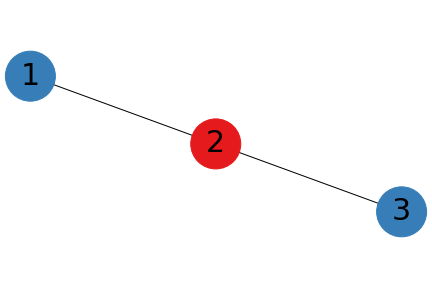}} & \multirow{5}{*}{5} & \multirow{5}{*}{5}
	& XZX & 1/3 & 2 \\
	& & & & ZXZ & 1/6 & 4 \\
	& & & & ZYY & 1/6 & 4 \\
	& & & & YXY & 1/6 & 4 \\
	& & & & YYZ & 1/6 & 4 \\
	\cline{5-7}
	
	\multirow{6}{*}{4} & \multirow{6}{*}{\includegraphics[width=45mm]{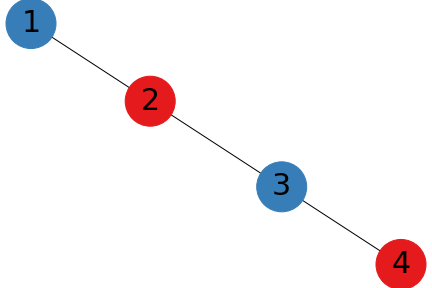}} & \multirow{6}{*}{9} & \multirow{6}{*}{6}
	& ZXZX & 1/6 & 4 \\
	& & & & XZXZ & 1/6 & 4 \\
	& & & & XZYY & 1/6 & 4 \\
	& & & & YYZX & 1/6 & 4 \\
	& & & & ZYYZ & 1/6 & 8 \\
	& & & & YXXY & 1/6 & 8 \\
	\cline{5-7}
	
	\multirow{10}{*}{6} & \multirow{10}{*}{\includegraphics[width=45mm]{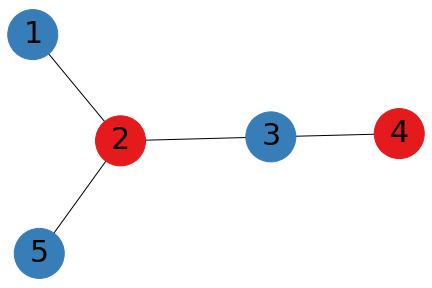}} & \multirow{10}{*}{15} & \multirow{10}{*}{10}
	& XZXZX & 1/6 & 4 \\
	& & & & XZYYX & 1/6 & 4 \\
	& & & & ZXZXZ & 1/12 & 8 \\
	& & & & ZYZXY & 1/12 & 8 \\
	& & & & YXZXY & 1/12 & 8 \\
	& & & & YYZXZ & 1/12 & 8 \\
	& & & & ZXXYY & 1/12 & 16 \\
	& & & & ZYXYZ & 1/12 & 16 \\
	& & & & YXYZZ & 1/12 & 16 \\
	& & & & YYYZY & 1/12 & 16 \\
	\cline{5-7}
	
	\multirow{6}{*}{7} & \multirow{6}{*}{\includegraphics[width=45mm, trim=0 0 0 1in]{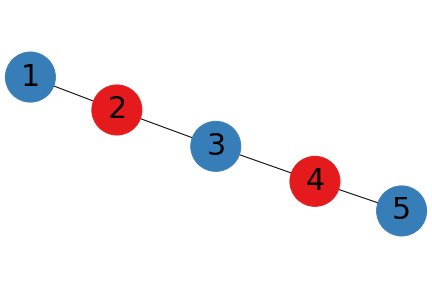}} & \multirow{6}{*}{17} & \multirow{6}{*}{6}
	& XZXZX & 1/6 & 4 \\
	& & & & ZXZXZ & 1/6 & 8 \\
	& & & & ZYYZX & 1/6 & 8 \\
	& & & & XZYYZ & 1/6 & 8 \\
	& & & & YYZYY & 1/6 & 8 \\
	& & & & YXXXY & 1/6 & 16 \\
	\cline{5-7}
	
	\multirow{6}{*}{8} & \multirow{6}{*}{\includegraphics[width=45mm]{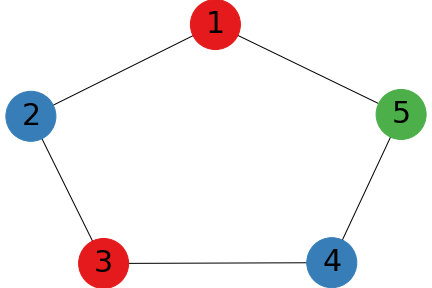}} & \multirow{6}{*}{21} & \multirow{6}{*}{6}
	& ZZXZX & 1/6 & 8 \\
	& & & & ZXZZX & 1/6 & 8 \\
	& & & & XZYYZ & 1/6 & 8 \\
	& & & & XXYYY & 1/6 & 8 \\
	& & & & YYZXZ & 1/6 & 8 \\
	& & & & YYXXY & 1/6 & 8 \\
	\cline{5-7}
	
	\multirow{18}{*}{10} & \multirow{18}{*}{\includegraphics[width=50mm]{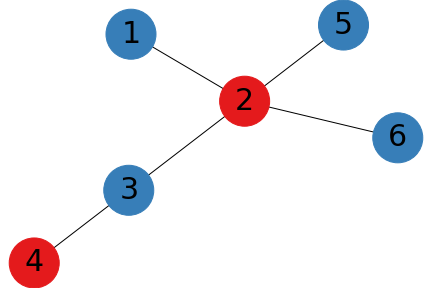}} & \multirow{18}{*}{27} & \multirow{18}{*}{18}
	& XZXZXX & 1/6 & 4 \\
	& & & & XZYYXX & 1/6 & 4 \\
	& & & & ZXZXZZ & 1/24 & 16 \\
	& & & & ZXZXYY & 1/24 & 16 \\
	& & & & ZYZXZY & 1/24 & 16 \\
	& & & & ZYZXYZ & 1/24 & 16 \\
	& & & & YXZXZY & 1/24 & 16 \\
	& & & & YXZXYZ & 1/24 & 16 \\
	& & & & YYZXZZ & 1/24 & 16 \\
	& & & & YYZXYY & 1/24 & 16 \\
	& & & & ZXXYZY & 1/24 & 32 \\
	& & & & ZXXYYZ & 1/24 & 32 \\
	& & & & ZYXYZZ & 1/24 & 32 \\
	& & & & ZYXYYY & 1/24 & 32 \\
	& & & & YXYZZZ & 1/24 & 32 \\
	& & & & YXYZYY & 1/24 & 32 \\
	& & & & YYYZZY & 1/24 & 32 \\
	& & & & YYYZYZ & 1/24 & 32 \\
	\cline{5-7}
	
	\multirow{12}{*}{11} & \multirow{12}{*}{\includegraphics[width=50mm]{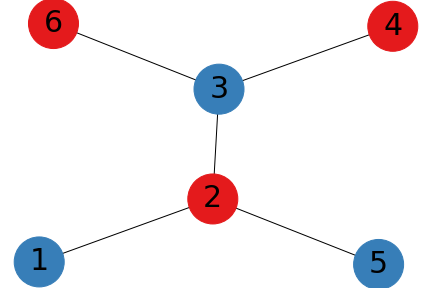}} & \multirow{12}{*}{25} & \multirow{12}{*}{12}
	& ZXZXZX & 1/12 & 8 \\
	& & & & ZYZXYX & 1/12 & 8 \\
	& & & & XZXZXZ & 1/12 & 8 \\
	& & & & XZXYXY & 1/12 & 8 \\
	& & & & XZYZXY & 1/12 & 8 \\
	& & & & XZYYXZ & 1/12 & 8 \\
	& & & & YXZXYX & 1/12 & 8 \\
	& & & & YYZXZX & 1/12 & 8 \\
	& & & & ZXXZYY & 1/12 & 32 \\
	& & & & ZYYYZY & 1/12 & 32 \\
	& & & & YXXYZZ & 1/12 & 32 \\
	& & & & YYYZYZ & 1/12 & 32 \\
	\cline{5-7}
	
	\multirow{11}{*}{12} & \multirow{11}{*}{\includegraphics[width=50mm]{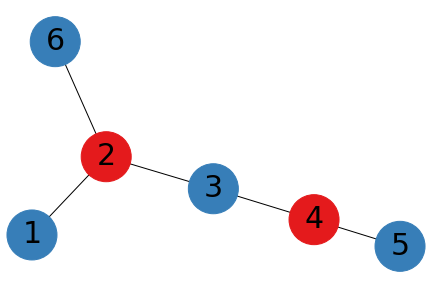}} & \multirow{11}{*}{29} & \multirow{11}{*}{11}
	& XZXZXX & 1/6 & 4 \\
	& & & & XZYXYX & 1/12 & 8 \\
	& & & & XZYYZX & 1/12 & 8 \\
	& & & & ZXZXZZ & 1/12 & 16 \\
	& & & & ZXYZXY & 1/12 & 16 \\
	& & & & ZYZXZY & 1/12 & 16 \\
	& & & & ZYYZXZ & 1/12 & 16 \\
	& & & & YXZYYY & 1/12 & 16 \\
	& & & & YYZYYZ & 1/12 & 16 \\
	& & & & YXXXYZ & 1/12 & 32 \\
	& & & & YYXYZY & 1/12 & 32 \\
	\cline{5-7}
	
	\multirow{10}{*}{13} & \multirow{10}{*}{\includegraphics[width=50mm]{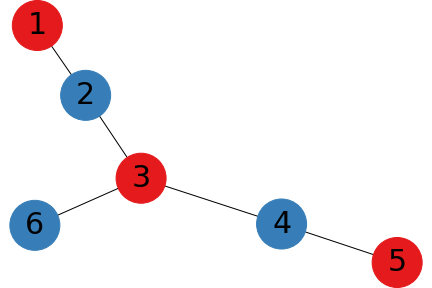}} & \multirow{10}{*}{27} & \multirow{10}{*}{10}
	& ZXZXZX & 1/6 & 8 \\
	& & & & XZXZXZ & 1/12 & 8 \\
	& & & & XZYZXY & 1/12 & 8 \\
	& & & & YYZYYX & 1/6 & 8 \\
	& & & & XZXXYY & 1/12 & 16 \\
	& & & & XZYXYZ & 1/12 & 16 \\
	& & & & YXXZXY & 1/12 & 16 \\
	& & & & YXYZXZ & 1/12 & 16 \\
	& & & & ZYXYZZ & 1/12 & 32 \\
	& & & & ZYYYZY & 1/12 & 32 \\
	\cline{5-7}
	
	\multirow{8}{*}{14} & \multirow{8}{*}{\includegraphics[width=50mm]{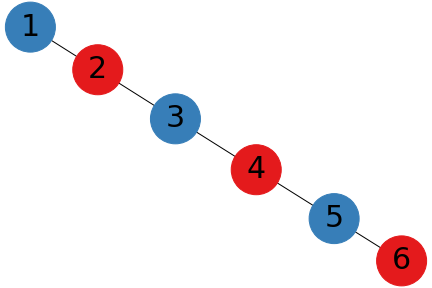}} & \multirow{8}{*}{31} & \multirow{8}{*}{8}
	& ZXZXZX & 1/6 & 8 \\
	& & & & XZXZXZ & 1/6 & 8 \\
	& & & & XZYYZX & 1/6 & 8 \\
	& & & & YXYZYY & 1/6 & 16 \\
	& & & & YYZYXY & 1/12 & 16 \\
	& & & & YYZYYZ & 1/12 & 16 \\
	& & & & ZYXXXY & 1/12 & 32 \\
	& & & & ZYXXYZ & 1/12 & 32 \\
	\cline{5-7}
	
	\multirow{6}{*}{15} & \multirow{6}{*}{\includegraphics[width=50mm,trim=0 0 0 1in]{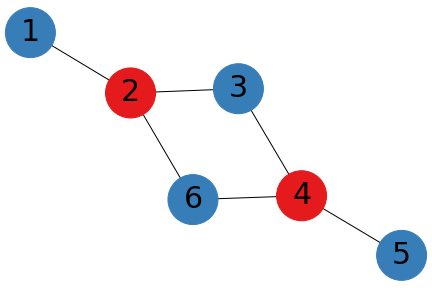}} & \multirow{6}{*}{33} & \multirow{6}{*}{6}
	& XZXZXX & 1/6 & 4 \\
	& & & & ZXZXZZ & 1/6 & 16 \\
	& & & & ZYZZXY & 1/6 & 16 \\
	& & & & XZYYZZ & 1/6 & 16 \\
	& & & & YXXXYX & 1/6 & 16 \\
	& & & & YYYYYY & 1/6 & 16 \\
	\cline{5-7}
	
	\multirow{8}{*}{16} & \multirow{8}{*}{\includegraphics[width=50mm, trim = 0 0 0 1in]{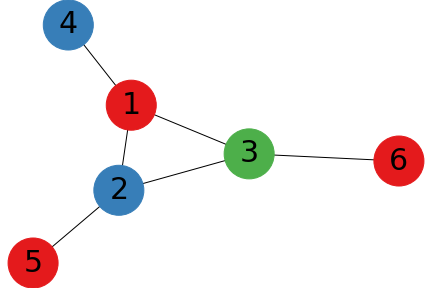}} & \multirow{8}{*}{35} & \multirow{8}{*}{8}
	& ZZXXXZ & 1/6 & 8 \\
	& & & & ZXZXZX & 1/6 & 8 \\
	& & & & XZZZXX & 1/6 & 8 \\
	& & & & XXXYYY & 1/12 & 16 \\
	& & & & XXYYYZ & 1/12 & 16 \\
	& & & & YYXZZY & 1/12 & 16 \\
	& & & & YYYZZZ & 1/12 & 16 \\
	& & & & YYYYYY & 1/6 & 32 \\
	\cline{5-7}
	
\pagebreak
	
	\multirow{14}{*}{17} & \multirow{14}{*}{\includegraphics[width=50mm]{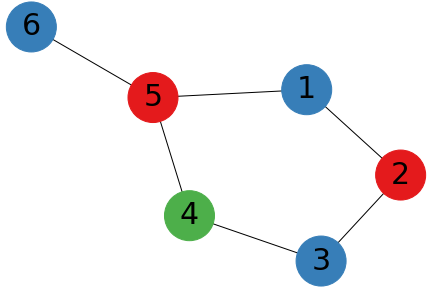}} & \multirow{14}{*}{37} & \multirow{14}{*}{14}
	& ZXZXZX & 1/8 & 8 \\
	& & & & XZZXZX & 1/12 & 8 \\
	& & & & XZXZZX & 1/12 & 8 \\
	& & & & ZZXZXZ & 1/24 & 16 \\
	& & & & ZYXYZX & 1/24 & 16 \\
	& & & & ZYYZYY & 1/8 & 16 \\
	& & & & XXYYXY & 1/8 & 16 \\
	& & & & YZXZYZ & 1/24 & 16 \\
	& & & & YXXYYZ & 1/12 & 16 \\
	& & & & YYXXXY & 1/24 & 16 \\
	& & & & YYYXXZ & 1/12 & 16 \\
	& & & & XYZZYZ & 1/24 & 32 \\
	& & & & YZZYXZ & 1/24 & 32 \\
	& & & & YZZYYY & 1/24 & 32 \\
	\cline{5-7}
	
	\multirow{12}{*}{18} & \multirow{12}{*}{\includegraphics[width=50mm]{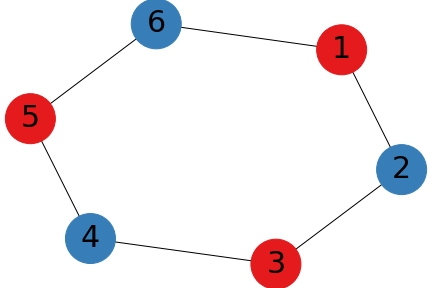}} & \multirow{12}{*}{43} & \multirow{12}{*}{12}
	& ZXZXZX & 2/9 & 8 \\
	& & & & XZXZXZ & 2/9 & 8 \\
	& & & & ZZXZYY & 1/27 & 16 \\
	& & & & ZXZZYY & 1/27 & 16 \\
	& & & & ZYXYZX & 1/27 & 16 \\
	& & & & XYXYXX & 1/27 & 16 \\
	& & & & XYYXYY & 2/27 & 16 \\
	& & & & YZYYZY & 2/27 & 16 \\
	& & & & YXYZXZ & 1/27 & 16 \\
	& & & & YXYXXX & 1/27 & 16 \\
	& & & & YYZYYZ & 2/27 & 16 \\
	& & & & YYYYYY & 1/9 & 16 \\
	\cline{5-7}
	
	\multirow{19}{*}{19} & \multirow{19}{*}{\includegraphics[width=50mm]{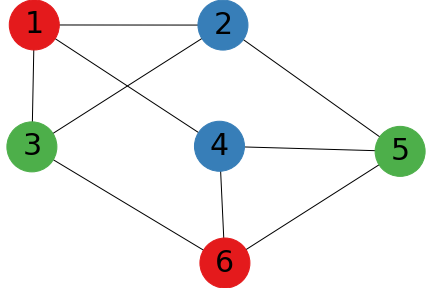}} & \multirow{19}{*}{63} & \multirow{19}{*}{19}
	& ZZZYYY & 16/129 & 16 \\
	& & & & ZZXZXZ & 8/129 & 16 \\
	& & & & ZZYYYX & 1/129 & 16 \\
	& & & & ZXZXZZ & 1/129 & 16 \\
	& & & & ZXYZYX & 4/43 & 16 \\
	& & & & ZYYXXY & 5/129 & 16 \\
	& & & & XZZZZX & 8/129 & 16 \\
	& & & & XXZXXZ & 8/129 & 16 \\
	& & & & XXXXXX & 7/129 & 16 \\
	& & & & XXYZYY & 5/129 & 16 \\
	& & & & XYXXYX & 5/129 & 16 \\
	& & & & XYYYZZ & 10/129 & 16 \\
	& & & & YZXXZY & 7/129 & 16 \\
	& & & & YZYZXZ & 1/43 & 16 \\
	& & & & YXXYZY & 10/129 & 16 \\
	& & & & YYZXYX & 4/129 & 16 \\
	& & & & YYZYXX & 2/43 & 16 \\
	& & & & YYYZZZ & 7/129 & 16 \\
	& & & & YYXXXZ & 2/43 & 32 \\
	\cline{5-7}
	
	\multirow{34}{*}{21} & \multirow{34}{*}{\includegraphics[width=50mm]{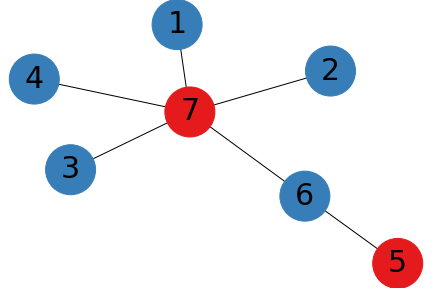}} & \multirow{34}{*}{51} & \multirow{34}{*}{34}
	& XXXXZXZ & 1/6 & 4 \\
	& & & & XXXXYYZ & 1/6 & 4 \\
	& & & & ZZZZXZX & 1/48 & 32 \\
	& & & & ZZZYXZY & 1/48 & 32 \\
	& & & & ZZYZXZY & 1/48 & 32 \\
	& & & & ZZYYXZX & 1/48 & 32 \\
	& & & & ZYZZXZY & 1/48 & 32 \\
	& & & & ZYZYXZX & 1/48 & 32 \\
	& & & & ZYYZXZX & 1/48 & 32 \\
	& & & & ZYYYXZY & 1/48 & 32 \\
	& & & & YZZZXZY & 1/48 & 32 \\
	& & & & YZZYXZX & 1/48 & 32 \\
	& & & & YZYZXZX & 1/48 & 32 \\
	& & & & YZYYXZY & 1/48 & 32 \\
	& & & & YYZZXZX & 1/48 & 32 \\
	& & & & YYZYXZY & 1/48 & 32 \\
	& & & & YYYZXZY & 1/48 & 32 \\
	& & & & YYYYXZX & 1/48 & 32 \\
	& & & & ZZZZZYY & 1/48 & 64 \\
	& & & & ZZZYZYX & 1/48 & 64 \\
	& & & & ZZYZYXX & 1/48 & 64 \\
	& & & & ZZYYYXY & 1/48 & 64 \\
	& & & & ZYZZYXX & 1/48 & 64 \\
	& & & & ZYZYZYY & 1/48 & 64 \\
	& & & & ZYYZZYY & 1/48 & 64 \\
	& & & & ZYYYYXX & 1/48 & 64 \\
	& & & & YZZZZYX & 1/48 & 64 \\
	& & & & YZZYYXY & 1/48 & 64 \\
	& & & & YZYZYXY & 1/48 & 64 \\
	& & & & YZYYZYX & 1/48 & 64 \\
	& & & & YYZZYXY & 1/48 & 64 \\
	& & & & YYZYZYX & 1/48 & 64 \\
	& & & & YYYZYXX & 1/48 & 64 \\
	& & & & YYYYZYY & 1/48 & 64 \\
	\cline{5-7}
	
	\multirow{20}{*}{22} & \multirow{20}{*}{\includegraphics[width=50mm]{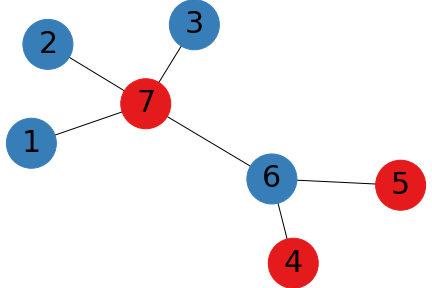}} & \multirow{20}{*}{45} & \multirow{20}{*}{20}
	& XXXZZXZ & 1/12 & 8 \\
	& & & & XXXZYYZ & 1/12 & 8 \\
	& & & & XXXYZYZ & 1/12 & 8 \\
	& & & & XXXYYXZ & 1/12 & 8 \\
	& & & & ZZZXXZX & 1/24 & 16 \\
	& & & & ZZYXXZY & 1/24 & 16 \\
	& & & & ZYZXXZY & 1/24 & 16 \\
	& & & & ZYYXXZX & 1/24 & 16 \\
	& & & & YZZXXZY & 1/24 & 16 \\
	& & & & YZYXXZX & 1/24 & 16 \\
	& & & & YYZXXZX & 1/24 & 16 \\
	& & & & YYYXXZY & 1/24 & 16 \\
	& & & & ZZZZZYY & 1/24 & 64 \\
	& & & & ZZYZYXX & 1/24 & 64 \\
	& & & & ZYZZYXX & 1/24 & 64 \\
	& & & & ZYYZZYY & 1/24 & 64 \\
	& & & & YZZYZXX & 1/24 & 64 \\
	& & & & YZYYYYY & 1/24 & 64 \\
	& & & & YYZYYYY & 1/24 & 64 \\
	& & & & YYYYZXX & 1/24 & 64 \\
	\cline{5-7}
	
	\multirow{19}{*}{23} & \multirow{19}{*}{\includegraphics[width=50mm]{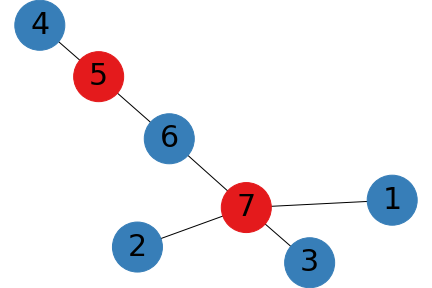}} & \multirow{19}{*}{53} & \multirow{19}{*}{19}
	& XXXXZXZ & 1/6 & 4 \\
	& & & & XXXZYYZ & 1/12 & 8 \\
	& & & & XXXYXYZ & 1/12 & 8 \\
	& & & & ZZZZXZX & 1/24 & 32 \\
	& & & & ZZZXZYY & 1/24 & 32 \\
	& & & & ZZYXZYX & 1/24 & 32 \\
	& & & & ZZYYYZY & 1/24 & 32 \\
	& & & & ZYZXZYX & 1/24 & 32 \\
	& & & & ZYZYYZY & 1/24 & 32 \\
	& & & & ZYYZXZX & 1/24 & 32 \\
	& & & & YZZYYZY & 1/24 & 32 \\
	& & & & YZYYYZX & 1/24 & 32 \\
	& & & & YYZZXZX & 1/24 & 32 \\
	& & & & YYZXZYY & 1/24 & 32 \\
	& & & & YYYZXZY & 1/24 & 32 \\
	& & & & ZYYZYXY & 1/24 & 64 \\
	& & & & YZZYXXX & 1/24 & 64 \\
	& & & & YZYYXXY & 1/24 & 64 \\
	& & & & YYYZYXX & 1/24 & 64 \\
	\cline{5-7}
	
	\multirow{11}{*}{24} & \multirow{11}{*}{\includegraphics[width=50mm,trim=0 0 0 1in]{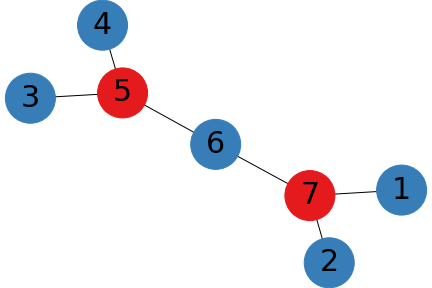}} & \multirow{11}{*}{49} & \multirow{11}{*}{11}
	& XXXXZXZ & 1/6 & 4 \\
	& & & & ZZXXZYY & 1/12 & 16 \\
	& & & & XXZZYYZ & 1/12 & 16 \\
	& & & & XXYZXYZ & 1/12 & 16 \\
	& & & & YZXXZYX & 1/12 & 16 \\
	& & & & ZZZZXZX & 1/12 & 32 \\
	& & & & ZYYYXZY & 1/12 & 32 \\
	& & & & YZYZYZY & 1/12 & 32 \\
	& & & & YYZYYZX & 1/12 & 32 \\
	& & & & ZYYYYXX & 1/12 & 64 \\
	& & & & YYZYXXY & 1/12 & 64 \\
	\cline{5-7}
	
	\multirow{18}{*}{25} & \multirow{18}{*}{\includegraphics[width=50mm,trim=0 0 0 4.5in]{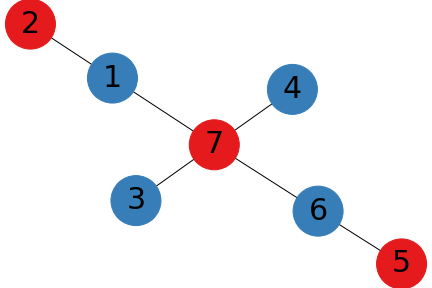}} & \multirow{18}{*}{45} & \multirow{18}{*}{18}
	& XZXXZXZ & 1/6 & 8 \\
	& & & & YYXXYYZ & 1/6 & 8 \\
	& & & & ZXZZXZX & 1/24 & 16 \\
	& & & & ZXZYXZY & 1/24 & 16 \\
	& & & & ZXYZXZY & 1/24 & 16 \\
	& & & & ZXYYXZX & 1/24 & 16 \\
	& & & & ZXZZZYY & 1/24 & 32 \\
	& & & & ZXZYYXX & 1/24 & 32 \\
	& & & & ZXYZYXX & 1/24 & 32 \\
	& & & & ZXYYYXY & 1/24 & 32 \\
	& & & & XYZZXZY & 1/24 & 32 \\
	& & & & XYZYXZX & 1/24 & 32 \\
	& & & & XYYZXZX & 1/24 & 32 \\
	& & & & XYYYXZY & 1/24 & 32 \\
	& & & & YZZZZYX & 1/24 & 64 \\
	& & & & YZZYZYY & 1/24 & 64 \\
	& & & & YZYZZYY & 1/24 & 64 \\
	& & & & YZYYYXX & 1/24 & 64 \\
	\cline{5-7}
	
	\multirow{12}{*}{26} & \multirow{12}{*}{\includegraphics[width=50mm]{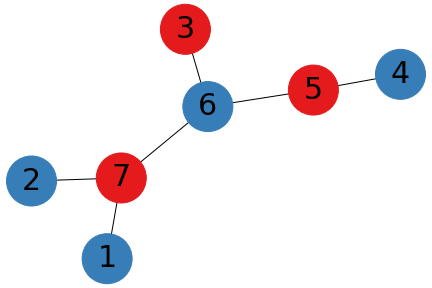}} & \multirow{12}{*}{45} & \multirow{12}{*}{12}
	& XXZXZXZ & 1/12 & 8 \\
	& & & & XXYXZYZ & 1/12 & 8 \\
	& & & & ZZXZXZX & 1/12 & 16 \\
	& & & & ZYXYYZY & 1/12 & 16 \\
	& & & & XXZYXYZ & 1/12 & 16 \\
	& & & & XXYZYXZ & 1/12 & 16 \\
	& & & & YZXYYZY & 1/12 & 16 \\
	& & & & YYXZXZX & 1/12 & 16 \\
	& & & & ZYYXZXX & 1/12 & 32 \\
	& & & & YZZXZYX & 1/12 & 32 \\
	& & & & ZZZZYXY & 1/12 & 64 \\
	& & & & YYYYXYY & 1/12 & 64 \\
	\cline{5-7}
	
	\multirow{10}{*}{27} & \multirow{10}{*}{\includegraphics[width=50mm]{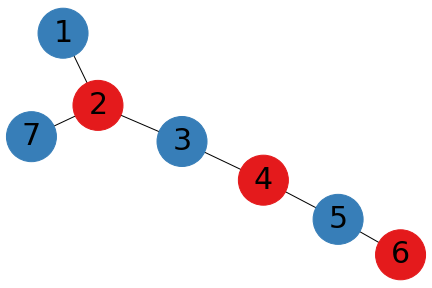}} & \multirow{10}{*}{53} & \multirow{10}{*}{10}
	& XZXZXZX & 1/6 & 8 \\
	& & & & ZXZXZXZ & 1/12 & 16 \\
	& & & & ZYZXZXY & 1/12 & 16 \\
	& & & & XZYXYZX & 1/6 & 16 \\
	& & & & ZXXYZXY & 1/12 & 32 \\
	& & & & ZYYZYYZ & 1/12 & 32 \\
	& & & & YXZYXYY & 1/12 & 32 \\
	& & & & YXYZYYZ & 1/12 & 32 \\
	& & & & YYZYXYZ & 1/12 & 32 \\
	& & & & YYXYZXY & 1/12 & 32 \\
	\cline{5-7}
	
	\multirow{11}{*}{28} & \multirow{11}{*}{\includegraphics[width=50mm]{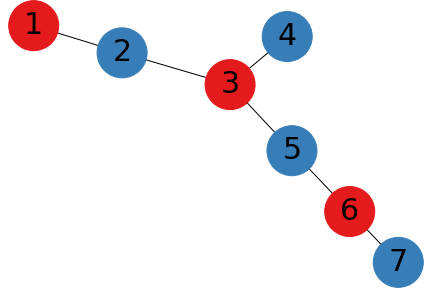}} & \multirow{11}{*}{51} & \multirow{11}{*}{11}
	& ZXZXXZX & 1/6 & 8 \\
	& & & & XZXZZXZ & 1/12 & 16 \\
	& & & & XZXYYZX & 1/12 & 16 \\
	& & & & XZYYZYY & 1/12 & 16 \\
	& & & & YYZXYXY & 1/12 & 16 \\
	& & & & YYZXYYZ & 1/12 & 16 \\
	& & & & ZYXZYZX & 1/12 & 32 \\
	& & & & ZYYZZYY & 1/12 & 32 \\
	& & & & XZYZXYZ & 1/12 & 32 \\
	& & & & YXXYZXZ & 1/12 & 32 \\
	& & & & YXYYXXY & 1/12 & 64 \\
	\cline{5-7}
	
	\multirow{10}{*}{29} & \multirow{10}{*}{\includegraphics[width=50mm]{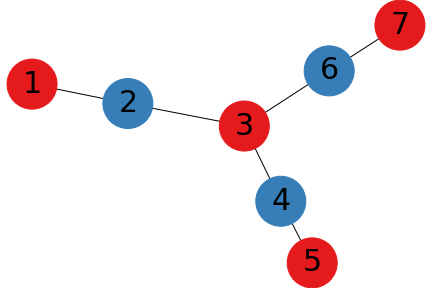}} & \multirow{10}{*}{53} & \multirow{10}{*}{10}
	& XZXZXZX & 1/12 & 8 \\
	& & & & ZXZXZXZ & 1/6 & 16 \\
	& & & & ZYYZXZX & 1/12 & 16 \\
	& & & & XZYZXYZ & 1/12 & 16 \\
	& & & & XZYYZZX & 1/12 & 16 \\
	& & & & YYZYYYY & 1/6 & 16 \\
	& & & & ZYXXYZX & 1/12 & 32 \\
	& & & & XZXYZXY & 1/12 & 32 \\
	& & & & YXXZXYZ & 1/12 & 32 \\
	& & & & YXYXYXY & 1/12 & 64 \\
	\cline{5-7}
	
\pagebreak	

	\multirow{12}{*}{30} & \multirow{12}{*}{\includegraphics[width=50mm]{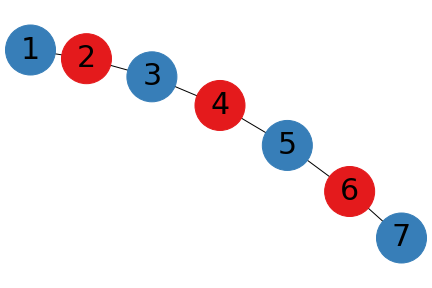}} & \multirow{12}{*}{57} & \multirow{12}{*}{12}
	& XZXZXZX & 1/6 & 8 \\
	& & & & ZXZXZXZ & 1/8 & 16 \\
	& & & & ZXZYYZX & 1/24 & 16 \\
	& & & & XZYXYZX & 1/24 & 16 \\
	& & & & XZYYZYY & 1/8 & 16 \\
	& & & & YYZXZYY & 1/24 & 16 \\
	& & & & ZYXXYZX & 1/12 & 32 \\
	& & & & ZYYZYXY & 1/24 & 32 \\
	& & & & YXXYZXZ & 1/24 & 32 \\
	& & & & YXYZYXY & 1/8 & 32 \\
	& & & & YYZYXYZ & 1/8 & 32 \\
	& & & & ZYXXXYZ & 1/24 & 64 \\
	\cline{5-7}
	
	\multirow{11}{*}{31} & \multirow{11}{*}{\includegraphics[width=50mm]{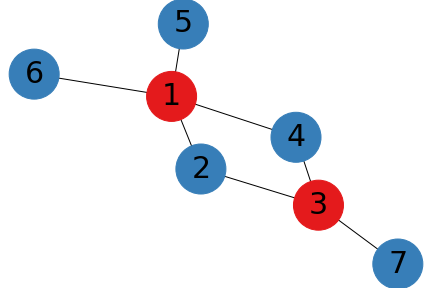}} & \multirow{11}{*}{57} & \multirow{11}{*}{11}
	& ZXZXXXX & 1/6 & 4 \\
	& & & & ZZYYXXZ & 1/12 & 16 \\
	& & & & ZYXZXXY & 1/12 & 16 \\
	& & & & XZZYZYX & 1/12 & 32 \\
	& & & & XZXZZZZ & 1/12 & 32 \\
	& & & & XXXXYZY & 1/12 & 32 \\
	& & & & XYYYYYY & 1/12 & 32 \\
	& & & & YZYZZYY & 1/12 & 32 \\
	& & & & YXYXZZZ & 1/12 & 32 \\
	& & & & YYZZYYX & 1/12 & 32 \\
	& & & & YYXYYZZ & 1/12 & 32 \\
	\cline{5-7}
	
	\multirow{12}{*}{32} & \multirow{12}{*}{\includegraphics[width=50mm]{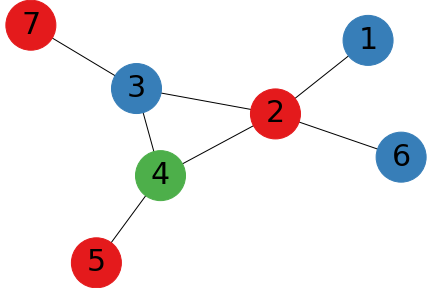}} & \multirow{12}{*}{61} & \multirow{12}{*}{12}
	& XXZXZXZ & 1/6 & 8 \\
	& & & & XXXZXZZ & 1/6 & 8 \\
	& & & & ZZXXZZX & 1/12 & 16 \\
	& & & & ZYXXZZY & 1/12 & 16 \\
	& & & & ZZZZYYY & 1/12 & 32 \\
	& & & & ZYZYXYX & 1/24 & 32 \\
	& & & & ZYYYXXX & 1/24 & 32 \\
	& & & & YZYYXXX & 1/12 & 32 \\
	& & & & YYZZYYY & 1/24 & 32 \\
	& & & & YYYZYXY & 1/24 & 32 \\
	& & & & YZYYYYY & 1/12 & 64 \\
	& & & & YYYYYYX & 1/12 & 64 \\
	\cline{5-7}
	
	\multirow{17}{*}{33} & \multirow{17}{*}{\includegraphics[width=50mm]{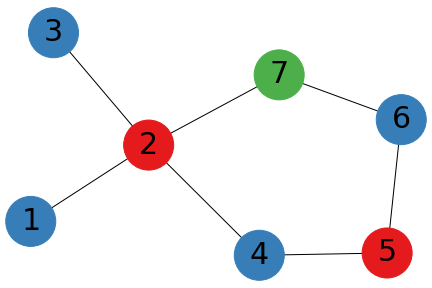}} & \multirow{17}{*}{65} & \multirow{17}{*}{17}
	& XXZZXZX & 1/8 & 8 \\
	& & & & XXZXZZX & 1/8 & 8 \\
	& & & & XXZZYXY & 1/24 & 16 \\
	& & & & XXZYXYZ & 1/24 & 16 \\
	& & & & ZZXXYYY & 1/12 & 32 \\
	& & & & ZZYYZXZ & 1/24 & 32 \\
	& & & & ZZYYXXY & 1/24 & 32 \\
	& & & & ZYXYZXZ & 1/12 & 32 \\
	& & & & ZYYZZXZ & 1/24 & 32 \\
	& & & & ZYYZYYZ & 1/24 & 32 \\
	& & & & YZXXXYY & 1/24 & 32 \\
	& & & & YZYZZXZ & 1/24 & 32 \\
	& & & & YZYXYYY & 1/24 & 32 \\
	& & & & YYXYYYX & 1/12 & 32 \\
	& & & & YYYZXZY & 1/24 & 32 \\
	& & & & YYYYXXY & 1/24 & 32 \\
	& & & & YZXXYZZ & 1/24 & 64 \\
	\cline{5-7}
	
	\multirow{11}{*}{34} & \multirow{11}{*}{\includegraphics[width=50mm]{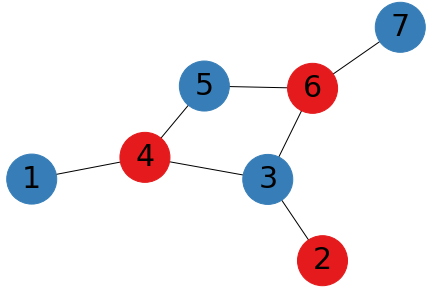}} & \multirow{11}{*}{59} & \multirow{11}{*}{11}
	& XZXZXZX & 1/6 & 8 \\
	& & & & ZXZXZXZ & 1/6 & 16 \\
	& & & & ZXZYYZX & 1/12 & 16 \\
	& & & & XXZZYYZ & 1/12 & 16 \\
	& & & & ZZYYZZX & 1/24 & 32 \\
	& & & & ZYYYXYZ & 1/24 & 32 \\
	& & & & XZYZZXY & 1/24 & 32 \\
	& & & & XZYZZYZ & 1/24 & 32 \\
	& & & & YZYXZZX & 1/24 & 32 \\
	& & & & YYXYYYY & 1/6 & 32 \\
	& & & & YYYXXXY & 1/8 & 32 \\
	\cline{5-7}
	
\pagebreak	
	
	\multirow{10}{*}{35} & \multirow{10}{*}{\includegraphics[width=50mm]{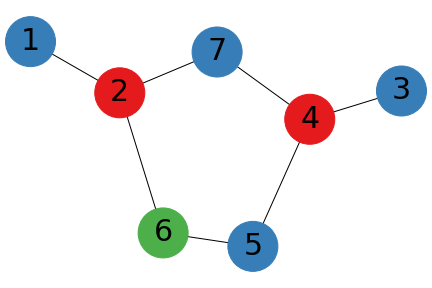}} & \multirow{10}{*}{67} & \multirow{10}{*}{10}
	& XXZZXZX & 1/6 & 8 \\
	& & & & XXZXZZX & 1/12 & 8 \\
	& & & & ZXZXZXZ & 1/12 & 16 \\
	& & & & XZYZXZY & 1/12 & 16 \\
	& & & & ZZXXYYY & 1/6 & 32 \\
	& & & & ZYXYZXZ & 1/12 & 32 \\
	& & & & YZYYZYZ & 1/12 & 32 \\
	& & & & YYXYXYY & 1/12 & 32 \\
	& & & & YYYZYXZ & 1/12 & 32 \\
	& & & & YYYYYXX & 1/12 & 32 \\
	\cline{5-7}
	
	\multirow{12}{*}{36} & \multirow{12}{*}{\includegraphics[width=50mm]{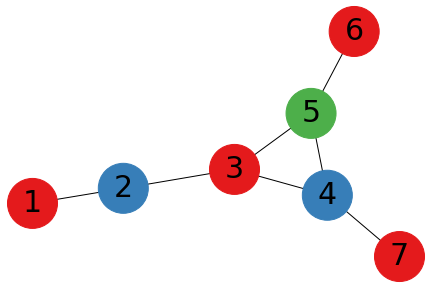}} & \multirow{12}{*}{61} & \multirow{12}{*}{12}
	& XZXZZXX & 1/6 & 8 \\
	& & & & ZXZZXZX & 1/12 & 16 \\
	& & & & ZXZXZXZ & 1/12 & 16 \\
	& & & & XZYXXYY & 1/8 & 16 \\
	& & & & XZYYYZZ & 1/24 & 16 \\
	& & & & YYZZXZX & 1/12 & 16 \\
	& & & & YYZYZXY & 1/12 & 16 \\
	& & & & ZYXYYZZ & 1/8 & 32 \\
	& & & & YXXXXYY & 1/24 & 32 \\
	& & & & ZYYYYYY & 1/24 & 64 \\
	& & & & YXYXYYZ & 1/12 & 64 \\
	& & & & YXYYYYY & 1/24 & 64 \\
	\cline{5-7}
	
	\multirow{15}{*}{37} & \multirow{15}{*}{\includegraphics[width=50mm]{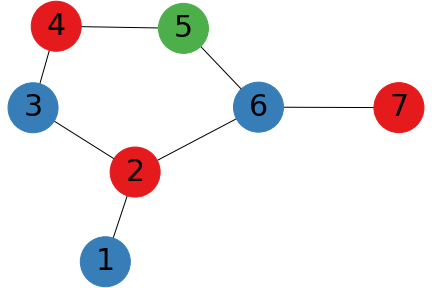}} & \multirow{15}{*}{65} & \multirow{15}{*}{15}
	& XZXZXZX & 1/6 & 8 \\
	& & & & ZXZXZZX & 1/24 & 16 \\
	& & & & ZYYZXZX & 5/72 & 16 \\
	& & & & XZZXZXZ & 5/72 & 16 \\
	& & & & XZXZYYZ & 1/24 & 16 \\
	& & & & XZYYZXZ & 1/18 & 16 \\
	& & & & YYZXZZX & 1/18 & 16 \\
	& & & & ZXXYYXY & 1/8 & 32 \\
	& & & & ZYZXZXY & 1/24 & 32 \\
	& & & & ZYYXXYZ & 1/18 & 32 \\
	& & & & YXZXZYZ & 5/72 & 32 \\
	& & & & YXYYXYY & 1/24 & 32 \\
	& & & & YYYYYXZ & 1/24 & 32 \\
	& & & & YYYYYYY & 5/72 & 32 \\
	& & & & YXZZYYY & 1/18 & 64 \\
	\cline{5-7}
	
	\multirow{15}{*}{38} & \multirow{15}{*}{\includegraphics[width=50mm]{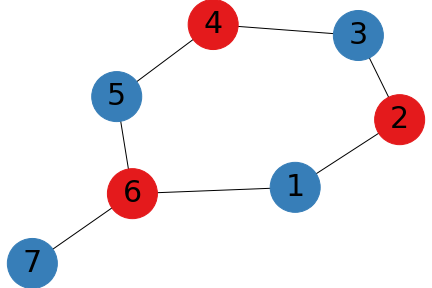}} & \multirow{15}{*}{75} & \multirow{15}{*}{15}
	& XZXZXZX & 4/21 & 8 \\
	& & & & ZXZXZXZ & 1/6 & 16 \\
	& & & & ZXZXZYY & 1/42 & 16 \\
	& & & & XZYXYZX & 1/21 & 16 \\
	& & & & YXYZXZX & 1/21 & 16 \\
	& & & & YYZYYZX & 1/21 & 16 \\
	& & & & ZXZYXXY & 1/21 & 32 \\
	& & & & ZYXYZYY & 1/21 & 32 \\
	& & & & ZYYZYXY & 1/21 & 32 \\
	& & & & XYXXXYZ & 1/21 & 32 \\
	& & & & YZYYZXY & 1/14 & 32 \\
	& & & & YZYYZYZ & 1/42 & 32 \\
	& & & & YXXXYYY & 1/21 & 32 \\
	& & & & YYYYYYZ & 2/21 & 32 \\
	& & & & XYZZYYY & 1/21 & 64 \\
	\cline{5-7}
	
	\multirow{13}{*}{39} & \multirow{13}{*}{\includegraphics[width=50mm]{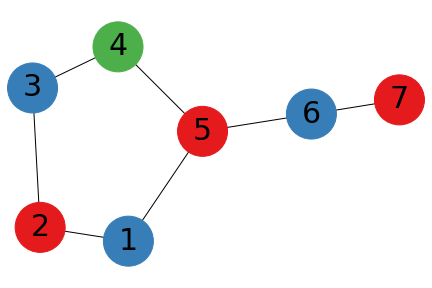}} & \multirow{13}{*}{67} & \multirow{13}{*}{13}
	& ZZXZXZX & 1/6 & 16 \\
	& & & & ZXZXZYY & 1/12 & 16 \\
	& & & & XZZXZXZ & 1/12 & 16 \\
	& & & & XZYYZYY & 1/24 & 16 \\
	& & & & XXYYYZX & 1/8 & 16 \\
	& & & & YZXZYZX & 1/24 & 16 \\
	& & & & YYZXZXZ & 1/12 & 16 \\
	& & & & YYZXZYY & 1/24 & 16 \\
	& & & & ZYYZYXY & 1/12 & 32 \\
	& & & & XYYYYXY & 1/24 & 32 \\
	& & & & YXXYXYZ & 1/8 & 32 \\
	& & & & YYYXYXY & 1/24 & 32 \\
	& & & & XYZZXYZ & 1/24 & 64 \\
	\cline{5-7}
	
	\multirow{17}{*}{40} & \multirow{17}{*}{\includegraphics[width=50mm, trim = 0 0 0 1in]{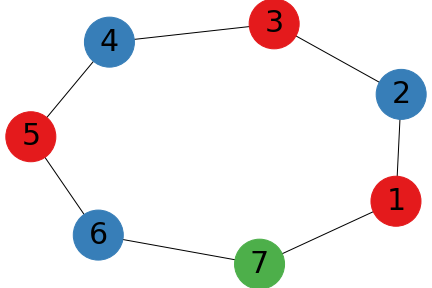}} & \multirow{17}{*}{85} & \multirow{17}{*}{17}
	& ZZXZXZX & 2/15 & 16 \\
	& & & & ZXZZXZX & 1/30 & 16 \\
	& & & & ZXZYYZX & 1/30 & 16 \\
	& & & & XZZXZXZ & 1/30 & 16 \\
	& & & & XZXZZXZ & 1/30 & 16 \\
	& & & & YYZXZXZ & 1/6 & 16 \\
	& & & & ZYXXYZX & 1/30 & 32 \\
	& & & & ZYYZZYY & 1/15 & 32 \\
	& & & & ZYYZYXY & 1/30 & 32 \\
	& & & & XZZYXYZ & 1/15 & 32 \\
	& & & & XXXYYYX & 1/10 & 32 \\
	& & & & XXYYXXY & 1/15 & 32 \\
	& & & & XYXYYYY & 1/30 & 32 \\
	& & & & YZYXYZY & 1/30 & 32 \\
	& & & & YZYYZZY & 1/30 & 32 \\
	& & & & YXYZXZZ & 1/30 & 32 \\
	& & & & YXYXYYY & 1/15 & 32 \\
	\cline{5-7}
	
	\multirow{17}{*}{41} & \multirow{17}{*}{\includegraphics[width=50mm]{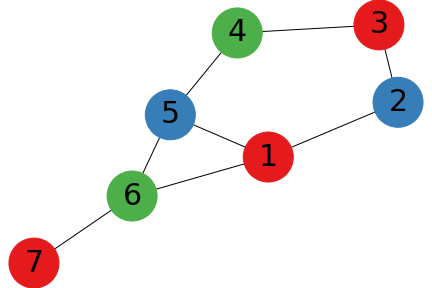}} & \multirow{17}{*}{77} & \multirow{17}{*}{17}
	& ZZXZXZX & 2/17 & 16 \\
	& & & & ZXZXZXZ & 7/51 & 16 \\
	& & & & XXYYYZX & 5/34 & 16 \\
	& & & & XYYYXZX & 1/34 & 16 \\
	& & & & YZXZYXY & 5/34 & 16 \\
	& & & & YYZXZZX & 1/102 & 16 \\
	& & & & YYYXXZX & 1/34 & 16 \\
	& & & & ZZYYZYY & 1/102 & 32 \\
	& & & & ZYYZZYY & 7/102 & 32 \\
	& & & & XYZXZYZ & 11/102 & 32 \\
	& & & & XYZYXYZ & 1/51 & 32 \\
	& & & & XYXYXYY & 1/34 & 32 \\
	& & & & YZZYXXY & 1/51 & 32 \\
	& & & & YZZYXYZ & 2/51 & 32 \\
	& & & & YXYXXXZ & 1/34 & 32 \\
	& & & & YXYXXYY & 1/51 & 32 \\
	& & & & YYXYYYY & 2/51 & 32 \\
	\cline{5-7}
	
	\multirow{6}{*}{42} & \multirow{6}{*}{\includegraphics[width=50mm, trim=0 0 0 1in]{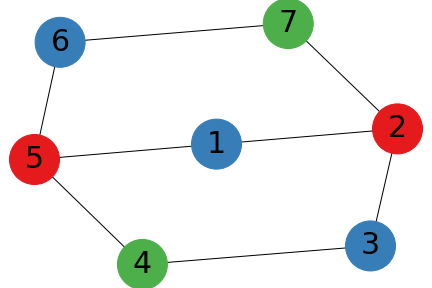}} & \multirow{6}{*}{87} & \multirow{6}{*}{6}
	& ZZXZXZX & 1/6 & 16 \\
	& & & & ZYYZXZX & 1/6 & 16 \\
	& & & & XZZXZXZ & 1/6 & 16 \\
	& & & & XYZXZYZ & 1/6 & 16 \\
	& & & & YXXYYXY & 1/6 & 16 \\
	& & & & YXYYYYY & 1/6 & 16 \\
	\cline{5-7}
	
	\multirow{13}{*}{43} & \multirow{13}{*}{\includegraphics[width=50mm]{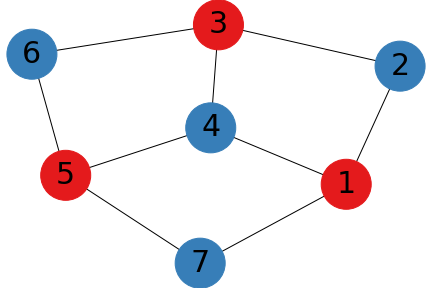}} & \multirow{13}{*}{101} & \multirow{13}{*}{13}
	& ZXZXZXX & 1/4 & 8 \\
	& & & & XZXZXZZ & 2/9 & 16 \\
	& & & & YYYYYYY & 1/9 & 16 \\
	& & & & ZZYXXXY & 1/36 & 32 \\
	& & & & ZYYZYZY & 1/36 & 32 \\
	& & & & ZYYYZYX & 1/36 & 32 \\
	& & & & XZYXXZY & 1/36 & 32 \\
	& & & & XXXYYYY & 1/18 & 32 \\
	& & & & XYYYZYZ & 1/36 & 32 \\
	& & & & YZYYYZX & 1/18 & 32 \\
	& & & & YXYXZYY & 1/36 & 32 \\
	& & & & YYZZYYZ & 1/12 & 32 \\
	& & & & YYXYXXY & 1/18 & 32 \\
	\cline{5-7}
	
	\multirow{17}{*}{44} & \multirow{17}{*}{\includegraphics[width=50mm]{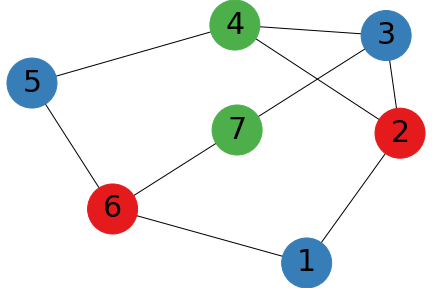}} & \multirow{17}{*}{103} & \multirow{17}{*}{17}
	& ZZYYYZX & 23/123 & 16 \\
	& & & & ZXZXZXZ & 3/41 & 16 \\
	& & & & XZXZZXZ & 11/123 & 16 \\
	& & & & XXZZZXZ & 3/41 & 16 \\
	& & & & XXZZXZZ & 11/123 & 16 \\
	& & & & YXZYZXZ & 1/123 & 16 \\
	& & & & YYXXXYX & 17/123 & 16 \\
	& & & & ZZYYZYY & 2/41 & 32 \\
	& & & & ZYYZXXY & 1/41 & 32 \\
	& & & & XYZYYYY & 10/123 & 32 \\
	& & & & YZZZXZY & 1/123 & 32 \\
	& & & & YXYZYYY & 1/41 & 32 \\
	& & & & YXYXXYY & 5/123 & 32 \\
	& & & & YYXXZZY & 5/123 & 32 \\
	& & & & YYXXYXY & 5/123 & 32 \\
	& & & & YYYYXZX & 1/123 & 32 \\
	& & & & YXXZXXY & 1/41 & 64 \\
	\cline{5-7}
	
	\multirow{19}{*}{45} & \multirow{19}{*}{\includegraphics[width=50mm]{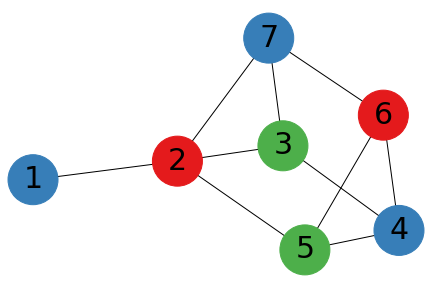}} & \multirow{19}{*}{105} & \multirow{19}{*}{19}
	& XZZZXZX & 23/375 & 16 \\
	& & & & XZZYYZX & 1/25 & 16 \\
	& & & & XZZYYXY & 18/871 & 16 \\
	& & & & XZXXZXX & 67/600 & 16 \\
	& & & & XZYZXZY & 58/911 & 16 \\
	& & & & XZYXYYZ & 9/250 & 16 \\
	& & & & ZXZZZXZ & 47/493 & 32 \\
	& & & & ZXZZYXY & 7/250 & 32 \\
	& & & & ZXXXXZZ & 13/300 & 32 \\
	& & & & ZYZXYXY & 3/100 & 32 \\
	& & & & ZYZYXYZ & 29/500 & 32 \\
	& & & & ZYYZZZY & 19/1000 & 32 \\
	& & & & YXXZYYX & 33/500 & 32 \\
	& & & & YXXYZZY & 53/1000 & 32 \\
	& & & & YXYYXXZ & 47/986 & 32 \\
	& & & & YYXXXYY & 48/809 & 32 \\
	& & & & YYYXYZZ & 53/1000 & 32 \\
	& & & & YYYYZYX & 42/773 & 32 \\
	& & & & ZYYYYYY & 25/419 & 64 \\
	\hhline{=======}
\end{longtable*}

\section{Table of protocols based on $X$ and $Z$ measurements}

\begin{longtable*}{ccccccccc} 
	\caption{\label{tab:minimal-settings}Optimal verification protocols 
based on $X$ and $Z$ measurements and protocols with the minimum number of settings for the same graph states as shown in Table~\ref{tab:optimal} (cf.~\cite{hein2004}). For each graph, the optimal protocol can achieve the spectral gap $\nu(\Omega)=1/2$, which attains the upper bound in  Corollary~\ref{cor:optimalXZ}. The minimum number of settings $\tilde{\chi}(G)$ is equal to  the chromatic number $\chi(G)$ shown in the fourth column (cf. Proposition~\ref{pro:LocalCoverNum});  the corresponding protocol achieves the spectral gap $1/\chi(G)$ (cf. Theorem~\ref{thm:SpectralMinMax}). All protocols with the minimum number of settings shown in the table are coloring protocols \cite{ZhuH2019E} determined by the colorings shown in Table~\ref{tab:optimal}. 
Each protocol is specified by  Pauli measurement settings shown in the seventh column: all settings are measured with the same probability unless noted otherwise (for graphs No. 41 and No. 42). When $\chi(G)=2$, only one protocol for $|G\>$ is shown because the protocol can achieve the maximum spectral gap $\nu(\Omega)=1/2$ and meanwhile requires the minimum number of settings. For completeness, the table also shows the qubit number $n$ in the second column, the rank of each canonical test projector in the eighth column, and the minimum number of generators of the local subgroup in the ninth column (cf.~Sec.~\ref{sec:TestProjGraph}). In addition, $\eta_{XZ}(G)$ denotes the total number of admissible test projectors for the graph state $|G\>$ that are based on $X$ and $Z$ measurements (cf. Sec.~\ref{sec:AdmissibleXZ}).
$\chi^*(G)$ denotes the fractional chromatic number of $G$ \cite{godsil2001}, whose inverse is the maximum spectral gap achievable by the cover protocol proposed in \cite{ZhuH2019E}.
Note that $\eta_{XZ}(G)$, $\chi(G)$, and $\chi^*(G)$ are not LC-invariant. 	
}\\
	\hhline{=========}
	No. & $n$ & $\eta_{XZ}(G)$ & $\chi(G)$ & $\chi^*(G)$ & $\nu(\Omega)$ & Setting & Rank & Number of \\ 
	&  & & & & & & & generators \\ \hline
	\endfirsthead
	
	\\\hline
	No. & $n$ & $\eta_{XZ}(G)$ & $\chi(G)$ & $\chi^*(G)$ & $\nu(\Omega)$ & Setting & Rank & Number of \\ 
	&  & & & & & & & generators \\ \hline
	\endhead

	\multirow{2}{*}{2} & \multirow{2}{*}{3} & \multirow{2}{*}{2} & \multirow{2}{*}{2} & \multirow{2}{*}{2} & \multirow{2}{*}{1/2}
	& XZX & 2 & 2 \\
	& & & & & & ZXZ & 4 & 1 \\
	\hline
	
	\multirow{2}{*}{3} & \multirow{2}{*}{4} & \multirow{2}{*}{2} & \multirow{2}{*}{2} & \multirow{2}{*}{2} & \multirow{2}{*}{1/2}
	& ZXXX & 2 & 3 \\
	& & & & & & XZZZ & 8 & 1 \\
	\hline
	
	\multirow{2}{*}{4} & \multirow{2}{*}{4} & \multirow{2}{*}{3} & \multirow{2}{*}{2} & \multirow{2}{*}{2} & \multirow{2}{*}{1/2}
	& ZXZX & 4 & 2 \\
	& & & & & & XZXZ & 4 & 2 \\
	\hline
	
	\multirow{2}{*}{5} & \multirow{2}{*}{5} & \multirow{2}{*}{2} & \multirow{2}{*}{2} & \multirow{2}{*}{2} & \multirow{2}{*}{1/2}
	& ZXXXX & 2 & 4 \\
	& & & & & & XZZZZ & 16 & 1 \\
	\hline
	
	\multirow{2}{*}{6} & \multirow{2}{*}{5} & \multirow{2}{*}{3} & \multirow{2}{*}{2} & \multirow{2}{*}{2} & \multirow{2}{*}{1/2}
	& XZXZX & 4 & 3 \\
	& & & & & & ZXZXZ & 8 & 2 \\
	\hline
	
	\multirow{2}{*}{7} & \multirow{2}{*}{5} & \multirow{2}{*}{4} & \multirow{2}{*}{2} & \multirow{2}{*}{2} & \multirow{2}{*}{1/2}
	& XZXZX & 4 & 3 \\
	& & & & & & ZXZXZ & 8 & 2 \\
	\hline
	
	\multirow{9}{*}{8} & \multirow{9}{*}{5} & \multirow{9}{*}{6} & \multirow{9}{*}{3} & \multirow{9}{*}{5/2} &  \multirow{3}{*}{1/3}
	& XZXZZ & 8 & 2 \\
	& & & & & & ZXZXZ & 8 & 2 \\
	& & & & & & ZZZZX & 16 & 1 \\
	\cline{6-9}
	\multirow{6}{*}{} & \multirow{6}{*}{} & \multirow{6}{*}{} & \multirow{6}{*}{} & \multirow{6}{*}{} & \multirow{6}{*}{1/2}
	& ZZXZX & 8 & 2 \\ 
	& & & & & & ZXZZX & 8 & 2 \\ 
	& & & & & & ZXZXZ & 8 & 2 \\ 
	& & & & & & XZZXZ & 8 & 2 \\ 
	& & & & & & XZXZZ & 8 & 2 \\ 
	& & & & & & XXXXX & 16 & 1 \\ 
	\hline
	
	\multirow{2}{*}{9} & \multirow{2}{*}{6} & \multirow{2}{*}{2} & \multirow{2}{*}{2} & \multirow{2}{*}{2} & \multirow{2}{*}{1/2}
	& ZXXXXX & 2 & 5 \\
	& & & & & & XZZZZZ & 32 & 1 \\
	\hline
	
	\multirow{2}{*}{10} & \multirow{2}{*}{6} & \multirow{2}{*}{3} & \multirow{2}{*}{2} & \multirow{2}{*}{2} & \multirow{2}{*}{1/2}
	& XZXZXX & 4 & 4 \\
	& & & & & & ZXZXZZ & 16 & 2 \\
	\hline
	
	\multirow{2}{*}{11} & \multirow{2}{*}{6} & \multirow{2}{*}{3} & \multirow{2}{*}{2} & \multirow{2}{*}{2} & \multirow{2}{*}{1/2}
	& ZXZXZX & 8 & 3 \\
	& & & & & & XZXZXZ & 8 & 3 \\
	\hline
	
	\multirow{2}{*}{12} & \multirow{2}{*}{6} & \multirow{2}{*}{4} & \multirow{2}{*}{2} & \multirow{2}{*}{2} & \multirow{2}{*}{1/2}
	& XZXZXX & 4 & 4 \\
	& & & & & & ZXZXZZ & 16 & 2 \\
	\hline
	
	\multirow{2}{*}{13} & \multirow{2}{*}{6} & \multirow{2}{*}{5} & \multirow{2}{*}{2} & \multirow{2}{*}{2} & \multirow{2}{*}{1/2}
	& ZXZXZX & 8 & 3 \\
	& & & & & & XZXZXZ & 8 & 3 \\
	\hline
	
	\multirow{2}{*}{14} & \multirow{2}{*}{6} & \multirow{2}{*}{5} & \multirow{2}{*}{2} & \multirow{2}{*}{2} & \multirow{2}{*}{1/2}
	& ZXZXZX & 8 & 3 \\
	& & & & & & XZXZXZ & 8 & 3 \\
	\hline
	
	\multirow{2}{*}{15} & \multirow{2}{*}{6} & \multirow{2}{*}{5} & \multirow{2}{*}{2} & \multirow{2}{*}{2} & \multirow{2}{*}{1/2}
	& XZXZXX & 4 & 4 \\
	& & & & & & ZXZXZZ & 16 & 2 \\
	\hline
	
	\multirow{7}{*}{16} & \multirow{7}{*}{6} & \multirow{7}{*}{5} & \multirow{7}{*}{3} & \multirow{7}{*}{3} &  \multirow{3}{*}{1/3}
	& XZZZXX & 8 & 3 \\
	& & & & & & ZXZXZZ & 16 & 2 \\
	& & & & & & ZZXZZZ & 32 & 1 \\
	\cline{6-9}
	\multirow{4}{*}{} & \multirow{4}{*}{} & \multirow{4}{*}{} & \multirow{4}{*}{} & \multirow{4}{*}{} & \multirow{4}{*}{1/2}
	& ZZXXXZ & 8 & 3 \\ 
	& & & & & & ZXZXZX & 8 & 3 \\ 
	& & & & & & XZZZXX & 8 & 3 \\ 
	& & & & & & XXXZZZ & 32 & 1 \\ 
	\hline
	
	\multirow{9}{*}{17} & \multirow{9}{*}{6} & \multirow{9}{*}{6} & \multirow{9}{*}{3} & \multirow{9}{*}{5/2} &  \multirow{3}{*}{1/3}
	& ZXZZXZ & 16 & 2 \\
	& & & & & & XZXZZX & 8 & 3 \\
	& & & & & & ZZZXZZ & 32 & 1 \\
	\cline{6-9}
	\multirow{6}{*}{} & \multirow{6}{*}{} & \multirow{6}{*}{} & \multirow{6}{*}{} & \multirow{6}{*}{} & \multirow{6}{*}{1/2}
	& ZXZXZX & 8 & 3 \\ 
	& & & & & & XZZXZX & 8 & 3 \\ 
	& & & & & & XZXZZX & 8 & 3 \\ 
	& & & & & & ZZXZXZ & 16 & 2 \\ 
	& & & & & & ZXZZXZ & 16 & 2 \\ 
	& & & & & & XXXXXZ & 32 & 1 \\ 
	\hline
	
	\multirow{2}{*}{18} & \multirow{2}{*}{6} & \multirow{2}{*}{6} & \multirow{2}{*}{2} & \multirow{2}{*}{2} & \multirow{2}{*}{1/2}
	& ZXZXZX & 8 & 3 \\
	& & & & & & XZXZXZ & 8 & 3 \\
	\hline
	
	\multirow{9}{*}{19} & \multirow{9}{*}{6} & \multirow{9}{*}{12} & \multirow{9}{*}{3} & \multirow{9}{*}{3} &  \multirow{3}{*}{1/3}
	& XZZZZX & 16 & 2 \\
	& & & & & & ZXZXZZ & 16 & 2 \\
	& & & & & & ZZXZXZ & 16 & 2 \\
	\cline{6-9}
	\multirow{6}{*}{} & \multirow{6}{*}{} & \multirow{6}{*}{} & \multirow{6}{*}{} & \multirow{6}{*}{} & \multirow{6}{*}{1/2}
	& ZZXZXZ & 16 & 2 \\ 
	& & & & & & ZZXXZZ & 16 & 2 \\ 
	& & & & & & ZXZZZX & 16 & 2 \\ 
	& & & & & & XZZZZX & 16 & 2 \\ 
	& & & & & & XXZXXZ & 16 & 2 \\ 
	& & & & & & XXXXXX & 16 & 2 \\ 
	\hline
	
	\multirow{2}{*}{20} & \multirow{2}{*}{7} & \multirow{2}{*}{2} & \multirow{2}{*}{2} & \multirow{2}{*}{2} & \multirow{2}{*}{1/2}
	& ZXXXXXX & 2 & 6 \\
	& & & & & & XZZZZZZ & 64 & 1 \\
	\hline
	
	\multirow{2}{*}{21} & \multirow{2}{*}{7} & \multirow{2}{*}{3} & \multirow{2}{*}{2} & \multirow{2}{*}{2} & \multirow{2}{*}{1/2}
	& XXXXZXZ & 4 & 5 \\
	& & & & & & ZZZZXZX & 32 & 2 \\
	\hline
	
	\multirow{2}{*}{22} & \multirow{2}{*}{7} & \multirow{2}{*}{3} & \multirow{2}{*}{2} & \multirow{2}{*}{2} & \multirow{2}{*}{1/2}
	& XXXZZXZ & 8 & 4 \\
	& & & & & & ZZZXXZX & 16 & 3 \\
	\hline
	
	\multirow{2}{*}{23} & \multirow{2}{*}{7} & \multirow{2}{*}{4} & \multirow{2}{*}{2} & \multirow{2}{*}{2} & \multirow{2}{*}{1/2}
	& XXXXZXZ & 4 & 5 \\
	& & & & & & ZZZZXZX & 32 & 2 \\
	\hline
	
	\multirow{2}{*}{24} & \multirow{2}{*}{7} & \multirow{2}{*}{4} & \multirow{2}{*}{2} & \multirow{2}{*}{2} & \multirow{2}{*}{1/2}
	& XXXXZXZ & 4 & 5 \\
	& & & & & & ZZZZXZX & 32 & 2 \\
	\hline
	
	\multirow{2}{*}{25} & \multirow{2}{*}{7} & \multirow{2}{*}{5} & \multirow{2}{*}{2} & \multirow{2}{*}{2} & \multirow{2}{*}{1/2}
	& XZXXZXZ & 8 & 4 \\
	& & & & & & ZXZZXZX & 16 & 3 \\
	\hline
	
	\multirow{2}{*}{26} & \multirow{2}{*}{7} & \multirow{2}{*}{5} & \multirow{2}{*}{2} & \multirow{2}{*}{2} & \multirow{2}{*}{1/2}
	& XXZXZXZ & 8 & 4 \\
	& & & & & & ZZXZXZX & 16 & 3 \\
	\hline
	
	\multirow{2}{*}{27} & \multirow{2}{*}{7} & \multirow{2}{*}{5} & \multirow{2}{*}{2} & \multirow{2}{*}{2} & \multirow{2}{*}{1/2}
	& XZXZXZX & 8 & 4 \\
	& & & & & & ZXZXZXZ & 16 & 3 \\
	\hline
	
	\multirow{2}{*}{28} & \multirow{2}{*}{7} & \multirow{2}{*}{6} & \multirow{2}{*}{2} & \multirow{2}{*}{2} & \multirow{2}{*}{1/2}
	& ZXZXXZX & 8 & 4 \\
	& & & & & & XZXZZXZ & 16 & 3 \\
	\hline
	
	\multirow{2}{*}{29} & \multirow{2}{*}{7} & \multirow{2}{*}{8} & \multirow{2}{*}{2} & \multirow{2}{*}{2} & \multirow{2}{*}{1/2}
	& XZXZXZX & 8 & 4 \\
	& & & & & & ZXZXZXZ & 16 & 3 \\
	\hline
	
	\multirow{2}{*}{30} & \multirow{2}{*}{7} & \multirow{2}{*}{7} & \multirow{2}{*}{2} & \multirow{2}{*}{2} & \multirow{2}{*}{1/2}
	& XZXZXZX & 8 & 4 \\
	& & & & & & ZXZXZXZ & 16 & 3 \\
	\hline
	
	\multirow{2}{*}{31} & \multirow{2}{*}{7} & \multirow{2}{*}{5} & \multirow{2}{*}{2} & \multirow{2}{*}{2} & \multirow{2}{*}{1/2}
	& ZXZXXXX & 4 & 5 \\
	& & & & & & XZXZZZZ & 32 & 2 \\
	\hline
	
	\multirow{7}{*}{32} & \multirow{7}{*}{7} & \multirow{7}{*}{5} & \multirow{7}{*}{3} & \multirow{7}{*}{3} &  \multirow{3}{*}{1/3}
	& ZXZZXZX & 16 & 3 \\
	& & & & & & XZXZZXZ & 16 & 3 \\
	& & & & & & ZZZXZZZ & 64 & 1 \\
	\cline{6-9}
	\multirow{4}{*}{} & \multirow{4}{*}{} & \multirow{4}{*}{} & \multirow{4}{*}{} & \multirow{4}{*}{} & \multirow{4}{*}{1/2}
	& XXZXZXZ & 8 & 4 \\ 
	& & & & & & XXXZXZZ & 8 & 4 \\ 
	& & & & & & ZZXXZZX & 16 & 3 \\ 
	& & & & & & ZZZZXXX & 64 & 1 \\ 
	\hline
	
	\multirow{9}{*}{33} & \multirow{9}{*}{7} & \multirow{9}{*}{6} & \multirow{9}{*}{3} & \multirow{9}{*}{5/2} &  \multirow{3}{*}{1/3}
	& ZXZZXZZ & 32 & 2 \\
	& & & & & & XZXXZXZ & 8 & 4 \\
	& & & & & & ZZZZZZX & 64 & 1 \\
	\cline{6-9}
	\multirow{6}{*}{} & \multirow{6}{*}{} & \multirow{6}{*}{} & \multirow{6}{*}{} & \multirow{6}{*}{} & \multirow{6}{*}{1/2}
	& XXZZXZX & 8 & 4 \\ 
	& & & & & & XXZXZZX & 8 & 4 \\ 
	& & & & & & XXZXZXZ & 8 & 4 \\ 
	& & & & & & ZZXZZXZ & 32 & 2 \\ 
	& & & & & & ZZXZXZZ & 32 & 2 \\ 
	& & & & & & ZZXXXXX & 64 & 1 \\ 
	\hline
	
	\multirow{2}{*}{34} & \multirow{2}{*}{7} & \multirow{2}{*}{6} & \multirow{2}{*}{2} & \multirow{2}{*}{2} & \multirow{2}{*}{1/2}
	& XZXZXZX & 8 & 4 \\
	& & & & & & ZXZXZXZ & 16 & 3 \\
	\hline
	
	\multirow{9}{*}{35} & \multirow{9}{*}{7} & \multirow{9}{*}{6} & \multirow{9}{*}{3} & \multirow{9}{*}{5/2} &  \multirow{3}{*}{1/3}
	& ZXZXZZZ & 32 & 2 \\
	& & & & & & XZXZXZX & 8 & 4 \\
	& & & & & & ZZZZZXZ & 64 & 1 \\
	\cline{6-9}
	\multirow{6}{*}{} & \multirow{6}{*}{} & \multirow{6}{*}{} & \multirow{6}{*}{} & \multirow{6}{*}{} & \multirow{6}{*}{1/2}
	& XXZZXZX & 8 & 4 \\ 
	& & & & & & XXZXZZX & 8 & 4 \\ 
	& & & & & & ZXZXZXZ & 16 & 3 \\ 
	& & & & & & XZXZXZZ & 16 & 3 \\ 
	& & & & & & ZZXZZXZ & 32 & 2 \\ 
	& & & & & & ZZXXXXX & 64 & 1 \\ 
	\hline
	
	\multirow{7}{*}{36} & \multirow{7}{*}{7} & \multirow{7}{*}{7} & \multirow{7}{*}{3} & \multirow{7}{*}{3} &  \multirow{3}{*}{1/3}
	& XZXZZXX & 8 & 4 \\
	& & & & & & ZXZXZZZ & 32 & 2 \\
	& & & & & & ZZZZXZZ & 64 & 1 \\
	\cline{6-9}
	\multirow{4}{*}{} & \multirow{4}{*}{} & \multirow{4}{*}{} & \multirow{4}{*}{} & \multirow{4}{*}{} & \multirow{4}{*}{1/2}
	& XZXZZXX & 8 & 4 \\ 
	& & & & & & ZXZZXZX & 16 & 3 \\ 
	& & & & & & ZXZXZXZ & 16 & 3 \\ 
	& & & & & & XZXXXZZ & 32 & 2 \\ 
	\hline
	
	\pagebreak
	
	\multirow{9}{*}{37} & \multirow{9}{*}{7} & \multirow{9}{*}{7} & \multirow{9}{*}{3} & \multirow{9}{*}{5/2} &  \multirow{3}{*}{1/3}
	& ZXZXZZX & 16 & 3 \\
	& & & & & & XZXZZXZ & 16 & 3 \\
	& & & & & & ZZZZXZZ & 64 & 1 \\
	\cline{6-9}
	\multirow{6}{*}{} & \multirow{6}{*}{} & \multirow{6}{*}{} & \multirow{6}{*}{} & \multirow{6}{*}{} & \multirow{6}{*}{1/2}
	& XZXZXZX & 8 & 4 \\ 
	& & & & & & ZXZZXZX & 16 & 3 \\ 
	& & & & & & ZXZXZZX & 16 & 3 \\ 
	& & & & & & XZZXZXZ & 16 & 3 \\ 
	& & & & & & XZXZZXZ & 16 & 3 \\ 
	& & & & & & ZXXXXXZ & 64 & 1 \\ 
	\hline
	
	\multirow{2}{*}{38} & \multirow{2}{*}{7} & \multirow{2}{*}{7} & \multirow{2}{*}{2} & \multirow{2}{*}{2} & \multirow{2}{*}{1/2}
	& XZXZXZX & 8 & 4 \\
	& & & & & & ZXZXZXZ & 16 & 3 \\
	\hline
	
	\multirow{9}{*}{39} & \multirow{9}{*}{7} & \multirow{9}{*}{9} & \multirow{9}{*}{3} & \multirow{9}{*}{5/2} &  \multirow{3}{*}{1/3}
	& ZXZZXZX & 16 & 3 \\
	& & & & & & XZXZZXZ & 16 & 3 \\
	& & & & & & ZZZXZZZ & 64 & 1 \\
	\cline{6-9}
	\multirow{6}{*}{} & \multirow{6}{*}{} & \multirow{6}{*}{} & \multirow{6}{*}{} & \multirow{6}{*}{} & \multirow{6}{*}{1/2}
	& ZZXZXZX & 16 & 3 \\ 
	& & & & & & ZXZZXZX & 16 & 3 \\ 
	& & & & & & ZXZXZXZ & 16 & 3 \\ 
	& & & & & & XZZXZXZ & 16 & 3 \\ 
	& & & & & & XZXZZXZ & 16 & 3 \\ 
	& & & & & & XXXXXZX & 32 & 2 \\ 
	\hline
	
	\multirow{11}{*}{40} & \multirow{11}{*}{7} & \multirow{11}{*}{8} & \multirow{11}{*}{3} & \multirow{11}{*}{7/3} &  \multirow{3}{*}{1/3}
	& XZXZXZZ & 16 & 3 \\
	& & & & & & ZXZXZXZ & 16 & 3 \\
	& & & & & & ZZZZZZX & 64 & 1 \\
	\cline{6-9}
	\multirow{8}{*}{} & \multirow{8}{*}{} & \multirow{8}{*}{} & \multirow{8}{*}{} & \multirow{8}{*}{} & \multirow{8}{*}{1/2}
	& ZZXZXZX & 16 & 3 \\ 
	& & & & & & ZXZZXZX & 16 & 3 \\ 
	& & & & & & ZXZXZZX & 16 & 3 \\ 
	& & & & & & ZXZXZXZ & 16 & 3 \\ 
	& & & & & & XZZXZXZ & 16 & 3 \\ 
	& & & & & & XZXZZXZ & 16 & 3 \\ 
	& & & & & & XZXZXZZ & 16 & 3 \\ 
	& & & & & & XXXXXXX & 64 & 1 \\ 
	\hline
	
	\multirow{10}{*}{41} & \multirow{10}{*}{7} & \multirow{10}{*}{10} & \multirow{10}{*}{3} & \multirow{10}{*}{3} &  \multirow{3}{*}{1/3}
	& XZXZZZX & 16 & 3 \\
	& & & & & & ZXZZXZZ & 32 & 2 \\
	& & & & & & ZZZXZXZ & 32 & 2 \\
	\cline{6-9}
	\multirow{7}{*}{} & \multirow{7}{*}{} & \multirow{7}{*}{} & \multirow{7}{*}{} & \multirow{7}{*}{} & \multirow{7}{*}{1/2}
	& ZZXZXZX & 16 & 3 \\ 
	& & & & & & ZXZZXZX & 16 & 3 \\ 
	& & & & & & ZXZXZXZ (1/4) & 16 & 3 \\ 
	& & & & & & XZZXZZX & 16 & 3 \\ 
	& & & & & & XZXZZZX & 16 & 3 \\ 
	& & & & & & XZXZXXZ & 32 & 2 \\ 
	& & & & & & XXXXXXZ & 32 & 2 \\ 
	\hline
	
	\multirow{9}{*}{42} & \multirow{9}{*}{7} & \multirow{9}{*}{10} & \multirow{9}{*}{3} & \multirow{9}{*}{5/2} &  \multirow{3}{*}{1/3}
	& ZXZZXZZ & 32 & 2 \\
	& & & & & & XZXZZXZ & 16 & 3 \\
	& & & & & & ZZZXZZX & 32 & 2 \\
	\cline{6-9}
	\multirow{6}{*}{} & \multirow{6}{*}{} & \multirow{6}{*}{} & \multirow{6}{*}{} & \multirow{6}{*}{} & \multirow{6}{*}{1/2}
	& ZZXZXZX (1/4) & 16 & 3 \\ 
	& & & & & & ZXZXZZX & 16 & 3 \\ 
	& & & & & & XZZXZXZ (1/4) & 16 & 3 \\ 
	& & & & & & XXZZXZZ & 16 & 3 \\ 
	& & & & & & ZXXXXXZ & 64 & 1 \\ 
	& & & & & & XXXZZXX & 64 & 1 \\ 
	\hline
	
	\multirow{2}{*}{43} & \multirow{2}{*}{7} & \multirow{2}{*}{9} & \multirow{2}{*}{2} & \multirow{2}{*}{2} & \multirow{2}{*}{1/2}
	& ZXZXZXX & 8 & 4 \\
	& & & & & & XZXZXZZ & 16 & 3 \\
	\hline
	
	\multirow{7}{*}{44} & \multirow{7}{*}{7} & \multirow{7}{*}{11} & \multirow{7}{*}{3} & \multirow{7}{*}{3} &  \multirow{3}{*}{1/3}
	& ZXZZZXZ & 32 & 2 \\
	& & & & & & XZXZXZZ & 16 & 3 \\
	& & & & & & ZZZXZZX & 32 & 2 \\
	\cline{6-9}
	\multirow{4}{*}{} & \multirow{4}{*}{} & \multirow{4}{*}{} & \multirow{4}{*}{} & \multirow{4}{*}{} & \multirow{4}{*}{1/2}
	& ZXZXZXZ & 16 & 3 \\ 
	& & & & & & XXZZXZZ & 16 & 3 \\ 
	& & & & & & ZZXZZZX & 32 & 2 \\ 
	& & & & & & XZXXXXX & 32 & 2 \\ 
	\hline
	
	\pagebreak
	
	\multirow{7}{*}{45} & \multirow{7}{*}{7} & \multirow{7}{*}{12} & \multirow{7}{*}{3} & \multirow{7}{*}{3} &  \multirow{3}{*}{1/3}
	& ZXZZZXZ & 32 & 2 \\
	& & & & & & XZZXZZX & 16 & 3 \\
	& & & & & & ZZXZXZZ & 32 & 2 \\
	\cline{6-9}
	\multirow{4}{*}{} & \multirow{4}{*}{} & \multirow{4}{*}{} & \multirow{4}{*}{} & \multirow{4}{*}{} & \multirow{4}{*}{1/2}
	& XZZZXZX & 16 & 3 \\ 
	& & & & & & XZXXZXX & 16 & 3 \\ 
	& & & & & & ZXZZZXZ & 32 & 2 \\ 
	& & & & & & ZXXXXZZ & 32 & 2 \\ 
	
	\hhline{=========}
\end{longtable*}
\end{appendices}


\bibliographystyle{plainnat-no-url}

\end{document}